\documentclass[11pt]{article}
\usepackage{microtype}

\usepackage{amsmath,amssymb}
\usepackage{amsfonts,bm,bbm} 
\usepackage{amssymb,amsthm,enumitem}
\usepackage{tikz}
\usetikzlibrary{decorations.pathmorphing,positioning,patterns,snakes}
\usepackage{multirow}
\usepackage{xcolor,graphicx,float,booktabs}
\usepackage{tabularx}
\newcolumntype{L}[1]{>{\raggedright\arraybackslash}m{#1}}
\newcolumntype{C}[1]{>{\centering\arraybackslash}m{#1}}
\newcolumntype{R}[1]{>{\raggedleft\arraybackslash}m{#1}}

\usepackage{epstopdf}

\usepackage{verbatim} 
\allowdisplaybreaks[4]
\usepackage{authblk}

\usepackage[round]{natbib}
\usepackage[
margin=1in,
includefoot,
footskip=30pt,
]{geometry}

\newcommand{\cG}{\mathcal{G}}
\newcommand{\cO}{\mathcal{O}}
\newcommand{\cU}{\mathcal{U}}
\newcommand{\cV}{\mathcal{V}}
\newcommand{\cL}{\mathcal{L}}
\newcommand{\cN}{\mathcal{N}}
\newcommand{\E}{\mathbb{E}}
\newcommand{\R}{\mathbb{R}}
\newcommand{\Var}{\mathrm{Var}}
\newcommand{\MSE}{\mathrm{MSE}}
\newcommand{\Bias}{\mathrm{Bias}}
\renewcommand{\P}{\mathbb{P}}

\renewcommand{\d}{\mathrm{d}}
\newcommand{\1}{\mathbbm{1}}

\newcommand{\bS}{\bm{S}}

\newcommand{\avgmj}{\frac{1}{m}\sum_{j=1}^{m}}
\newcommand{\avgni}{\frac{1}{n}\sum_{i=1}^{n}}

\newcommand{\Li}{L_{i}}
\newcommand{\Lmi}{L_{m,i}}

\newcommand{\hatH}{\widehat{H}}
\newcommand{\hatHij}{\hatH_{ij}}
\newcommand{\hatR}{\widehat{R}}

\newcommand{\ftilde}{\widetilde{f}}

\newcommand{\geps}{g_\epsilon}
\newcommand{\gepsm}{g_{\epsilon_m}}

\newcommand{\condist}{\stackrel{d}{\rightarrow}}
\newcommand{\conprob}{\stackrel{p}{\rightarrow}}
\newcommand{\conlone}{\stackrel{\cL^1}{\rightarrow}}
\newtheorem{theorem}{Theorem}[section]
\newtheorem{proposition}[theorem]{Proposition}
\newtheorem{lemma}[theorem]{Lemma}
\newtheorem{corollary}[theorem]{Corollary}
\newtheorem{definition}[theorem]{Definition}

\newtheorem{assumption}{Assumption}[section]
\numberwithin{equation}{section}

\usepackage{draftwatermark}
\SetWatermarkText{Preprint Version}
\SetWatermarkScale{3}

\title{Sample Recycling for Nested Simulation\\
	with Application in Portfolio Risk Measurement}

\author[1]{Kun Zhang\footnote{kunzhang@ruc.edu.cn}}
\author[2]{Ben M. Feng\footnote{ben.feng@uwaterloo.ca}}
\author[3]{Guangwu Liu\footnote{msgw.liu@cityu.edu.hk}}
\author[3]{Shiyu Wang\footnote{shiyuwang7-c@my.cityu.edu.hk}}
\affil[1]{Institute of Statistics and Big Data\protect\\Renmin University of China\protect\\Beijing, China}
\affil[2]{Department of Statistics and Actuarial Science\protect\\University of Waterloo\protect\\Waterloo, ON, Canada}
\affil[3]{Department of Management Sciences\protect\\City University of Hong Kong\protect\\Tat Chee Avenue, Kowloon, Hong Kong, China}

\begin{document}
\normalsize

\maketitle

\begin{abstract}
Nested simulation is a natural approach to tackle nested estimation problems in operations research and financial engineering.
The outer-level simulation generates outer scenarios and the inner-level simulations are ran in each outer scenario to estimate the corresponding conditional expectation.
The resulting sample of conditional expectations is then used to estimate different risk measures of interest.
Despite its flexibility, nested simulation is notorious for its heavy computational burden.
We introduce a novel simulation procedure that reuses inner simulation outputs to improve the efficiency and accuracy in solving nested estimation problems. We analyze the convergence rates of the bias, variance, and MSE of the resulting estimator. In addition, central limit theorems and variance estimators are presented, which lead to asymptotically valid confidence intervals for the nested risk measure of interest.
We conduct numerical studies on two financial risk measurement problems.
Our numerical studies show consistent results with the asymptotic analysis and show that the proposed approach outperforms the standard nested simulation and a state-of-art regression approach for nested estimation problems.
\\
\emph{Key words}: nested simulation, risk management, likelihood ratio method, central limit theorem, confidence interval
\end{abstract}

\section{Introduction}\label{sec:intro}

\textit{Nested estimation} is the problem of estimating a functional of a conditional expectation.
In this study, we propose and analyze an efficient simulation method for a class of nested estimation problems.
Specifically, the quantity to be estimated is
\begin{equation}\label{eq:rho}
\rho = \rho(\E\left[H(X,Y)|X\right]) = \E\left[g(\E\left[H(X,Y)|X\right])\right],
\end{equation}
where $X$ and $Y$ are both random vectors of fixed dimensions, $H(\cdot,\cdot)$ is a multi-variate mapping, and $g(\cdot)$ is a real-value function.
In a nested simulation, we call $X$ the outer scenario, $Y$ the inner-level random variable, $H(\cdot,\cdot)$ the inner simulation model, and $g(\cdot)$ the risk function.
Nested estimation~\citep{hong2017kernel} has important applications in operations research, such as risk measurement~\citep{lee1998monte,gordy2010nested} and input uncertainty quantification~\citep{cheng1997sensitivity,barton2012,zhu2020risk}.

\textit{Nested simulation}~\citep{gordy2010nested,broadie2011efficient}, which is also known as two-level and stochastic-on-stochastic simulation, is a natural solution for the above nested estimation problems:
Consider measuring some risk measures of a portfolio of financial instruments whose values are affected by different risk factors such as equity returns, interest rates, mortality rates, etc.
In this case, $X$ represents the evolution of the underlying risk factors up to a future time (i.e., the \textit{risk horizon}), say in one month, when risk measurement is required.
The outer-level simulation generates $n$ realizations of $X$, which are called the \textit{scenarios}.
Given a scenario $X$, $Y|X$ denotes the risk factors' evolution between the risk horizon and the portfolio's maturity, say in one year, $H(X,Y)$ denotes the (discounted) loss of the portfolio at maturity, and $\E[H(X,Y)|X]$ denotes the portfolio's mark-to-market loss at the risk horizon.
For each scenario $X$, an inner simulation is performed where $m'$ sample paths of $Y|X$ are generated.
The discounted losses $H(X,Y)$ can then be calculated, whose sample average can be used to estimate the loss of scenario $X$, i.e., $\E[H(X,Y)|X]$.
As $X$ is stochastic, so is $\E[H(X,Y)|X]$.

Depending on the risk function $g(\cdot)$, the nested estimation problem~\eqref{eq:rho} can be used to estimate popular risk measures like the exceedance probability, conditional value-at-risk (CVaR), and squared tracking error of $\E[H(X,Y)|X]$.
For example, for an indicator function $g(x)=\1\{x\geq x_0\}$ and a quadratic function $g(x)=(x-x_0)^2$ for some threshold $x_0$, $\rho(\E\left[H(X,Y)|X\right])$ is the exceedance probability beyond $x_0$ and the squared tracking error, respectively.
For a hockey-stick function $g(x)= x_0 + \frac{1}{1-\alpha}\max\{x-x_0, 0\}$ where $x_0$ is the $\alpha$-Value-at-Risk (VaR) of $\E\left[H(X,Y)|X\right]$, then $\rho(\E\left[H(X,Y)|X\right])$ is the $\alpha$-CVaR.
Interested readers can refer to~\cite{broadie2015risk} and~\cite{hong2017kernel} on nested estimation for these risk measures.

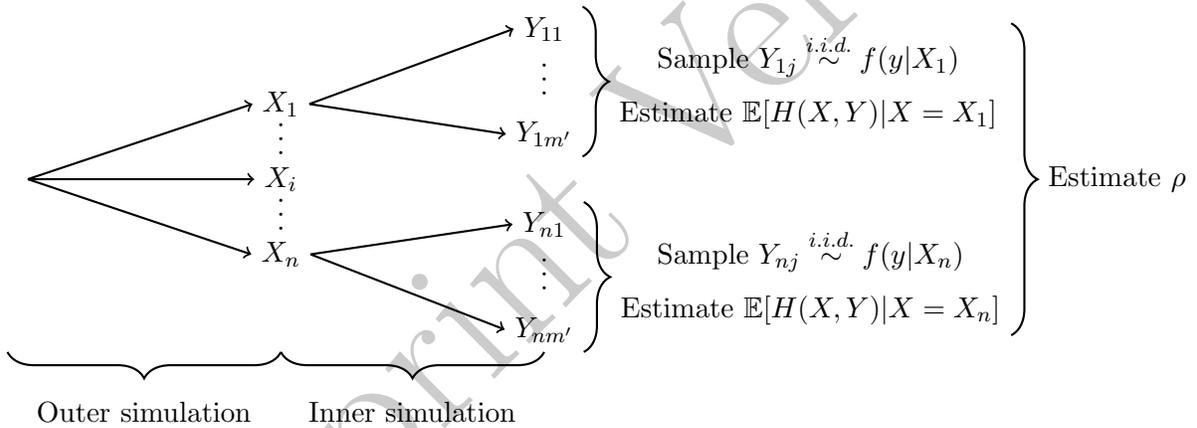
\begin{figure}[h!]
	\newcommand{\frontheightsind}{1/1, 0/i, -1/n}
	\newcommand{\backheightsind}{2/{11}, 0.6/{1m'}, -0.6/{n1}, -2/{nm'}}
	\def\xa{0}
	\def\xb{3.5}
	\def\xc{7}
	\begin{center}
		\begin{tikzpicture}
			\draw (\xa, 0) node(s0) {};
			\foreach \y/\ind in \frontheightsind{
				\draw (\xb,\y) node(st\ind) {$X_{\ind}$};
				\draw[thick,->] (s0.east) -- (st\ind.west);
			}
			
			\foreach \y/\ind in \backheightsind{
				\draw (\xc,\y) node(sp\ind) {$Y_{\ind}$};
			}
			
			\foreach \indori/\inddes in {st1/sti, sti/stn, sp11/sp1m', spn1/spnm'}{
				\path (\indori) -- (\inddes) node [font=\large, midway, sloped] {$\dots$};	
			}
			
			\foreach \indori/\inddes in {st1/sp11, st1/sp1m', stn/spn1, stn/spnm'}{
				\draw[thick,->] (\indori.east) -- (\inddes.west);
			}
			
			\draw [thick, decorate,decoration={brace,amplitude=10pt}](sp11.north -| sp1m'.east) -- (sp1m'.south east) node[black,right, midway,xshift=0.35cm, align=center] (inner1) {Sample $Y_{1j}\stackrel{i.i.d.}{\sim} f(y|X_1)$\\Estimate $\E[H(X,Y)|X=X_1]$};
			
			\draw [thick, decorate,decoration={brace,amplitude=10pt}](spn1.north -| spnm'.east) -- (spnm'.south -| spnm'.east) node[black,right, midway,xshift=0.35cm, align=center] (innerm) {Sample $Y_{nj}\stackrel{i.i.d.}{\sim} f(y|X_n)$\\Estimate $\E[H(X,Y)|X=X_n]$};
			
			\draw [thick, decorate,decoration={brace,amplitude=10pt}](inner1.north -| innerm.east) -- (innerm.south east) node[black,right, midway,xshift=0.35cm, align=center] {Estimate $\rho$};
			
			\draw [thick, decorate,decoration={brace,amplitude=10pt,mirror}](spnm'.south -| s0.west) -- (spnm'.south -| stn.center) node[black,midway,yshift=-0.8cm] {Outer simulation};
			\draw [thick, decorate,decoration={brace,amplitude=10pt,mirror}](spnm'.south -| stn.center) -- (spnm'.south) node[black,midway,yshift=-0.8cm] {Inner simulation};
		\end{tikzpicture}
	\end{center}
	\caption{Schematic illustration of standard two-level nested simulation. The outer stage generates $n$ scenarios $X_1,\ldots,X_n$. Conditional on $X_i$, $m'$ inner replications $Y_{i1},\ldots,Y_{im'}$ are generated.}
	\label{fig:NestedSim}
\end{figure}

Figure~\ref{fig:NestedSim} is a schematic illustration of a standard nested simulation procedure.
The standard nested simulation procedure estimates $\E[H|X]$ for each scenario $X$ by considering the inner replications of that scenario only.
This exclusivity leads to the nested structure, which then requires $\Gamma = m'n$ inner replications in total, e.g., $Y_{ij}$ and $H(X_i,Y_{ij})$ for $i=1,\ldots,n$ and $j=1,\ldots,m'$; $\Gamma$ is called the simulation budget.
In theory, the risk estimator in a nested simulation procedure converges to the true risk measure as the numbers of outer and inner simulations grow.
However, depending on the complexity of the risk factor models and the derivative payoffs, every inner replication can be quite time-consuming to compute.
So, in practice, the simulation budget $\Gamma$ can be an excessive computational burden and unbearably large computations may be required to achieve satisfactory accuracy.


Alleviating the computational burden, by different means and in different ways, has attracted much research attention in the simulation literature.
Firstly, some studies focus on intelligent allocations of a fixed simulation budget $\Gamma$ so that the resulting risk measure $\rho$ is accurately estimated.
\cite{lee1998monte},~\cite{lee2003computing}, and~\cite{gordy2010nested} analyze the nested simulation estimator and demonstrate that, under some assumptions, the asymptotic mean squared error (MSE) of the standard nested risk estimator diminishes at an optimal rate of $\Gamma^{-2/3}$;~\cite{gordy2010nested}  shows that this optimal convergence rate is achieved when $m'=\cO(\Gamma^{1/3})$ and $n=\cO(\Gamma^{2/3})$ as $\Gamma\rightarrow \infty$.
~\cite{broadie2011efficient} proposes a sequential allocation scheme where different outer scenarios have different number of inner replications when estimating the probabilities of large portfolio losses.
The MSE of the resulting risk estimator is shown to have a rate of convergence of $\Gamma^{-4/5+\varepsilon}$ for any $\varepsilon>0$.
~\cite{liu2010efficient} and~\cite{lan2010confidence} use ranking-and-selection techniques to adaptively allocate the simulation budget to estimate CVaR and its confidence interval, respectively.

A second line of research aims to reduce the standard nested simulation's computational burden by estimating $\E[H|X]$ via regression or metamodeling techniques.
For example, least-square Monte Carlo (LSMC) \citep{longstaff2001valuing, tsitsiklis2001regression} is a quintessential parametric approach for pricing American options, where a regression model is used to approximate the conditional expectation $\E[H|X]$.
See also~\cite{carriere1996valuation} for a general discussion of nonparametric regression techniques in Monte Carlo simulation.
\cite{broadie2015risk} applies this LSMC approach in nested estimation of financial risk and shows that the MSE of the resulting risk estimator converges at the order of $\Gamma^{-1+\delta}$ for any $\delta>0$.
Despite fast convergence rate, the MSE generally converges to a nonzero asymptotic squared bias that depends on the selection of basis functions.
~\cite{liu2010stochastic} considers a metamodeling approach that estimates $\E[H|X]$ by a stochastic kriging model~\citep{ankenman2010stochastic}.
Besides selecting appropriate basis functions and covariance functions, the implementation of stochastic kriging is not trivial and may be prone to numerical instability~\citep{staum2009better}.
\cite{hong2017kernel} proposes a kernel smoothing approach, which estimates $\E[H|X]$ by the well-known Nadaraya-Watson kernel estimator~\citep{nadaraya1964estimating,watson1964smooth}.
The MSE of the resulting risk estimator achieves a convergence rate of $\Gamma^{-\min\{1,4/(d+2)\}}$, where $d$ is the problem dimension.
These approaches use simulation outputs from different scenarios, sometimes from a pilot experiment, to calibrate the regression model or metamodel that approximates or predicts $\E[H|X]$ for different scenarios.
While the pooling of simulation outputs improves simulation efficiency, these approaches suffer from modeling errors that depend on selection of basis functions, covariance functions, or kernel bandwidth.
As a result, these approaches lead to biased estimators; sometimes this bias vanishes asymptotically, sometimes the bias persists.

In this article we study a novel simulation procedure, called the green nested simulation (GNS) procedure, that pools inner simulation outputs from different outer scenarios but avoids the difficulties in the regression- and metamodeling-based techniques.
The contributions of our study include:
\begin{enumerate}
	\item We propose an efficient simulation procedure that is non-nested in nature and recycles the same set of inner simulation outputs via the likelihood ratio method to estimate $\E[H|X]$ in different scenarios.
	The proposed procedure does not require any model selection or calibration.
	
	
	\item We establish that the asymptotic bias, variance, and MSE of the risk estimator all converge to zero at rate $\cO(\Gamma^{-1})$.
	This convergence rate is faster than that of nested stimulation with optimal allocation and that of the kernel-based approach.
	Most importantly, $\cO(\Gamma^{-1})$ is the same fast convergence rate as a non-nested Monte Carlo simulation.
	
	\item We establish central limit theorem (CLT) and valid variance estimates for the nested simulation estimators for different forms of $\rho$.
	These results enable users to construct valid confidence intervals for nested simulation estimators without running macro replications.
	The analysis is non-trivial considering that all conditional expectations are estimated using the same set of inner simulation outputs thus are all correlated.
	
	
	
\end{enumerate}

In essence, the GNS procedure recycles the same set of simulation outputs, via the likelihood ratio method~\citep{beckman1987monte,l1990unified}, to estimate the conditional expectation $\E[H|X]$ for different scenarios $X$.
The GNS procedure is inspired by green simulation~\citep{feng2017green} and likelihood ratio metamodeling~\citep{dong2018unbiased}, which improve simulation efficiency by reusing simulation outputs.
Stochastic mesh for American option pricing~\citep{broadie2000pricing,broadie2004stochastic,avramidis1999efficiency,avramidis2004convergence} is also an application of the likelihood ratio method.
The GNS procedure and the stochastic mesh are mathematically similar but the two approaches tackle different problems, serve different purposes, and are applied in different contexts.
The former aims to solve nested estimation problems (risk measurement) while the latter solves a dynamic programming problem (American option pricing).

The rest of this paper is organized as follows.
The problem statement and general mathematical framework are given in Section~\ref{sec:problem}.
Sections~\ref{sec:AnalysisLoss} and~\ref{sec:AnalysisRisk} present the main asymptotic analyses: Section~\ref{sec:AnalysisLoss} analyzes the convergence of the green loss estimator to the conditional expectation random variable and Section~\ref{sec:AnalysisRisk} analyzes the asymptotic bias, variance, MSE, as well as the CLT and valid confidence interval of the portfolio risk estimator.
Numerical experiments are summarized in Section~\ref{sec:Experiment}, followed by conclusions in Section~\ref{sec:Conclusions}.
Technical proofs and auxiliary discussions are provided in the appendices.

\section{A Sample Recycling Approach}\label{sec:problem}

\subsection{Standard Nested Simulation}\label{subsec:SNS}

Standard nested simulation (SNS), as illustrated in Figure~\ref{fig:NestedSim}, is a common approach for estimating the quantity in Equation~\eqref{eq:rho}.
\begin{enumerate}
	\item (Outer simulation) Simulate $n$ independent and identically distributed (i.i.d.) outer scenarios, denoted by $X_1,\ldots,X_n$.
	
	\item (Inner simulation) For each scenario $X_i$, $i=1,\ldots,n$, simulate $m'$ i.i.d. inner replications, e.g., $Y_{i1},\ldots,Y_{im'} \stackrel{i.i.d.}{\sim}f(y|X=x_i)$ then estimate $L(X_i)$ by $		L^{SNS}_{m'}(X_i) = \frac{1}{m'}\sum_{j=1}^{m'} H(X_i,Y_{ij})$.
	
	\item (Risk estimation) Estimate the risk measure $\rho$ in~\eqref{eq:rho} by $\rho^{SNS}_{m'n} = \avgni g(L^{SNS}_{m'}(X_i))$.
\end{enumerate}
In general, the risk estimation step treats $L^{SNS}_{m'}(X_1),\ldots,L^{SNS}_{m'}(X_n)$ as i.i.d. samples of $L(X)$ to estimate different risk measures.
In this study, we focus on risk measures of the form~\eqref{eq:rho} with different risk functions $g:\R \mapsto \R$.

For illustrative purpose, we present a financial risk measurement example.
Let $S_t$ be a vector of risk factors, which may be the values of equities, bonds, interest rates, exchange rates, etc., at any time $t\geq 0$.
Consider a portfolio of financial instruments, which may include stocks, bonds, and derivatives whose values are affected by the risk factors.
Let $t=0$ be the current time when the initial risk factor values $S_0$ are known and $T>0$ be the maximum maturity of all the instruments in the portfolio.
The portfolio manager is interested in estimating some risk measures of the portfolio's profit and loss at a fixed future time $\tau\in(0,T)$.
Specifically, let $V_\tau$ be the portfolio value at time $\tau$ so the time~$\tau$ portfolio loss is given by $L_\tau = V_0-V_\tau$, which is a random variable at time~$0$.
Nested simulation can be used to estimate risk measures of $L_\tau$: The risk factors \textit{up to $\tau$} are denoted by $X=\{S_t:t\in[0,\tau]\}$, which are the outer-level scenarios.
The risk factors \textit{exceeding $\tau$} are denoted by $Y = \{S_t: t\in (\tau,T]\}$, which are the inner-level sample paths.
The inner simulation model $H(X,Y)$ is the discounted portfolio payoff for the simulated path $(X,Y)$ and the risk function $g(\cdot)$ depends on the risk measure of interest.

As alluded in the introduction, important risk measures such as exceedance probability, Conditional Value-at-Risk (CVaR)\footnote{Also known as the expected shortfall (ES) and conditional tail expectation (CTE).}, and squared tracking error, can all be written as~\eqref{eq:rho} with different risk functions like the indicator function $g(x)=\1\{x\geq x_0\}$, the hockey-stick function $g(x)=(x-x_0)^+ = \max\{x-x_0, 0\}$, and the quadratic function $g(x)=(x-x_0)^2$.
These three risk functions can also be used to approximate more general risk functions, such as those with a finite number of non-differentiable or discontinuous points~\citep[see discussions in][]{hong2017kernel}.

Standard nested simulation is computationally burdensome due to its nested nature, which requires a simulation budget of $\Gamma=m'n$ inner replications.
Moreover, this nested structure leads to a wasteful use of the simulation budget because each estimator $L^{SNS}_{m'}(X_i)$ only uses the $m'$ inner stimulation outputs associated with scenario $X_i$ and ignores the $m'(n-1)$ inner simulation outputs from the other scenarios.

In the next section, we propose an efficient simulation procedure that circumvents the nested structure between the outer and inner simulation by recycling all inner simulation outputs in estimating $L(X_i)$ for every scenario $X_i$.
This recycling saves computations and improves efficiency.

\subsection{Sample Recycling via Likelihood Ratios}

Let $\mathcal{X}\subseteq \R^d$ be the scenario space and $X\in \mathcal{X}$ be a given scenario.
For example, $\mathcal{X}$ may be the support of the random scenario $X$.
Also, let $f(y|X)$ be the conditional density of the inner random variable $Y$ given the scenario $X$.
In other words, the distribution of the inner random variable $Y$ is characterized by the outer scenario $X$.
This is a mild limitation of our method, as majority of risk measurement problems and many nested estimation problems in operations research satisfy this condition.

Suppose there exists a \textit{sampling density} $\ftilde(y)$.
We assume that one can generate samples from $\ftilde(y)$ and can calculate values for both $\ftilde(y)$ and $f(y|x)$.
Moreover, the sampling density $\ftilde$ satisfies the condition that $H(x,y) f(y|x) = 0$ whenever $\ftilde(y)=0$.
Then $L(X)=\E[H(X,Y)|X]$ can be written as
\begin{equation}\label{eq:LtauLR}
L(X) = \E[H(X,Y)|X] =\E_{\ftilde}\left[H(X,Y)\frac{f(Y|X)}{\ftilde(Y)}\right] = \E_{\ftilde}\left[\hatH(X,Y)\right],
\end{equation}
where shorthand notation $\hatH(x,y):=H(x,y)\frac{f(y|x)}{\ftilde(y)}$ denotes the likelihood-ratio-weighted simulation output and the subscript in the expectations indicates that $Y\sim \ftilde$.
The identity~\eqref{eq:LtauLR} is mathematically identical to importance sampling, but we do not select the sampling density for variance reduction.
We assume that the sampling density $\ftilde$ is given and we only use the likelihood ratio as a way to recycle simulation outputs for different outer scenarios.
As we see in the numerical experiments, in practical applications usually there is a natural choice of sampling distribution $\ftilde$.


In light of~\eqref{eq:LtauLR}, we propose the following green nested simulation (GNS) procedure:
\begin{enumerate}
	\item (Outer simulation) Simulate $n$ independent and identically distributed (i.i.d.) outer scenarios, denoted by $X_1,\ldots,X_n$.
	
	\item (Inner simulation) Simulate $m$ i.i.d. inner replications, e.g., $Y_{1},\ldots,Y_{m} \stackrel{i.i.d.}{\sim}\ftilde(y)$ then estimate $L(X_i)$ by
	\begin{equation}\label{eq:LmXi}
		L_m(X_i) = \avgmj H(X_i,Y_j) \frac{f(Y_j|X_i) }{\ftilde(Y_j)} = \avgmj \hatH(X_i,Y_j), \quad i=1,\ldots,n.
 	\end{equation}
	
	\item (Risk estimation) Estimate the risk measure $\rho$ in~\eqref{eq:rho} by
	\begin{equation}\label{eq:rhomn}
	\rho_{mn} = \avgni g(L_m(X_i)).
	\end{equation}
\end{enumerate}
Figure~\ref{fig:GNS} depicts the GNS procedure, which does not have the nested structure as in Figure~\ref{fig:NestedSim}.
In the GNS procedure, the outer scenarios $X_1,\ldots,X_n$ and the inner replications $Y_1,\ldots,Y_m$ are simulated separately and independently.
The same inner replications are recycled to estimate all conditional expectations $L(X_1),\ldots,L(X_n)$.

\begin{figure}[h!]
	\newcommand{\frontheightsind}{1.5/1, 0/i, -1.5/n}
	\newcommand{\backheightsind}{2/{1}, 0/{j}, -2/{m}}
	\def\xa{0}
	\def\xb{3.5}
	\def\xc{7}
	\begin{center}
		\begin{tikzpicture}[every node/.style={scale=0.9}]
			\draw (\xa, 0) node(s0) {};
			\foreach \y/\ind in \frontheightsind{
				\draw (\xb,\y) node(st\ind) {$X_{\ind}$};
				\draw[thick,->] (s0.east) -- (st\ind.west);
			}
			
			\foreach \y/\ind in \backheightsind{
				\draw (\xc,\y) node(sp\ind) {$Y_{\ind}$};
				\foreach \y/\find in \frontheightsind{
					\draw (st\find.east) -- (sp\ind.west);
				}
			}
			
			\foreach \indori/\inddes in {st1/sti, sti/stn, sp1/spj, spj/spm}{
				\path (\indori) -- (\inddes) node [font=\large, midway, sloped] {$\dots$};	
			}
			
			\draw [thick, decorate,decoration={brace,amplitude=10pt}](sp1.north -| spm.east) -- (spm.south east) node[black,right, midway,xshift=0.35cm,align=center] (sample) {Sample $Y_j\stackrel{i.i.d.}{\sim} \ftilde(y)$\\\\Estimate $\E[H(X,Y)|X=X_1]$ by $L_m(X_1)$\\\vdots\\Estimate $\E[H(X,Y)|X=X_n]$ by $L_m(X_n)$};
			
			\draw [thick, decorate,decoration={brace,amplitude=10pt}](sample.north east) -- (sample.south east) node[black,right, midway,xshift=0.35cm, align=center] {Estimate $\rho$};

			\draw [thick, decorate,decoration={brace,amplitude=10pt,mirror}](spm.south -| s0.west) -- (spm.south -| stn.center) node[black,midway,yshift=-0.8cm] {Outer simulation};
			\draw [thick, decorate,decoration={brace,amplitude=10pt,mirror}](spm.south -| stn.center) -- (spm.south) node[black,midway,yshift=-0.8cm,align=center] {Sample recycling via\\likelihood ratio};
		\end{tikzpicture}
	\end{center}
	\caption{Schematic illustration of the GNS procedure. The inner simulation replications $Y_j\sim\ftilde$ are recycled for all outer scenarios by weighting the corresponding simulation outputs by appropriate likelihood ratios.}
	\label{fig:GNS}
\end{figure}
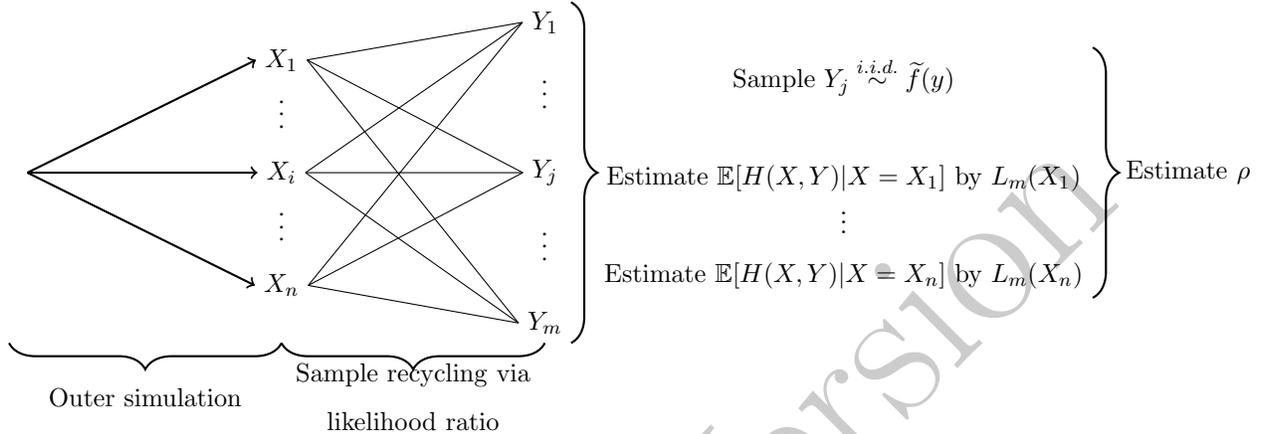

One advantage of the GNS procedure over the standard nested simulation is the computational saving because of sample recycling.
Specifically, when $m=m'$, the GNS procedure and the standard nested simulation use the same number of inner simulation outputs, likelihood-ratio-weighted or not, to estimate each $L(X_i)$.
Then, the computational saving is significant:
\begin{itemize}
	\item The GNS procedure generates $n$ times less inner samples compared to the standard nested simulation.
	In particular, the former simulates $\{Y_{j}\sim \ftilde(y), j=1,\ldots,m\}$ while the latter simulates $\{Y_{ij}\sim f(y|X_i), j=1,\ldots,m', i=1,\ldots,n\}$.
	
	\item In many applications, the inner simulation model can be decomposed into two components, one depends on the scenario $X$ and the other depends on the inner replication $Y$, i.e., $H(X,Y)=H'(h_1(X),h_2(Y))$ for some functions $H'$, $h_1$, and $h_2$. For example, for Asian option payoffs, $h_1(X)$ and $h_2(Y)$ may be the averages of $X$ and $Y$, respectively.
	In these cases, the standard nested simulation requires $m'n$ calculations of the second component $h_2(Y_{ij})$ while the GNS procedure only requires $m$ such computations.
	This is an $n$-fold saving on the second component of the inner simulation model.
	
	\item In some applications, e.g., non-path-dependent payoffs, the inner simulation model depends sole on the inner replication, i.e., $H(X,Y)=H(Y)$.
	Then, the GNS procedure only calculates $m$ inner simulation outputs, i.e., $\{H(Y_{j}), j=1,\ldots,m\}$, then recycle and reuse them in estimating $L(X_i)$ for all $n$ outer scenarios.
	The standard nested simulation, in contracts, calculates $m'$ inner simulation outputs, i.e., $\{H(Y_{ij}), j=1,\ldots,m'\}$, for each of the $n$ outer scenarios.
	This is an $n$-fold saving on the entire simulation output computation.	
	
	\item Moreover, if the user chooses to increase the number of outer scenarios after an experiment, the GNS procedure can continue reusing the same set of inner simulation outputs while more inner replications are required for standard nested simulation.
	
	\item Admittedly, the GNS procedure requires likelihood ratio calculations to reuse the inner simulation outputs, but in most applications computational efforts of the likelihood ratio ${f(Y|X)/\ftilde(X)}$ is small or even negligible compared to the inner simulation model $H(X,Y)$.
	For example, as we see in Section~\ref{sec:Experiment}, in risk management applications where the underlying asset model is Markovian, the likelihood ratio calculation can be simplified.
	Thus this additional cost is worth paying for the savings in generating new inner replications and calculating additional simulation outputs.
\end{itemize}

A second advantage of the GNS procedure is its high accuracy.
When $m=m'n$ so the GNS procedure matches the same simulation budget as standard nested simulation, each $L(X_i)$ is estimated by $m=m'n$ inner simulation outputs in the former versus $m'$ in the latter.
Despite the likelihood ratio weight, since the GNS procedure estimates each $L(X_i)$ with $n$ times more inner simulation outputs than standard nested simulation so the former is expected to be much more accurate than the latter.
Also, as indicated in Equation~\eqref{eq:LtauLR}, the likelihood ratio estimator~\eqref{eq:LmXi} is unbiased.
Compared to the LSMC~\citep{longstaff2001valuing} and to the kernel smoothing approach~\citep{hong2017kernel} for nested simulation, the GNS procedure does not have model error because it does not require the user to select any basis function, kernel function, or kernel bandwidth.

A third advantage of the GNS procedure is the strong convergence of the estimator $L_m(X)$ to $L(X)$ and the fast convergence of the risk estimator $\rho_{mn}$ to $\rho$ as $\min\{m,n\}\to \infty$, which are shown by the asymptotic analyses in Sections~\ref{sec:AnalysisLoss} and~\ref{sec:AnalysisRisk}.


\section{Asymptotic Analysis for Conditional Expectation Estimator $L_m( X)$ }\label{sec:AnalysisLoss}

For notational convenience, in subsequent discussions where no confusion will arise we write simply $L$, $L_m$, and $\hatH$ in places for $L(X)$, $L_m(X)$, and $\hatH(X,Y)$ respectively.
For simulated samples we will use shorthand notations $\Li$, $\Lmi$, and $\hatHij$ for $L(X_i)$, $L_m(X_i)$, and $\hatH(X_i,Y_j)$, respectively.
For example, we may write $\rho_{mn} = \avgni g(\Lmi) = \avgni g(\avgmj \hatHij)$.

\begin{assumption}\label{assm:basic}
	\begin{enumerate}[label=(\Roman*)]
		\item\label{assm:abscon} The support for conditional density $f( y|X)$ is the same for any scenario $X\in\mathcal{X}$.
		Moreover, the sampling density satisfies $H(X, y) f( y|X) = 0$ whenever $\ftilde(y)=0$ for all $X\in\mathcal{X}$.
		\item\label{assm:independence} The inner sample $Y\sim \ftilde(y)$ is independent of the outer scenario $X$. The simulated $\{X_i, i=1,\ldots,n\}$ and $\{Y_{j}, j=1,\ldots,m\}$ are i.i.d. samples of $X$ and $Y$, respectively.
	\end{enumerate}
\end{assumption}

The absolute continuity Assumption~\ref{assm:basic}~\ref{assm:abscon} ensures that the likelihood ratio in~\eqref{eq:LmXi} is well-defined; this is a standard assumption for analyzing importance sampling and the likelihood ratio method.
It can be satisfied if the common support of $f(y|X)$ is contained in the support of $\ftilde(y)$.
The independence Assumption~\ref{assm:basic}~\ref{assm:independence} enables us to use Independence Lemma~\citep[Lemma 2.3.4 in][]{shreve2004stochastic} and properties of U-Statistics~\citep[Section 5 in][]{serfling2009approximation} in our analysis.

For any fixed scenario $X=x$, $L_m(x)$ is an unbiased estimator for $L(x)$ according to Equation~\eqref{eq:LtauLR}.
Our risk measurement problem is more complicated because the scenario $X$ is stochastic.
Nonetheless, we can analyze useful properties of the random variable $\hatH(X,Y)$ and $L_m(X)$.

We first state a useful result for later discussions.
\begin{lemma}\label{lem:finitemoment}
	If Assumption~\ref{assm:basic}~\ref{assm:abscon} holds and $\E\left[|\hatH|^p\right]<\infty$ for some positive integer $p$, then $\E[|L|^p] < \infty$.
\end{lemma}
\begin{proof}
	For any positive integer $p$, $|x|^p$ is a convex function. Then, by the Jensen's inequality
	\begin{equation*}\label{eq:boundedLp}
		\E[|L|^p] = \E[(\E[\hatH| X])^p] \leq \E[\E[|\hatH|^p| X]] = \E[|\hatH|^p] < \infty.
	\end{equation*}
\end{proof}

Lemma~\ref{lem:finitemoment} means that $\E\left[|\hatH|^p\right]<\infty$ directly implies $\E[|L|^p]<\infty$ so the latter does not need to be explicitly stated provided the former holds.
This simplifies the statements of our propositions and theorems, e.g., Proposition~\ref{prop:Lm_as}, whose proof is in Appendix~\ref{app:AnalysisLoss}.

\begin{proposition}\label{prop:Lm_as}
	If Assumption~\ref{assm:basic}~\ref{assm:abscon} holds, then $L_m(x)$ is an unbiased estimator for $L(x)$ for any fixed scenario $x$, i.e., $\E\left[L_m(x)\right] = L(x).$

	In addition, if Assumption~\ref{assm:basic}~\ref{assm:independence} also holds and $\E\left[|\hatH|\right]<\infty$, then $L_m( X)$ is a strongly consistent estimator for $L( X)$, i.e.,
	\begin{equation*}
	L_m( X)\stackrel{a.s.}{\rightarrow} L( X) \mbox{ as } m\rightarrow \infty.
	\end{equation*}
\end{proposition}

The first part of Proposition~\ref{prop:Lm_as} is the well-known unbiasedness of the importance sampling estimator.
The second part shows the almost sure convergence of $L_m(X)$ to $L(X)$ in light of the stochastic of $X$.
This almost sure convergence is useful for establishing asymptotic properties of the GNS risk estimator $\rho_{mn}$.

To facilitate further analysis in Section~\ref{sec:AnalysisRisk}, we establish two more useful lemmas below.
Even though we attribute Lemmas~\ref{lem:rv_moment2p} and~\ref{lem:avg_moment2p} to~\cite{avramidis2004convergence}, our lemmas are extensions of theirs to accommodate the general analysis in this article.
For completeness, we provide their detailed proofs in Appendix~\ref{app:AnalysisLoss}.

\begin{lemma}[Lemma 1 in~\cite{avramidis2004convergence}]\label{lem:rv_moment2p}
	Suppose $R$ is a random variable with $\E[R^{2p}]<\infty$ for some positive integer $p$. Then, for any arbitrary $\sigma$-field $\cG$,
	\begin{equation*}
	\E\left[(R-\E\left[R|\cG\right])^{2p}\right] \leq 4^p\E\left[R^{2p}\right].
	\end{equation*}
\end{lemma}

\begin{lemma}[Lemma 2 in~\cite{avramidis2004convergence}]\label{lem:avg_moment2p}
	Consider identically distributed random variables $\{R_j\}_{j=1}^{m}$ such that $\E[R_1^{2p}]<\infty$ for some positive integer $p$.
	In addition, conditional on an arbitrary $\sigma$-field $\cG$, $\{R_j\}_{j=1}^{m}$ are mutually independent and $\E\left[R_j|\cG\right]=0$ for all $1\leq j\leq m$.
	Then,
	\begin{equation*}\label{eq:Rmoment2p}
	\E\left[\left(\avgmj R_j\right)^{2p}\right]= \frac{c_1}{m^{p}} \E\left[R_1^{2p}\right] +\cO\left(\frac{1}{m^{p+1}}\right)=\cO(m^{-p}),  \mbox{ as } m\rightarrow\infty,
	\end{equation*}
	where $c_1=\binom{2p}{2}\binom{2p-2}{2}\cdots\binom{2}{2}/{p!}$.	
	In particular, for $p=1$,
	\begin{equation}\label{eq:Rmoment2piid}
	\E\left[\left(\frac{1}{m}\sum_{j=1}^m R_j\right)^2\right]=\frac{\E \left[R_1^2\right]}{m}.
	\end{equation}
\end{lemma}

\begin{theorem}\label{thm:Lm_moment2p}
	If Assumption~\ref{assm:basic} holds and $\E\left[\hatH^{2p}\right]<\infty$ for some positive integer $p$, then
	\begin{equation*}
	\E\left[\left(L_m-L\right)^{2p}\right]= \cO\left(m^{-p}\right) \mbox{ as } m\rightarrow\infty.
	\end{equation*}
\end{theorem}

Theorem~\ref{thm:Lm_moment2p} demonstrates the $\cL^{2p}$ convergence of $L_m$ to $L$ at rate $\cO(m^{-p})$.
This is also an important result to establish asymptotic properties, such as bias, variance, MSE, and CLT, of the GNS risk estimator $\rho_{mn}$.

\section{Asymptotic Analysis for Risk Estimator $\rho_{mn}$}\label{sec:AnalysisRisk}

While sample recycling in the GNS procedure leads to computational savings and higher accuracy, as discussed in Section~\ref{sec:problem}, it also introduces dependency among the estimators $L_{m,i}$, $i=1,\ldots,n$.
Despite this intricate dependency, we analyze the asymptotic properties for the GNS estimators in~\eqref{eq:LmXi} and~\eqref{eq:rhomn}.

The asymptotic analysis for $\rho_{mn} = \avgni g(\Lmi)$ is different for different risk functions $g$.
For linear functions, i.e., $g(x)=ax+b$ for some constants $a$ and $b$, $\rho=\E\left[g(\E\left[H|X\right])\right] = a\E\left[H\right]+b$ can be estimated without nested simulation.
To make our study meaningful, we analyze three classes of nonlinear risk functions:
\begin{enumerate}
	\item \textbf{Smooth function}: $g:\R\mapsto\R$ is twice differentiable with a bounded second derivative, i.e., both $g'(x)$ and $g''(x)$ exist for all $x\in\R$ and there exists a nonnegative constant $C_g \in \R^+$ such that $|g''(x)|\leq C_g<\infty$.
	Analysis for this class of risk functions mainly based on the Taylor approximation
	\begin{equation}\label{eq:Taylor}
	g\left(L_m\right) = g\left(L\right) + g'\left(L\right)(L_m-L) + \frac{g''(\Lambda_m)}{2}(L_m-L)^2,
	\end{equation}
	where $\Lambda_m$ is a random variable that lies between $L_m$ and $L$.
	
	\item \textbf{Hockey-stick function}: $g(x): = \max\{x, 0\}$.
	The hockey-stick function has a kink hence is not differentiable at $x=0$, but it is Lipschitz continuous because $|g(x)-g(y)|\leq |x-y|$.
	Moreover, it is clear that $g(x)= x\cdot\1\{x\geq 0\}$ so we can define its derivative $g'(x) = \1\{x \geq 0\}$, which is valid everywhere except at $x=0$; this derivative suffices in our analysis.
	The valid bounds $g(x) \leq |x|$ and $g'(x) \leq 1$ are also useful in our analysis.

	\item \textbf{Indicator function}: $g(x)=\1\{x \geq 0\}$ is neither continuous nor differentiable at $x=0$, which leads to a more complicated analysis than the other two cases.
	When needed, we make additional assumptions and employ a smooth approximation to circumvent this difficulty.
\end{enumerate}
Though different assumptions and mathematical tools are required to analyze the three classes of risk functions, we present a concise and coherent analysis that sheds lights on their similarities and common properties.
We also note that the kink and discontinuity at $x=0$ in our analysis is only for simplification purpose, which can be generalized to any constant threshold $x_0\in\R$ with a change of variable.

Let $L_m-L = Z_m/\sqrt{m}$ and suppose that $Z_m$ has a nontrivial limiting distribution as $m\rightarrow \infty$.
Assumption~\ref{assm:jointdensity} states some assumptions on the joint density function $p_m(\ell,z)$ for $(L,Z_m)$ that aid later analysis.
\begin{assumption}\label{assm:jointdensity}
	\begin{enumerate}[label=(\Roman*)]
		
		\item\label{assm:jointdensity1}
		The joint density $p_m(\ell,z)$ of $(L,Z_m)$ and its partial derivative $\frac{\partial}{\partial \ell} p_m(\ell,z)$ exists for every positive integer $m\geq 1$ and for all $(\ell,z)$.
		
		\item\label{assm:jointdensity2} For every positive integer $m\geq 1$, there exist nonnegative functions $\bar{p}_{0,m}(\cdot)$ and $\bar{p}_{1,m}(\cdot)$ such that
		\begin{equation*}
		p_m(\ell,z) \leq \bar{p}_{0,m}(z) \mbox{ and } \left|\frac{\partial}{\partial \ell} p_m(\ell,z)\right| \leq \bar{p}_{1,m}(z),\quad \forall (\ell,z).
		\end{equation*}
			
		\item\label{assm:jointdensity3} For $i=0,1$ and $0\leq r \leq 2$
		\begin{equation*}
		\sup_m \int_{-\infty}^\infty |z|^r \bar{p}_{i,m}(z) dz <\infty.
		\end{equation*}
	\end{enumerate}
\end{assumption}
Assumption~\ref{assm:jointdensity} is difficult to verify in general, but as argued in~\cite{gordy2010nested}, it can be expected to be true if some of the instruments in the portfolio have sufficiently smooth payoffs.
Mathematically, Assumption~\ref{assm:jointdensity} imposes smoothness and boundedness assumptions on the joint densities $p_{m}(\ell,z)$, which are needed in our analysis to compensate for the lack of differentiability or continuity in the hockey-stick and indicator risk function $g$.
Moreover, Assumption~\ref{assm:jointdensity} implies that the marginal density function of $L$, i.e., $\widetilde{p}(\ell)=\int p_m(\ell,z)dz$ exists.

Using Assumption~\ref{assm:jointdensity}, we can show the two identities in Lemma~\ref{lem:usefuleqs} that are useful for later analysis.
Detailed proof for Lemma~\ref{lem:usefuleqs} is provided in Appendix~\ref{app:MSE}.

\begin{lemma}\label{lem:usefuleqs}
	Suppose Assumptions~\ref{assm:basic} and~\ref{assm:jointdensity} hold. Then,
	\begin{align}
	\E\left[\1\{L_m\geq 0\} - \1\{L\geq 0\}\right] &= \cO(m^{-1}),\mbox{ and}\label{eq:usefuleq1}\\
	\E\left[|L_m\cdot(\1\{L_m\geq 0\} - \1\{L\geq 0\})|\right] &= \cO(m^{-1})\label{eq:usefuleq3}.
	\end{align}
\end{lemma}


\subsection{Asymptotic Bias}\label{subsec:Analysis_Bias}
For any given risk function $g$, the bias of the GNS estimator $\rho_{mn}$ can be decomposed as
\begin{equation}\label{eq:bias}
\Bias[\rho_{mn}] =  \E\left[g\left(L_m\right) - g\left(L\right)\right]= \E\left[g'\left(L\right)(L_m - L)\right] + \E\left[r_m\right].
\end{equation}
for appropriately defined derivative $g'$ where the remainder term is
\begin{equation}\label{eq:remainder}
r_m = g\left(L_m\right) - g\left(L\right) - g'\left(L\right)(L_m - L).
\end{equation}

The first expectation in the RHS of~\eqref{eq:bias} contributes to the bias due to the linear approximation of $g(\cdot)$.
For well defined derivative $g'$ such as the case for smooth and hockey-stick risk functions, this contribution is zero because
\begin{align*}
	&\E\left[g'\left(L\right)(L_m - L)\right] = \E\left[\E\left[g'(L(X))(L_m(X) - L(X))|X\right]\right] \nonumber\\
	=& \E\left[g'(L(X))(\E\left[L_m(X)|X\right] - L(X))\right] \stackrel{(*)}{=}\E\left[g'(L(X))(L(X) - L(X))\right]=0, \label{eq:zerocontribution}
\end{align*}
where $(*)$ holds because $\E\left[L_m(X)|X\right]=L(X)$ by Proposition~\ref{prop:Lm_as}.

We then show that the bias~\eqref{eq:bias} converges to zero at the rate $\cO(m^{-1})$ for all three classes of risk functions.
Specifically, $|\E[r_m]|\leq \E[|r_m|] = \cO(m^{-1})$ for the smooth and hockey-stick risk functions, where the inequality holds by Jensen's inequality for the convex function $|x|$.
Equation~\eqref{eq:usefuleq1} in Lemma~\ref{lem:usefuleqs} directly indicates that the $\E\left[g\left(L_m\right) - g\left(L\right)\right]=\cO(m^{-1})$ for indicator risk function $g(x)=\1\{x\geq 0\}$.
\begin{itemize}
	\item For a smooth risk function $g$, using the Taylor approximation~\eqref{eq:Taylor} and Theorem~\ref{thm:Lm_moment2p} with $p=1$, we have
	\begin{equation}\label{eq:smoothbias}
	\left|\E\left[r_m\right]\right| \leq \E\left[|r_m|\right] = \E\left[\frac{|g''(\Lambda_m)|}{2}(L_m-L)^2\right] \leq \frac{C_g}{2}\E\left[(L_m-L)^2\right]=\cO(m^{-1}).
	\end{equation}
	
	\item For the hockey-stick risk function $g(x)=\max\{x,0\} = x\cdot\1\{x \geq 0\}$, we define $g'(x)=\1\{x \geq 0\}$ so
	\begin{align*}
		r_m = L_m\cdot\1\{L_m \geq 0\} - L\cdot\1\{L \geq 0\} - \1\{L \geq 0\} (L_m-L) = L_m\cdot(\1\{L_m\geq 0\} - \1\{L\geq 0\})
	\end{align*}
	Then, using Equation~\eqref{eq:usefuleq3} in Lemma~\ref{lem:usefuleqs}, we have
	\begin{equation}\label{eq:hockeysticbias}
	\left|\E[r_m]\right| \leq \E[|r_m|] = \E\left[|L_m\cdot(\1\{L_m\geq 0\} - \1\{L\geq 0\})|\right]=\cO(m^{-1}).
	\end{equation}
	
	
	\item For the indicator risk function $g(x)=\1\{x\geq 0\}$, we consider the bias directly, i.e., $\Bias[\rho_{mn}] =  \E\left[\1\{L_m\geq 0\} - \1\{L\geq 0\}\right]$, which, based on Equation~\eqref{eq:usefuleq1} in Lemma~\ref{lem:usefuleqs}, converges at the rate $\cO(m^{-1})$.

\end{itemize}

Proposition~\ref{prop:bias} summarizes the above discussions about asymptotic biases.
\begin{proposition}\label{prop:bias}
	Suppose that Assumption~\ref{assm:basic} and one of the following sets of assumptions hold:
	\begin{enumerate}
		\item The risk function $g(\cdot)$ is twice differentiable with a bounded second derivative and $\E\left[\hatH^2\right]<\infty$, or
		\item The risk function $g(\cdot)$ is a hockey-stick function and Assumption~\ref{assm:jointdensity} holds, or
		\item The risk function $g(\cdot)$ is an indicator function and Assumption~\ref{assm:jointdensity} holds.
	\end{enumerate}
	Then,
	\begin{equation*}
	\Bias[\rho_{mn}] = \cO(m^{-1}).
	\end{equation*}
\end{proposition}

We can see the advantage of our GNS estimator $L_m$ by comparing Proposition~\ref{prop:bias} to analogous bias results for other nested estimators.
In the GNS procedure, the total number of inner samples is $m$.
The total number of inner samples for the standard nested simulation is $\Gamma=m'n$, where $n$ is the number of outer scenarios and $m'$ is the number of inner samples per outer scenario.
So we consider $m=\Gamma$ a fair comparison, i.e., the same simulation budget, between these two procedures.
Proposition~\ref{prop:bias} shows that the bias of the GNS estimator $\rho_{mn}$ converges to zero at the rate of $\cO(m^{-1})=\cO(\Gamma^{-1})$, which is fast and depends only on the simulation budget.
In contrast, the bias of the standard nested simulation estimator using the optimal budget allocation scheme in~\cite{gordy2010nested} is $\cO(\Gamma^{-1/3})$.
The bias of the regression-based procedure in~\cite{broadie2015risk} depends on the selection of the basis functions and is generally non-zero regardless of the simulation budget.
The bias of the kernel-based procedure in~\cite{hong2017kernel} depends not only on the simulation budget but also on the kernel bandwidth.


\subsection{Asymptotic Variance}\label{subsec:Analysis_variance}

To analyze the variance of the GNS estimator $\rho_{mn}$, we first note that
\begin{align}
	&\Var[\rho_{mn}] = \E\left[\left(\avgni g\left(\Lmi\right) - \E\left[g\left(L_m\right)\right]\right)^2\right]\nonumber\\
	=& \E\left[\left(\avgni \left(g\left(\Lmi\right) - g\left(\Li\right)\right) + \avgni \left(g\left(\Li\right) -\E\left[g\left(L\right)\right]\right)+\left(\E\left[g\left(L\right)\right]-\E\left[g\left(L_m\right)\right]\right)\right)^2 \right]\nonumber\\
	\stackrel{(*)}{\leq}& 3\E\left[\left(\avgni \left(g\left(\Lmi\right) - g\left(\Li\right)\right)\right)^2  + \left(\avgni g\left(\Li\right) -\E\left[g\left(L\right)\right]\right)^2+\left(\E\left[g\left(L\right)-g\left(L_m\right)\right]\right)^2 \right]\nonumber\\
	\stackrel{(**)}{=}& 3\E\left[\left(\avgni \left(g\left(\Lmi\right) - g\left(\Li\right)\right)\right)^2 \right]  + \frac{3}{n}\Var\left[g\left(L\right)\right] +3\left(\Bias[\rho_{mn}]\right)^2,\label{eq:varbound}
\end{align}
where $(*)$ holds by inequality~\eqref{eq:CauthySchwarzIneq1} in Appendix~\ref{app:MSE} and $(**)$ holds by applying Equation~\eqref{eq:Rmoment2piid} in Lemma~\ref{lem:avg_moment2p} to the second term.


We then analyze the three terms in Equation~\eqref{eq:varbound} separately: The last term converges at the rate of $\cO(m^{-2})$ by Proposition~\ref{prop:bias}. For the second term, we assume that $\E[(g(L))^2]<\infty$ so $\Var[g(L)] < \infty$.
As a result, the second term in Equation~\eqref{eq:varbound} converges at the rate of $\cO(n^{-1})$.
Note that $\E[(g(L))^2]<\infty$ is a standard assumption, which dictates that the Monte Carlo estimator for $\rho$ has a finite variance.
For the hockey-stick function $g(x)=\max\{0,x\} \leq |x|$, $\E[(g(L))^2]<\infty$ is satisfied if $\E\left[L^2\right]<\infty$, which, by Lemma~\ref{lem:finitemoment}, holds if $\E[\hatH^2]<\infty$. For the indicator function $g(x)=\1\{x\geq 0\} \leq 1$, this assumption is implicitly satisfied because $\E[(g(L))^2]\leq 1$.

It remains to analyze the first term in Equation~\eqref{eq:varbound}.
Using the inequality~\eqref{eq:CauthySchwarzIneq2} in Appendix~\ref{app:MSE}, we have
\begin{equation*}
	\E\left[\left(\avgni \left(g\left(\Lmi\right) - g\left(\Li\right)\right)\right)^2 \right] \leq \E\left[(g\left(L_m\right)-g\left(L\right))^2\right].
\end{equation*}
\begin{itemize}
	\item For smooth risk functions, using the Taylor approximation~\eqref{eq:Taylor}, we have
	\begin{align}
	\E\left[(g\left(L_m\right)-g\left(L\right))^2\right]&= \E\left[\left(g'\left(L\right)(L_m-L)+\frac{g''(\Lambda_m)}{2}(L_m-L)^2\right)^2\right]\nonumber\\
 &\stackrel{(*)}{\leq} 2 \E\left[(g'\left(L\right)(L_m-L))^2\right] + 2\E\left[\left(\frac{g''(\Lambda_m)}{2}(L_m-L)^2\right)^2\right]\nonumber\\
&\stackrel{(**)}{\leq} 2\left(\E\left[(g'\left(L\right))^4\right]\right)^{1/2}\left(\E\left[(L_m-L)^4\right]\right)^{1/2} + \frac{C_g^2}{2}\E\left[(L_m-L)^4\right] \nonumber\\
	&\stackrel{(***)}{=} \cO(m^{-1}) + \cO(m^{-2}) = \cO(m^{-1})\label{eq:varsmooth}
	\end{align}
	where $(*)$, $(**)$, and $(***)$ hold by ~\eqref{eq:CauthySchwarzIneq2}, ~\eqref{eq:CauthySchwarzIneq3}, and Theorem~\ref{thm:Lm_moment2p} with $p=2$, respectively, provided that $\E\left[(g'\left(L\right))^4\right]<\infty$ and $\E[\hatH^4]<\infty$.
	
	\item For the hockey-stick risk function, due to its Lipschitz continuity, i.e., $|g(x)-g(y)|\leq |x-y|$, we have
	\begin{equation}\label{eq:varhockeystick}
	\E\left[(g\left(L_m\right)-g\left(L\right))^2\right] \leq \E\left[(L_m-L)^2\right] = \cO(m^{-1}),
	\end{equation}
	where the equality holds by Theorem~\ref{thm:Lm_moment2p} with $p=1$, provided that $\E[\hatH^2]<\infty$.
	
	\item For the indicator risk function, we consider the first term in Equation~\eqref{eq:varbound} directly and show that it converges at the rate of $\cO(m^{-1})+\cO(n^{-1})$.
	Assumption~\ref{assm:indicatordensity2} is needed for the analysis in this case.
	Detailed analysis is provided in Appendix~\ref{app:MSE}.
\end{itemize}

\begin{assumption}\label{assm:indicatordensity2}
	\begin{enumerate}[label=(\Roman*)]
		
		\item\label{assm:jointdensityB1}
		For any $i,k\in\{1,...,n\}$ and $i\neq k$, the joint density $q_m(\ell_1,\ell_2,z_1,z_2)$ of $(\Li,L_k,Z_m(X_i),Z_m(X_k))$ and its partial derivatives $\frac{\partial}{\partial \ell_u} q_m(\ell_1,\ell_2,z_1,z_2)$ ($u=1,2$) exist for every $m$ and for all $(\ell_1,\ell_2,z_1,z_2)$.
		
		\item\label{assm:jointdensityB2} For every $m\geq 1$, there exist nonnegative functions $\bar{q}_{v,m}(z_1,z_2), (v=0,1)$ such that for $u=1,2$,
		\begin{align*}
			q_m(\ell_1,\ell_2,z_1,z_2) \leq \bar{q}_{0,m}(z_1,z_2) \mbox{ and }
			\left|\frac{\partial}{\partial\ell_u} q_m(\ell_1,\ell_2,z_1,z_2)\right| \leq \bar{q}_{1,m}(z_1,z_2),\ \ \forall(\ell_1,\ell_2,z_1,z_2).
		\end{align*}
		
		\item\label{assm:jointdensityB3} For $v=0,1$ and any $r_1,r_2 \geq 0$ with $r_1 + r_2 \leq 3$,
		\begin{equation*}
			\sup_m \int_\R |z_1|^{r_1}|z_2|^{r_2} \bar{q}_{v,m}(z_1,z_2) dz_1dz_2 <\infty.
		\end{equation*}
	\end{enumerate}
\end{assumption}

\begin{proposition}\label{prop:var}
	Suppose that Assumption~\ref{assm:basic} and one of the following sets of assumptions hold:
	\begin{enumerate}
		\item The risk function $g(\cdot)$ is twice differentiable with a bounded second derivative, $\E[\left(g\left(L\right)\right)^2]<\infty$, $\E[\left(g'\left(L\right)\right)^4]<\infty$, and $\E[\hatH^4]<\infty$, or
		\item The risk function $g(\cdot)$ is a hockey-stick function, $\E[\hatH^2]<\infty$, and Assumption~\ref{assm:jointdensity} holds, or
		\item The risk function $g(\cdot)$ is an indicator function and Assumptions~\ref{assm:jointdensity} and \ref{assm:indicatordensity2} hold.
	\end{enumerate}
	Then,
	\begin{equation*}
	\Var[\rho_{mn}] = \cO(m^{-1}) +\cO(n^{-1}) = \cO(\max\{m^{-1},n^{-1}\}).
	\end{equation*}
\end{proposition}

Proposition~\ref{prop:var} implies that the number of outer scenarios should grow in the same order as the number of inner samples for the variance to converge quickly; this is also the condition for the MSE to converge quickly.
Matching the total number of inner samples in the GNS procedure and the standard nested simulation, i.e., $m=\Gamma$, Proposition~\ref{prop:var} states that the former's variance converges at $\cO(m^{-1})=\cO(\Gamma^{-1})$ while the latter's variance converges at $\cO(\Gamma^{-2/3})$~\citep[Proposition 2]{gordy2010nested}.



\subsection{Asymptotic Mean Square Error}\label{subsec:Analysis_MSE}
Combining Propositions~\ref{prop:bias} and~\ref{prop:var}, we immediately establish the asymptotic MSE of $\rho_{mn}$, as summarized in Theorem~\ref{thm:MSE}.
\begin{theorem}\label{thm:MSE}
	Suppose the conditions in Proposition~\ref{prop:var} hold. Then,
	
	\begin{equation*}
	\MSE(\rho_{mn}) = \cO(m^{-1}) +\cO(n^{-1}) = \cO(\max\{m^{-1},n^{-1}\}).
	\end{equation*}
\end{theorem}


Theorem~\ref{thm:MSE} shows that $n$ and $m$ should grow at the same rate for the MSE of the GNS estimator to converge quickly.
Then, matching the total number of inner samples in the GNS procedure and the standard nested simulation, i.e., $m=\Gamma$, the GNS estimator's MSE converges at $\cO(m^{-1})=\cO(\Gamma^{-1})$ but the MSE of the nested simulation with optimal simulation budget allocation in~\cite{gordy2010nested} converges at $\cO(\Gamma^{-2/3})$.
Clearly the GNS estimator's MSE converges faster.
Also, the GNS procedure is arguably easier to implement compared to the regression-based approach~\citep{broadie2015risk} and the kernel-based approach~\citep{hong2017kernel} because the GNS procedure does not require basis functions, kernel function, or kernel bandwidth.

\subsection{Central Limit Theorem and Variance Estimators}\label{subsec:CLT}
In this section we establish a Central Limit Theorem (CLT) of the GNS risk estimator $\rho_{mn}$ and prove a valid variance estimator for $\rho_{mn}$.
Constructing confidence intervals is a common use of CLT, but the variance of nested estimators are usually difficult to estimate, e.g., by running macro replications to get multiple estimates of $\rho$ then estimate the sample variance.
We propose a variance estimator for $\rho_{mn}$ that requires only one run of the GNS procedure and that converges to the asymptotic variance.
Simply put, the CLT result and variance estimator in this section lead to asymptotically valid confidence intervals of the GNS estimator $\rho_{mn}$.

The analysis for the smooth and hockey-stick risk functions are similar, but are different from the analysis for the indicator risk function.
So, for clarity, we provide separate presentations in Sections~\ref{subsubsec:CLTsmoothhockeystick} and~\ref{subsubsec:CLTindicator}.

\subsubsection{Analysis for Smooth and Hockey-Stick Functions}\label{subsubsec:CLTsmoothhockeystick}
The CLT for the GNS estimator $\rho_{mn}$ with smooth and hockey-stick risk functions are based on two-sample U-statistics~\citep[see Chapter 5 in][for example]{serfling2009approximation}, whose definition and asymptotic normality are stated below.

\begin{definition}\label{def:Ustats}
	Let $\{X_i,i=1,\ldots,n\}$ and $\{Y_j,j=1,\ldots,m\}$ be i.i.d. samples of two independent random variables $X$ and $Y$, respectively.
	For a given mapping $U(x,y)$, the average $\cU_{mn} = \frac{1}{mn}\sum_{i=1}^{n}\sum_{j=1}^{m}U(X_i,Y_j)$ is called a two-sample U-statistic.
\end{definition}

\begin{lemma}[Asymptotic Normality of a Two-Sample U-Statistic]\label{lem:Ustatistics}
	Let $\cU_{mn}$ be a two-sample U-statistic in Definition~\ref{def:Ustats}.
	If $\E\left[(U(X,Y))^2\right]<\infty$, then $\cU_{mn}$ is asymptotically normally distributed, as $\min\{m,n\}\rightarrow\infty$, with mean $\mu =\E\left[U(X,Y)\right]$ and variance $\sigma_{mn}^2 = \frac{\sigma_1^2}{n} + \frac{\sigma_2^2}{m}$ where $\sigma_1^2 =\Var\left[\E\left[U(X,Y)|X\right]\right]$ and $\sigma_2^2 = \Var\left[\E\left[U(X,Y)|Y\right]\right]$.
	Mathematically,
	\begin{equation*}
	\frac{\cU_{mn} - \mu}{\sigma_{mn}} \condist N(0,1), \mbox{ as } \min\{m,n\}\rightarrow \infty.
	\end{equation*}
\end{lemma}

In the following, we will show that the GNS estimator $\rho_{mn}$ can be decomposed into two terms: One term is a two-sample U-statistic and the other term vanishes so quickly that it does not affect the asymptotic distribution of $\rho_{mn}$.
Recall that $\Lmi = \avgmj \hatHij$, so $\rho_{mn}$ can be decomposed as
\begin{equation}\label{eq:decomposerho}
\rho_{mn} =\avgni g\left(\Lmi\right) = \cU_{mn} + r_{mn},
\end{equation}
where
\begin{align}
\cU_{mn}&:= \frac{1}{mn}\sum_{i=1}^n\sum_{j=1}^m \left[g\left(\Li\right) + g'\left(\Li\right)\left(\hatHij - \Li\right)\right], \mbox{ and }\label{eq:ustat}\\
r_{mn} &:= \avgni \left[g\left(\Lmi\right) - g\left(\Li\right) - g'\left(\Li\right)\left(\Lmi - \Li\right)\right] .\label{eq:remainderrmn}
\end{align}
By Assumption~\ref{assm:basic}~\ref{assm:independence}, the outer scenarios and the inner samples are i.i.d. samples of two independent random variables.
Then $\cU_{mn}$ in~\eqref{eq:ustat} is a two-sample U-statistic by Definition~\ref{def:Ustats} with the mapping
\begin{equation}\label{eq:Umapping}
U\left(X,Y\right) = g\left(L(X)\right) + g'\left(L\left(X\right)\right)\left(\hatH\left(X,Y\right) - L\left(X\right)\right).
\end{equation}

Next, we validate the conditions of Lemma~\ref{lem:Ustatistics} and restate its conclusion for the mapping~\eqref{eq:Umapping}.
Firstly, note that
\begin{align}
&\E\left[(U(X,Y))^2\right] \stackrel{(*)}{\leq} 3 \left(\E[(g(L))^2] + \E[(g'(L)\hatH)^2] + \E[(g'(L)L)^2]\right)\label{eq:decomposeUmap1}\\
& \stackrel{(**)}{\leq} 3 \left(\E[(g(L))^2] + (\E[(g'(L))^4])^{1/2}(\E[\hatH^4])^{1/2} + (\E[(g'(L))^4])^{1/2}(\E[L^4])^{1/2}\right) \label{eq:decomposeUmap2},
\end{align}
where $(*)$ and $(**)$ hold by inequalities~\eqref{eq:CauthySchwarzIneq1} and~\eqref{eq:CauthySchwarzIneq3} in Appendix~\ref{app:MSE}, respectively.
Then, the moment condition in Lemma~\ref{lem:Ustatistics}, i.e., $\E\left[(U(X,Y))^2\right]<\infty$ can be satisfied by the following:
\begin{itemize}
	\item For the smooth risk functions, in light of~\eqref{eq:decomposeUmap2}, sufficient conditions are $\E[(g(L))^2]<\infty$, $\E[(g'(L))^4]<\infty$, $\E[\hatH^4]<\infty$, and $\E[L^4]<\infty$.
	Also, by Lemma~\ref{lem:finitemoment}, $\E[\hatH^4]<\infty$ implies $\E[L^4]<\infty$ so only the first three conditions need to be explicitly stated.
    \item For the hockey-stick function, in light of~\eqref{eq:decomposeUmap1}, sufficient conditions are $\E[(g(L))^2]<\infty$, $\E[(g'(L)\hatH)^2]<\infty$, and $\E[(g'(L)L)^2]<\infty$.
    Because $g(x)=\max\{x,0\}\leq |x|$ we have $\E[(g(L))^2]\leq \E[L^2]$.
    Also, because $g'(x) = \1\{x\geq 0\}\leq 1$, we have $\E[(g'(L)\hatH)^2]\leq \E[\hatH^2]$ and $\E[(g'(L)L)^2]\leq \E[L^2]$.
    So the sufficient conditions are simplified to $\E[\hatH^2]<\infty$ and $\E[L^2]<\infty$.
    Lastly, by Lemma~\ref{lem:finitemoment}, these conditions are further simplified to $\E[\hatH^2]<\infty$.
\end{itemize}
Note that these moment conditions also ensure the existence of asymptotic variances $\sigma_{1}^2$ and $\sigma_2^2$.

Next, consider the mean $\mu$ and the two variances $\sigma_{1}^2$ and $\sigma_{2}^2$ in Lemma~\ref{lem:Ustatistics} for the mapping~\eqref{eq:Umapping}.
Note that $	\E[U(X,Y)|X] =  g\left(L(X)\right) + g'\left(L\left(X\right)\right)\left(\E\left[\hatH\left(X,Y\right)|X \right] - L\left(X\right)\right) \stackrel{(*)}{=} g(L(X)),$
where $(*)$ holds because $\E\left[\hatH\left(X,Y\right)|X \right]=L(X)$ by~\eqref{eq:LtauLR}.
Also, by the independence of $X$ and $Y$ we have $\E\left[U(X,Y)|Y\right] = \E[g(L)] + \E\left[g'(L)\hatH|Y\right] - \E\left[g'(L)L\right]$, where the first and the last expectations are constants.
Therefore, we have
\begin{align}
	\mu &=\E\left[U(X,Y)\right] = \E\left[\E\left[U(X,Y)|X\right]\right] = \E\left[g\left(L\right)\right] = \rho, \mbox{ and } \nonumber\\
	\sigma_{1}^2 &= \Var\left[\E\left[U(X,Y)|X\right]\right] = \Var\left[g\left(L\right)\right] = \E[g(L)^2] - (\E[g(L)])^2, \mbox{ and } \label{eq:asymvar1}\\
	\sigma_{2}^2 &=  \Var\left[\E\left[g'\left(L\right)\hatH|Y\right]\right] = \E\left[\left(\E\left[g'\left(L\right)\hatH|Y\right]\right)^2\right] - \left(\E\left[g'\left(L\right)\hatH\right]\right)^2.\label{eq:asymvar2}
\end{align}
%
%

Then Lemma~\ref{lem:Ustatistics} implies that $\sigma_{mn}^{-1}(\cU_{mn}-\rho)\condist \cN(0,1)$ as $\min\{m,n\}\rightarrow \infty$.
But, to make a conclusion about the asymptotic distribution of the GNS estimator $\rho_{mn}$, we also need to consider the remainder term $r_{mn}$ in~\eqref{eq:remainderrmn}.

Note that $r_{mn}$ in~\eqref{eq:remainderrmn} is an average of $n$ identically distributed samples of $r_{m}$ as defined in~\eqref{eq:remainder}.
By~\eqref{eq:smoothbias} and~\eqref{eq:hockeysticbias} we have $\E[|r_{mn}|]\leq \E[|r_m|]=\cO(m^{-1})$ and so
\begin{equation*}
	\E\left[\left|\frac{r_{mn}}{\sigma_{mn}}\right|\right] = \left( \frac{\sigma_1^2}{m} + \frac{\sigma_2^2}{n}\right)^{-1/2} \cO(m^{-1})=\cO\left(\left[m\left(\sigma_1^2+\sigma_2^2 \cdot\frac{m}{n}\right)\right]^{-\frac{1}{2}}\right)\rightarrow 0, \mbox{ as } \min\{m,n\}\rightarrow\infty.
\end{equation*}
This means that $\frac{r_{mn}}{\sigma_{mn}}\conlone 0$ and hence $\frac{r_{mn}}{\sigma_{mn}}\condist 0$ as $\min\{m,n\}\rightarrow\infty$.
Finally, applying the Slutsky's theorem to~\eqref{eq:decomposerho} we arrive at the desired CLT result for the GNS estimator $\rho_{mn}$, as stated in Theorem~\ref{thm:CLT}.
\begin{theorem}\label{thm:CLT}
	Suppose that Assumption~\ref{assm:basic} and one of the following sets of assumptions hold:
	\begin{enumerate}
		\item The risk function $g(\cdot)$ is twice differentiable with a bounded second derivative, $\E\left[(g\left(L\right))^2\right] < \infty$, $\E\left[(g'\left(L\right))^4\right] < \infty$, and $\E\left[\hatH^4\right]<\infty$, or
		
		\item The risk function $g(\cdot)$ is a hockey-stick function, Assumption~\ref{assm:jointdensity} holds, and $\E\left[\hatH^2\right]<\infty$.
	\end{enumerate}
	Then,
	\begin{equation*}\label{eq:CLTsmooth}
	\frac{\rho_{mn}-\rho}{\sigma_{mn}} \condist \cN(0,1), \mbox{ as } \min\{m,n\}\rightarrow\infty,
	\end{equation*}
	where $\sigma_{mn}^2 = \frac{\sigma_1^2}{n} + \frac{\sigma_2^2}{m}$, $\sigma_{1}^2 =  \Var\left[g\left(L\right)\right]$, and $\sigma_2^2 =  \Var\left[\E\left[g'\left(L\right)\hatH|Y\right]\right]$.
\end{theorem}

Theorem~\ref{thm:CLT} demonstrates the asymptotic normality of $\rho_{mn}$ and the asymptotic variance decomposition due to the stochasticities of $X$ and $Y$ separately.
The asymptotic variance has two parts: The first part, $\sigma_{1}^2 =  \Var\left[g\left(L\right)\right]$, is due to the stochasticity of the outer scenario $X$, and $\frac{\sigma_{1}^2}{n}$ would have been the asymptotic variance in a classical CLT for the sample average of $n$ i.i.d. samples of $g(L)$.
The second part, $\sigma_2^2 =  \Var\left[\E\left[g'\left(L\right)\hatH|Y\right]\right]$, is due to the stochasticity of the inner sample $Y$ that affects all outer scenarios due to sample recycling.
Moreover, the derivative $g'$ in the inner conditional expectation indicates that $\sigma_{2}^2$ is also affected by the nonlinearity of the risk function $g$.

A CLT result like Theorem~\ref{thm:CLT} is useful for constructing confidence intervals, typically by replacing unknown population mean and variance by the corresponding sample estimates.
However, the variance of nested simulation estimators are typically difficult or costly to estimate.
One way is by running \textit{macro replications}, i.e., independent repetitions of the entire simulation procedure, then estimate the sample variance of i.i.d. samples of nested simulation estimators.
However, standard nested simulation procedure is costly to run even once, so running macro replications is prohibitively burdensome.

In contrast, we propose a variance estimator for our GNS estimator $\rho_{mn}$ that only requires running the GNS procedure once.
Specifically, $\sigma_{mn}^2$ is estimated by $\widehat{\sigma}_{mn}^2 = \frac{\widehat{\sigma}_{1,mn}^2}{n} + \frac{\widehat{\sigma}_{2,mn}^2}{m}$, where the estimators for $\sigma_{1}^2$ and $\sigma_{2}^2$ are
\begin{align}
	\widehat{\sigma}_{1,mn}^2 &= \avgni \left(g\left(\Lmi\right)\right)^2 - \left(\avgni g\left(\Lmi\right)\right)^2, \mbox{ and }\label{eq:sig1hat}\\
	\widehat{\sigma}_{2,mn}^2 &= \avgmj \left(\avgni g'\left(\Lmi\right)\hatHij\right)^2 - \left(\avgni g'\left(\Lmi\right)\Lmi\right)^2, \mbox{ respectively.}\label{eq:sig2hat}
\end{align}
%
Theorem~\ref{thm:varianceestimate} shows that the proposed variance estimators are valid as they converge to the corresponding population variances.
The proof for Theorem~\ref{thm:varianceestimate} is provided in Appendix~\ref{app:CLT}.
\begin{theorem}\label{thm:varianceestimate}
	Suppose the conditions in Theorem~\ref{thm:CLT} hold. Then,
	\begin{equation*}\label{eq:CIsmooth}
		\widehat{\sigma}_{1,mn}^2 \conprob \sigma_1^2\quad\mbox{, }\quad \widehat{\sigma}_{2,mn}^2  \conprob  \sigma_2^2,\quad\mbox{ and }\quad \widehat{\sigma}_{mn}^2/\sigma_{mn}^2 \conprob 1, \quad\mbox{ as } \min\{m,n\}\rightarrow\infty.
		\end{equation*}
\end{theorem}

A direct result of Theorems~\ref{thm:CLT} and~\ref{thm:varianceestimate} is a valid confidence interval for $\rho$ with one run of the GNS procedure, as summarized in Corollary~\ref{cor:confidenceinterval}.
	
\begin{corollary}\label{cor:confidenceinterval}
	Suppose the conditions in Theorem~\ref{thm:CLT} hold. Then, the following is an asymptotically valid confidence interval for the nested estimator $\rho$ with a confidence level of $1-\alpha$:
	\begin{equation*}\label{eq:CI1}
	(\rho_{mn}-z_{1-\alpha/2}\cdot \widehat{\sigma}_{mn}, \ \rho_{mn}+z_{1-\alpha/2}\cdot\widehat{\sigma}_{mn}),
	\end{equation*}
	where $\widehat{\sigma}_{mn}^2 = \frac{\widehat{\sigma}_{1,mn}^2}{n} + \frac{\widehat{\sigma}_{2,mn}^2}{m}$ and $z_{1-\alpha/2}$ is the $1-\alpha/2$ quantile of the standard normal distribution.
\end{corollary}

\subsubsection{Analysis for the Indicator Function}\label{subsubsec:CLTindicator}
The discontinuity of the indicator risk function $g(x)=\1\{x\geq 0\}$ is a major difficulty in establishing CLT for the GNS estimator $\rho_{mn}$ in this case.
To circumvent this difficulty, we consider a sequence of smooth approximations of $g(x)$:
Let $\phi(u) = \frac{1}{4\pi} (1-\cos(u))\cdot \1\{|u| \leq 2\pi\}$, and for any $\epsilon > 0$ we define a function
\begin{equation}\label{eq:smoothapprox}
	\geps(x) = \int_{-\infty}^{x/\epsilon} \phi(u) du=\begin{cases}
		1,& x\geq 2\pi\epsilon,\\
		\dfrac{1}{4\pi}\left[\dfrac{x}{\epsilon}-\sin\left(\dfrac{x}{\epsilon}\right)\right]+\dfrac{1}{2}, &|x|<2\pi\epsilon,\\
		0,& x\leq -2\pi\epsilon.
	\end{cases}
\end{equation}
One can show that $\geps(x)$ is twice differentiable for any $\epsilon > 0$.
Also, $\geps(x)$ converges pointwisely to $\1\{x\geq 0\}$ as $\epsilon\to 0$ everywhere except at $x=0$.
To establish the desired CLT in this case, we use a sequence of $\epsilon_m$ that depends on the number of inner samples $m$. We carefully construct such a sequence in the proof of the CLT,  although this sequence is not part of the theorem statement.

As $\gepsm(x)$ is twice differentiable for any $\epsilon_m>0$, we use the Taylor's theorem for $\gepsm$ to decompose the GNS estimator $\rho_{mn} = \avgni g(\Lmi)$ as follows:
\begin{align}\label{eq:decomposerho_indicator}
	\rho_{mn} &=\avgni g\left(\Lmi\right) = \cU_{\epsilon_m,mn} + r_{\epsilon_m,mn}^a + r_{\epsilon_m,mn}^b + r_{\epsilon_m,mn}^c + r_{\epsilon_m,mn}^d,
\end{align}
where
\begin{align*}
\cU_{\epsilon_m,mn} &:= \avgni\left[ g(\Li)+\gepsm'(\Li)(\Lmi-\Li)\right],\\
r_{\epsilon_m,mn}^a &:= \avgni \gepsm''(\Li)(\Lmi-\Li)^2,\\
r_{\epsilon_m,mn}^b &:= \avgni \left[\gepsm(\Li)-g(\Li)\right],\\
r_{\epsilon_m,mn}^c &:= \avgni \left[g(\Lmi)-\gepsm(\Lmi)\right],
\end{align*}
and $r_{\epsilon_m,mn}^d$ is the higher-order remainder term in the Taylor's expansion of $\gepsm$.
%
The decomposition~\eqref{eq:decomposerho_indicator} is more complicated than~\eqref{eq:decomposerho} due to using $\gepsm$ and its Taylor expansion.
Nonetheless, the general strategy to analyze $\rho_{mn}$ in this case is similar to that in Section~\ref{subsubsec:CLTsmoothhockeystick}: First show that $\cU_{\epsilon_m,mn}$ converges to an asymptotically normal distribution then show that the other remainder terms quickly vanishes.
We provide some insights in this section and defer the detailed proofs to Appendix~\ref{app:CLTindicator}.

We assume that the joint density $\psi(x,\ell)$ of $(X,L(X))$ exists and define $\psi_0(x)=\psi(x,0)$ for notational convenience.
Assumption~\ref{assm:indicatordensity} is useful for establishing asymptotic results for $\rho_{mn}$ with the indicator risk function.
\begin{assumption}\label{assm:indicatordensity}
	\begin{enumerate}[label=(\Roman*)]
		\item\label{assum:existence} The partial derivative $\frac{\partial}{\partial \ell} \psi(x,\ell)$ exists for all $x$ and $\ell$ and there exists a nonnegative function $\psi_1(x)$ such that $|\frac{\partial}{\partial \ell} \psi(x,\ell)|\leq \psi_1(x)$ in any open neighborhood of $(x,0)$ for all $x$.
		
		\item\label{assum:boundedmoments} For $i=0,1$, the following quantities are finite,
		\begin{equation*}
			\int \psi_i(x)\d x<\infty,\,\, \E\left[\int \left(\hatH(x,Y)\right)^2 \psi_i(x)\d x\right]<\infty,\, \mbox{ and } \E\left[\left(\int \left|\hatH(x,Y)\right|\psi_i(x)\d x\right)^2\right]<\infty.
		\end{equation*}
	\end{enumerate}
	
\end{assumption}

Assumption~\ref{assm:indicatordensity}~\ref{assum:existence} is similar to~Assumption~\ref{assm:jointdensity}, which is useful for applying Taylor theorem to the joint density function $\psi(x,\ell)$ of $(X,L(X))$.
Assumption~\ref{assm:indicatordensity}~\ref{assum:boundedmoments} may seem intricate, but it is a moment condition in disguise: Similar to the moment conditions in Theorem~\ref{thm:CLT} for the smooth and hockey-stick risk functions, Assumption~\ref{assm:indicatordensity}~\ref{assum:boundedmoments} guarantees the existence of the asymptotic variance~\eqref{eq:IndVar2} for the indicator risk function.
We note that the first two conditions in Assumption~\ref{assm:indicatordensity}~\ref{assum:boundedmoments} are sufficient for the third one, but we state the latter explicitly nonetheless for ease of reference.

Define the mapping $U_{\epsilon_m}(X,Y)=g(L(X))+\gepsm'(L(X))(\hatH(X,Y)-L(X))$. Then we can write $\cU_{\epsilon_m,mn}= \frac{1}{mn}\sum_{i=1}^{n}\sum_{j=1}^{m}U_{\epsilon_m}(X_i,Y_j)$.
Despite the similarity, $\cU_{\epsilon_m,mn}$ is not a two-sample U-statistic as in Definition~\ref{def:Ustats} because the mapping $U_{\epsilon_m}(X,Y)$ depends on the number of scenarios $m$, so Lemma~\ref{lem:Ustatistics} does not apply.
Nonetheless, we show in Appendix~\ref{app:CLTindicator} that $\cU_{\epsilon_m,mn}$ has similar asymptotic properties as $\cU_{mn}$ in Lemma~\ref{lem:Ustatistics}, i.e., $\frac{\cU_{\epsilon_m,mn}-\rho}{\widetilde{\sigma}_{mn}} \condist \cN(0,1)$ where $\widetilde{\sigma}_{mn}^2 = \frac{\widetilde{\sigma}_{1}^2}{n} + \frac{\widetilde{\sigma}_{2}^{2}}{m}$,
\begin{align}
	\widetilde{\sigma}_{1}^{2} &= \Var[g(L)] = \E[\1\{L\geq 0 \}] - (\E[\1\{L\geq 0\}])^2,\mbox{ and }\label{eq:IndVar1}\\
	\widetilde{\sigma}_{2}^{2} &= \E\left[\left(\int\hatH(x,Y)\psi(x,0) \d x\right)^2\right].\label{eq:IndVar2}
\end{align}
We also show in Appendix~\ref{app:CLTindicator} that the remainder terms in~\eqref{eq:decomposerho_indicator} vanish quickly so that $\rho_{mn}$ has the same asymptotic distribution as $\cU_{\epsilon_m,mn}$.
Then we can establish the CLT for $\rho_{mn}$ with the indicator risk function, as stated in Theorem~\ref{thm:CLTIndicator}.

\begin{theorem}\label{thm:CLTIndicator}
	Consider the indicator risk function $g(x)=\1\{x\geq 0\}$. Suppose that Assumptions~\ref{assm:basic},~\ref{assm:jointdensity}, ~\ref{assm:indicatordensity2} and ~\ref{assm:indicatordensity} hold.
	Then,
	\begin{equation*}
	\frac{\rho_{mn}-\rho}{\sigma_{mn}} \stackrel{d}{\rightarrow} \cN(0,1), \mbox{ as } \min\{m,n\}\rightarrow\infty,
	\end{equation*}
	where $\widetilde{\sigma}_{mn}^2 = \frac{\widetilde{\sigma}_1^2}{n} + \frac{\widetilde{\sigma}_2^2}{m}$ and $\widetilde{\sigma}_{1}^2$ and $\widetilde{\sigma}_2^2$ are defined as~\eqref{eq:IndVar1} and~\eqref{eq:IndVar2}, respectively.
\end{theorem}

Next, we propose variance estimators that require only one run of the GNS procedure.
Specifically, $\widetilde{\sigma}_{mn}^2$ is estimated by $\widehat{\widetilde{\sigma}}_{mn}^2 = \frac{\widehat{\widetilde{\sigma}}_1^2}{n} + \frac{\widehat{\widetilde{\sigma}}_2^2}{m}$, where the estimators for $\widetilde{\sigma}_{1}^2$ and $\widetilde{\sigma}_2^2$ are
\begin{align}
	\widehat{\widetilde{\sigma}}_{1,mn}^2 &= \avgni \1\{\Lmi\geq 0\}- \left(\avgni \1\{\Lmi\geq 0\}\right)^2, \mbox{ and }\label{eq:sig1hatIndicator}\\
	\widehat{\widetilde{\sigma}}_{2,mn}^2 &= \avgmj \left(\avgni g'_\epsilon\left(\Lmi\right)\hatHij\right)^2, \label{eq:sig2hatIndicator}
\end{align}
respectively.
These variance estimators are valid as they converge to the corresponding asymptotic population variances, as stated in Theorem~\ref{thm:varianceestimateIndicator}; the proof is provided in Appendix~\ref{app:CLTindicator}. 

\begin{theorem}\label{thm:varianceestimateIndicator}
	Suppose the conditions in Theorem~\ref{thm:CLTIndicator} hold. If, in addition, $\E[\hatH^4]<\infty$ and the sequence $\epsilon$ satisfies $\epsilon\rightarrow 0$, $m\epsilon^5\rightarrow \infty$, and $n\epsilon^2\rightarrow \infty$ as $\min\{m,n\}\rightarrow\infty$. Then,
	\begin{equation*}\label{eq:CIsmoothindicator}
		\widehat{\widetilde{\sigma}}_{1,mn}^2 \conprob \widetilde{\sigma}_1^2\quad\mbox{, }\quad \widehat{\widetilde{\sigma}}_{2,mn}^2  \conprob  \widetilde{\sigma}_2^2,\quad\mbox{ and }\quad \widehat{\widetilde{\sigma}}_{mn}^2/\widetilde{\sigma}_{mn}^2 \conprob 1, \quad\mbox{ as } \min\{m,n\}\rightarrow\infty.
	\end{equation*}
\end{theorem}
Note that, unlike Theorem~\ref{thm:CLTIndicator}, the sequence $\epsilon$ is in the statement of Theorem~\ref{thm:varianceestimateIndicator}.
This is because the variance $\widetilde{\sigma}_{2}^{2}$ in~\eqref{eq:IndVar2} involves $\psi(x,0)$, the unknown density function of $(X, L(X))$. Its estimate in~\eqref{eq:sig2hatIndicator} thus requires the smooth approximation function $g_\epsilon$, with $\epsilon$ satisfying the regularity conditions specified in Theorem~\ref{thm:varianceestimateIndicator} to ensure its convergence.

A direct result of Theorems~\ref{thm:CLTIndicator} and~\ref{thm:varianceestimateIndicator} is an asymptotically valid confidence interval for $\rho$ with one run of the GNS procedure, as summarized in Corollary~\ref{cor:confidenceintervalIndicator}.
\begin{corollary}\label{cor:confidenceintervalIndicator}
	Suppose the conditions in Theorem~\ref{thm:varianceestimateIndicator} hold. The following is an asymptotically valid confidence interval for the nested estimator $\rho$ with a confidence level of $1-\alpha$:
	\begin{equation*}\label{eq:CIindicator}
		(\rho_{mn}-z_{1-\alpha/2}\widehat{\widetilde{\sigma}}_{mn},\ \rho_{mn}+z_{1-\alpha/2}\widehat{\widetilde{\sigma}}_{mn}),
	\end{equation*}
	where $\widehat{\widetilde{\sigma}}_{mn}^2 = \frac{\widehat{\widetilde{\sigma}}_{1,mn}^2}{n} + \frac{\widehat{\widetilde{\sigma}}_{2,mn}^2}{m}$ and $z_{1-\alpha/2}$ is the $1-\alpha/2$ quantile of the standard normal distribution.
\end{corollary}	

In summary, for all three classes of risk functions, we establish CLTs for our GNS estimator $\rho_{mn}$, propose valid variance estimators that require a single run of the GNS procedure, and construct asymptotically valid confidence intervals.

\section{Numerical Experiments}\label{sec:Experiment}
In this section, we consider two risk management examples to examine the performance of the proposed GNS procedure compared to the standard nested simulation procedure and a state-of-art regression-based procedure.
The first example shows that the GNS estimator's accuracy increases with the simulation budget and the convergence rate matches the asymptotic analysis in Section~\ref{sec:AnalysisRisk}.
The second example is a larger example with 240 options, which demonstrates the applicability and performance of the GNS procedure in practical problems.
In the examples, we set $m=n$ as this setting leads to the fastest convergence of MSE according to the asymptotic analysis in Section~\ref{sec:AnalysisRisk}.

%
%
In the examples, we consider option portfolios written on one or multiple, e.g., $d$, underlying assets, whose prices follow the Black-Scholes model.
For simplicity, we assume the same expected return $\mu$ for all underlying assets and a constant risk-free rate $r$.
That is, the price dynamics of the underlying assets $\bS_t=(S_{t}^1,...,S_{t}^d)^\top\in{\cal R}^d$ follows the follow stochastic differential equation
\begin{align*}
\d S_{t}^i=\mu' S_{t}^i \d t+\sum_{j=1}^{d}\sigma_{ij}S_{t}^i \d B_{t}^i,\quad i=1,...,d,
\end{align*}
where $\bm{B}_t=(B_{t}^1,...,B_{t}^d)$ is a $d$-dimensional standard Brownian motion and, without loss of generality, $\Sigma = [\sigma_{kk'}]$ is a $d\times d$ sub-triangular volatility matrix that specifies the volatility and correlations of the underlying assets.
Then the asset prices at any time~$t>0$ are
\begin{align}\label{eq:assetmodel}
	S_{t}^i=S_{0}^i \exp\left\{\left(\mu'- \frac{1}{2}\sum_{j=1}^{i} \sigma_{ij}^2\right)t + \sum_{j=1}^{i} \sigma_{ij} B_{t}^i\right\},\quad i=1,...,d.
\end{align}
We note that drift $\mu'$ equals the $k$-th asset's expected return under the real-world probability measure and equals the risk-free rate $r$ under the risk-neutral measure.

We define an option portfolio's maturity as the longest maturity among all options in the portfolio, which is denoted by $T$.
In our simulation experiments, the current time is $t=0$ and asset values are simulated at discrete times $0=t_0<t_1<\cdots<t_N=T$.
We are interested in measuring the portfolio risk at a future time $t_{k^*} = \tau \in(0,T)$, or $k^*\in\{1,\ldots,N-1\}$.
This is a nested estimation problem:
In a standard nested simulation procedure, one first simulates outer scenarios $X=\{\bS_{t_k}, k=1,...,k^*\}$ under the real-world measure then, given $X$, simulates inner sample paths $Y=\{\bS_{t_k}, k=k^*+1,...,N\}$ under the risk-neutral measure.
Denote the portfolio's current value and the payoff (discounted to time~$0$) by $V_0$ and $V_T(X,Y)$, respectively.
The portfolio's loss at time~$\tau$ given $X$ is
$$L(X)=\E\left[\left. V_0-V_T(X,Y)\right|X\right],$$
which is a random variable at time~$0$.
We want to measure the portfolio risk $\rho = \E\left[g(L(X))\right]$, where three risk functions $g$ are considered: a quadratic function $g(x)=(x-x_0)^2$, a hockey-stick function $g(x)=(x-x_0)^+$, and an indicator function $g(x)=\1\{x > x_0\}$, all with a pre-specified threshold $x_0$.

As the Black-Scholes asset model is Markovian, the likelihood ratio calculation is simplified.
Specifically, the outer scenarios $X=\{\bS_{t_k}, k=1,...,k^*\}$ are simulated using the Black-Scholes model under the real-world measure.
Independent to the outer scenarios, we simulate $\bS_{t_{k^*+1}}\sim\ftilde_{k^*+1}(s)$ where $\ftilde_{k^*+1}$ is the marginal log-normal distribution of $\bS_{t_{k^*+1}}$ according to~\eqref{eq:assetmodel} ($k^*$ steps under the real-world measure and 1 step under the risk neutral measure).
Conditional on $\bS_{t_{k^*+1}}$, we simulate later values $\bS_{t_{k^*+2}},\ldots,\bS_{t_{N}}$ under the risk-neutral measure.
Then the likelihood ratio can be calculated very efficiently, e.g.,
\begin{align*}
	\frac{f(Y|X)}{\ftilde(Y)} &= \frac{f(\bS_{t_{k^*+1}},\ldots,\bS_{t_N}|\bS_{t_{1}},\ldots,\bS_{t_{k^*}})}{\ftilde(\bS_{t_{k^*+1}},\ldots,\bS_{t_N})} =\frac{f(\bS_{t_{k^*+1}},\ldots,\bS_{t_N}|\bS_{t_{k^*}})}{\ftilde(\bS_{t_{k^*+1}},\ldots,\bS_{t_N})} \\
	&=\frac{f(\bS_{t_{k^*+1}}|\bS_{t_{k^*}})f(\bS_{t_{k^*+2}}|\bS_{t_{k^*+1}})\ldots f(\bS_{t_N}|\bS_{t_{N-1}})}{\ftilde_{k^*+1}(\bS_{t_{k^*+1}})f(\bS_{t_{k^*+2}}|\bS_{t_{k^*+1}})\ldots f(\bS_{t_N}|\bS_{t_{N-1}})} \\
	&=\frac{f(\bS_{t_{k^*+1}}|\bS_{t_{k^*}})}{\ftilde_{k^*+1}(\bS_{t_{k^*+1}})}.
\end{align*}
Also, calculating the likelihood ratio as a whole is faster than calculating two densities then taking the ratio.



\subsection{10 Barrier Options}
In this example, we consider 10 barrier options written on one underlying asset, i.e., $d=1$.
The asset model parameters are: $S_0^1=100$, $T=1$, $\tau=3/50$, $\mu=8\%$, $r=5\%$ and volatility $\sigma=20\%$.
The option portfolio include 10 barrier options with the same strike $K=90$ but different barriers:
\begin{itemize}
 \item 5 long up-and-out call options with barriers $U=118, 119, 120, 121, 122$, and
 \item 5 long down-and-out call options with barriers $D=78, 79, 80, 81, 82$.
\end{itemize}
In the implementation, when simulating the continuously monitoring maximum and minimum for barrier options, we use $N=200$ time steps and Brownian bridge approximation is applied for any two adjacent time points; see~\cite[pp. 367-368]{glasserman2013monte} for details of Brownian bridge approximations.

Even though there are 10 options in this example, because they are all written on the same underlying asset so we only need to calculate the likelihood ratio once to reuse different simulation outputs.
This is an appealing feature of the GNS procedure: The likelihood ratio calculation depends only on the dimension of the underlying assets, not the number of instruments in a portfolio.

To measure the performance of our GNS procedure, we accurately estimate the true value of $\rho$ as a benchmark:
So we generate a large number, i.e., $10^9$, i.i.d. scenarios $X$ then calculate the corresponding $L(X)$ and $g(L(X))$.
For barrier options, the loss $L(X)$ can be calculated analytically under the Black-Scholes model.
The 90\%-tile of these losses $L(X)$ is used as the threshold $x_0$ in the three different risk functions.
The sample mean of $g(L(X))$ is then an accurate estimate of $\rho$, which is then used to assess the accuracy of the GNS estimator $\rho_{mn}$.
All results reported are estimated based on 1,000 independent macro replications (using the same benchmark).


\begin{figure}[h!]
	\centering
    \hspace*{-2.33cm}
    \includegraphics[scale=0.393]{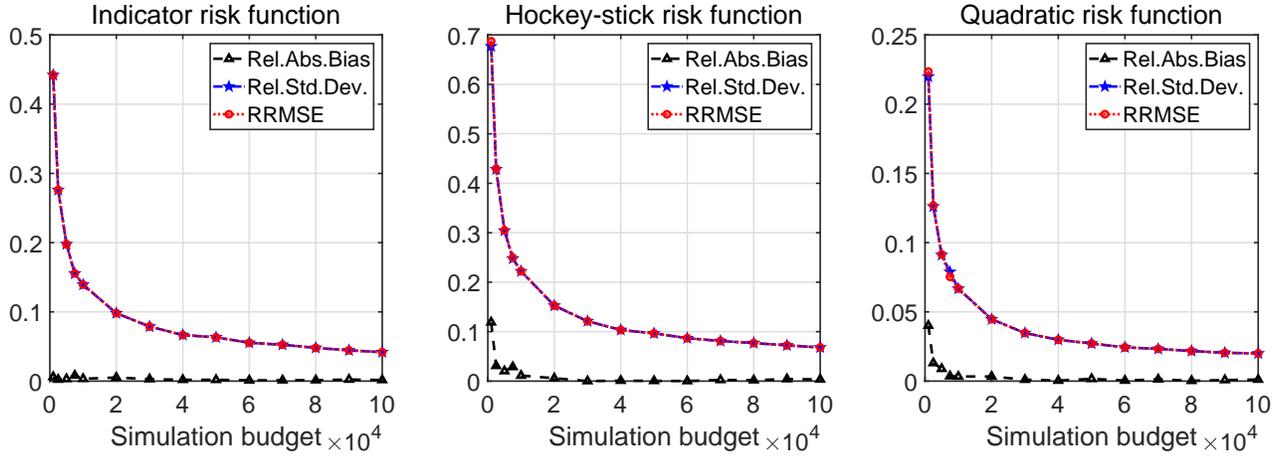}
	\caption{Plot in relative terms for the GNS estimators.}
	\label{fig:BarrierMSEBiasVar}
\end{figure}

Figure~\ref{fig:BarrierMSEBiasVar} depicts the relative absolute biases, relative standard deviations, and the relative root mean squared error (RRMSE) with different simulation budgets.
RRMSE is the ratio between the root MSE of the GNS estimator $\rho_{mn}$ and the benchmark estimate of $\rho$.
The error measures are relative to the benchmark estimate and have the same unit (by taking square roots of the variance and MSE).
We see that all three error measures decrease as the simulation budget increases, as expected.
Moreover, we see that the relative standard deviation almost coincides with the RRMSE, as the relative bias is small.
This is consistent with our intuition that the likelihood ratio estimator $L_m(X)$ is unbiased, which leads to relatively small bias in $g(L_m(X))$ and $\rho_{mn}$.


\begin{figure}[h!]
	\centering
    \hspace*{-2.33cm}
	\includegraphics[scale=0.393]{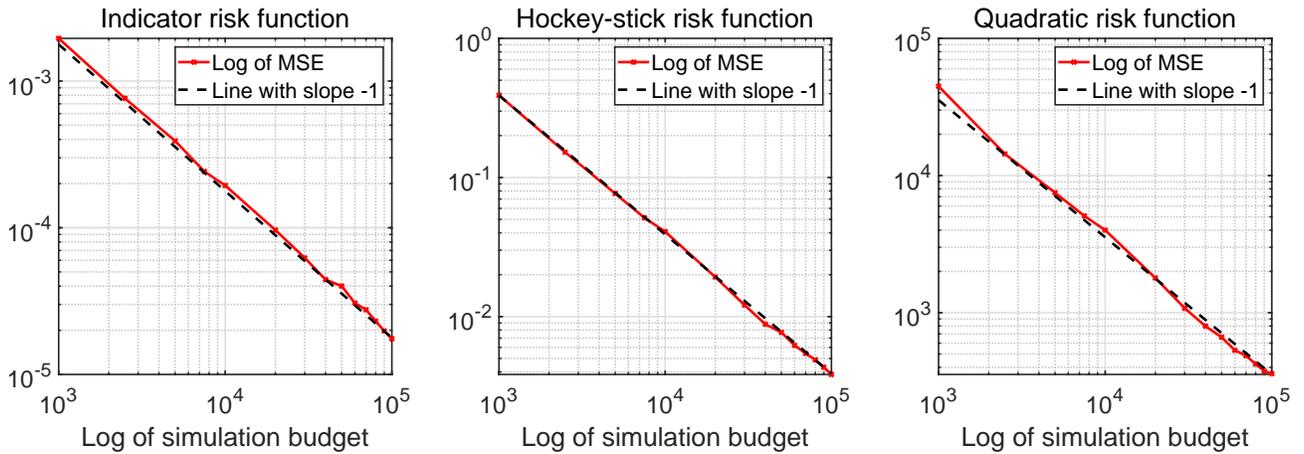}
	\caption{Illustration of the convergence rate of MSE for the GNS estimators.}
	\label{fig:BarrierMSE}
\end{figure}

Figure~\ref{fig:BarrierMSE} depicts the relative MSE (square of RRMSE) in log scale; a dashed line with slope $-1$ is added to the figure to aid visualization.
We see that the relative MSE follows closely with the dashed line, which means that it decreases at $\cO(\Gamma^{-1})$, where $\Gamma$ is the simulation budget of the GNS procedure.
This observation is consistent with Theorem~\ref{thm:MSE}, as the simulation budget is $\Gamma=m$ and we set $m=n$ in this experiment.

%
%
%
%
%

\begin{table}[h!]
	\centering
	\caption{Comparison of relative absolute bias, relative standard deviation, RRMSE, and 90\% confidence interval's coverage probability of the GNS procedure for different simulation budgets and different risk functions. The three error measures are in \% of the benchmark estimate of $\rho$.}
	\label{tab:large1}
	\begin{tabular}{ccccccc}
		\toprule
		Sim. Budget & Risk function $g$ & Rel.Abs.Bias & Rel.Std.Dev. & RRMSE & 90\% CI Cov.Prob.\\
		\midrule
		\multirow{3}{*}{$m=10^3$}
		&Indicator    & 0.61\% & 44.20\%& 44.20\%& 80.50\% \\
		&Hockey-stick & 11.89\%& 67.68\%& 68.72\%& 86.87\% \\
		&Quadratic    & 4.01\% & 21.99\%& 22.35\%& 91.20\% \\
		
		\midrule
		\multirow{3}{*}{$m=10^4$}
		&Indicator    & 0.34\% & 13.93\% & 13.93\% & 88.5\%\\
		&Hockey-stick & 1.12\% & 22.20\% & 22.23\% & 88.3\%\\
		&Quadratic    & 0.33\% & 6.67\% & 6.68\% & 90.7\%\\
		
		
		\midrule
		\multirow{3}{*}{$m=10^5$}
		&Indicator & 0.16\% & 4.18\% & 4.18\% & 90.8\%\\
		&Hockey-stick & 0.35\% & 6.81\% & 6.82\% & 88.4\%\\
		&Quadratic & 0.11\% & 2.00\% & 2.00\% & 88.8\%\\
		
		\bottomrule
	\end{tabular}
\end{table}


Table~\ref{tab:large1} presents a quantitative summary of this experiment.
Consistent with the observations in Figures~\ref{fig:BarrierMSEBiasVar} and~\ref{fig:BarrierMSE}, all three error measures decrease as the simulation budget increases.
Also, the main contribution in the RRMSE is the relative standard deviation, the relative bias is small in all configurations.
Besides the three relative error measures, the last column in Table~\ref{tab:large1} includes the coverage probabilities of the 90\% CIs.
That is, the percentage of 1,000 macro replications where the benchmark estimator falls in the 90\% CIs according to Corollaries~\ref{cor:confidenceinterval} and~\ref{cor:confidenceintervalIndicator}.
We see that the coverage probabilities presented in Table~\ref{tab:large1} are all close to 90\%.
This observation supports the proposed variance estimators for the GNS estimator.
We emphasize that these variance estimators are obtained in one run of the GNS procedure so no macro replication is needed.

\subsection{A Realistic Option Portfolio}
In this example, we consider an option portfolio with 240 options written on 60 different assets.
The assets are divided into three groups, each with 20 assets, and assets from different groups are assumed to be independent.
This is a more realistic risk management problem compared to the previous example.
We compare the GNS procedure's performance with standard nested simulation and a state-of-art regression based approach.

The option portfolio consists of 60 European call options, 60 geometric Asian call options, and 120 barrier options.
Specifically:
\begin{enumerate}
 \item In Group 1, there are 20 underlying assets. Three European call options with strikes $K=90,100,110$ are written on each asset in this group

 \item In Group 2, there are 20 underlying assets. Three geometric Asian call options with strikes $K=90,100,110$ are written on each asset in this group.
 The payoff of a geometric Asian call option is $((\prod_{k=1}^N S_{t_k}^i)^{1/N}-K)^+$, where $K$ is the strike price.
 In the implementation we use $N=50$ time steps for these Asian options.

 \item In Group 3, there are 20 underlying assets. Three up-and-out call options with barrier $U=120$ and three down-and-out call options with barrier $D=90$ are written on these assets.
 Both type of options have three different strikes $K=90,100,110$.
 In the implementation we use $N=200$ time steps for these barrier options use Brownian bridge approximation for any two adjacent time points to simulate the continuously monitoring maximum and minimum values.

\end{enumerate}

We compare the GNS estimator with standard nested simulation estimators and the regression estimator proposed in~\cite{broadie2015risk}.
We consider different budget allocations for the standard nested simulation estimators, to identify the one with the highest accuracy.
For the regression estimator, weighted Laguerre polynomials on the underlying asset price up to an order of 4 are used as the basis functions~\citep[see][for example]{longstaff2001valuing}.

Table~\ref{tab:large2} summarizes the RRMSEs of the three approaches.
We see that, based on the RRMSEs the GNS estimator is significantly more accurate that the standard nested simulation estimators.
For example, for a hockey-stick risk function with $10^5$ simulation budget, the lowest RRMSE of the standard nested simulation estimator, among all allocations presented in the table, is 21.98\%.
The RRMSE of the GNS estimator with the same configuration is only 2.75\%, which is 8 times smaller than the former.
Therefore, if we presume that the optimal convergence rate of nested simulation estimator is achieved, i.e., $\Gamma^{-1/3}$ for RRMSE, then the sampling budget for the nested simulation estimator needs to be $8^3$ times of the GNS estimator to achieve the same level of RRMSE.

\begin{table}[H]
\centering
\caption{Comparison of RRMSEs (\%) for the standard nested simulation estimator, regression estimator, and the GNS estimator. For the standard nested simulation, the allocation $n \times m'$ means that there are $n$ outer scenarios with $m'$ inner samples each.}
\label{tab:large2}
%
%
%
%
%
%
\begin{tabular}{@{}rccccccccc@{}}
	\toprule
	Sim. Budget&\phantom{1}&\multicolumn{4}{c}{Standard nested simulation}&\phantom{1}&Regression &\phantom{1}& GNS \\
	\cmidrule{1-1}\cmidrule{3-6}\cmidrule{8-8}\cmidrule{10-10}
	$m=10^3$ &&$10\times100$&$20\times50$& $40\times25$& $50\times20$&& && \\

Indicator    && 100.30\% & 78.88\% & 73.24\% & 77.61\% && 123.18\%&& 22.75\% \\
Hockey-stick && 148.50\% & 127.14\%& 138.66\%& 153.53\%&& 638.17\%&& 29.26\% \\
Quadratic    && 42.12\%  & 32.69\% & 29.71\% & 31.42\% && 753.40\%&& 13.26\% \\
	
	\midrule
	$m=10^4$ &&$50\times200$&$100\times100$& $200\times50$& $400\times25$&&& & \\
	
	Indicator && 44.50\% & 34.14\% & 34.12\% & 50.27\%&& 16.85\% && 7.00\% \\
	Hockey-stick && 60.12\% & 51.61\% & 58.48\% & 98.09\%&& 48.63\% && 8.87\%\\
	Quadratic && 18.29\% & 14.23\% & 12.81\% & 18.78\%&& 10.92\% && 3.88\% \\
	
%
	
	\midrule
	$m=10^5$ &&$200\times500$&$400\times250$& $1,\!000\times100$& $2,\!000\times50$&&& &\\
	
	Indicator && 21.27\% & 15.59\% & 16.38\% & 26.81\%&& 2.82\% && 2.13\% \\
	Hockey-stick && 28.64\% & 21.98\% & 27.00\% & 48.02\%&& 5.53\% && 2.75\% \\
	Quadratic && 8.84\% & 6.52\% & 6.01\% & 9.08\%&& 1.39\% && 1.16\%  \\
	\bottomrule
\end{tabular}
\end{table}

Table~\ref{tab:large2} also shows that the GNS estimator outperforms the regression estimator, sometimes significantly so, e.g., when the simulation budget is small.
In all experiments presented in Table~\ref{tab:large2}, the GNS estimator has smaller RRMSEs than the regression estimator, although the difference becomes smaller as the simulation budget increases. It should be pointed out that the bias of the regression estimator may persist regardless of how large the simulation budget is, due to the model error in selecting basis functions. By contrast, convergence of the GNS estimator to $\rho$ can be guaranteed theoretically as simulation budget increases.

\section{Conclusions}\label{sec:Conclusions}

We have proposed a green nested simulation (GNS) procedure, that pools inner simulation outputs from different outer scenarios, for solving nested estimation problems. Inner simulation outputs are weighted by likelihood ratios to ensure the unbiasedness of the conditional expectation estimates, helping to produce a convergent GNS estimator. The MSE of the GNS estimator is shown to converge at a rate of $\Gamma^{-1}$, the fastest rate that can be achieved by a typical simulation estimator, where $\Gamma$ is the simulation budget. This rate is achieved by simply recycling the inner simulation outputs weighted by likelihood ratios, without introducing modeling errors that appear to be common in existing regression-based and metamodeling-based methods when selecting basis functions, covariance functions or kernel bandwidth. CLT and variance estimates of the GNS procedure have been established, enabling the construction of asymptotically valid confidence intervals. Numerical examples on the portfolio risk measurement application have shown that the proposed GNS procedure works quite well.

\clearpage
\bibliographystyle{plainnat}
\bibliography{RiskMeasurementLR}

\clearpage
\appendix
\section{Auxiliary proofs for results in Section~\ref{sec:AnalysisLoss}}\label{app:AnalysisLoss}

\begin{proof}[Proofs of Proposition~\ref{prop:Lm_as}]
 Assumption~\ref{assm:basic}~\ref{assm:abscon} ensures that the likelihood ratio is well-defined so for any fixed scenario $x$ we have, by Equation~\eqref{eq:LtauLR},
 \begin{equation*}
 \E\left[L_m(x)\right]= \E\left[\avgmj\hatH(x,Y_j)\right] = \E\left[\hatH(x,Y)\right] = L(x).
 \end{equation*}

 Also, since $\E\left[|\hatH|\right]<\infty$ and $Y_j$, $j=1,\ldots,m$ are i.i.d., by the strong law of large numbers we have $L_m(x)\stackrel{a.s.}{\to} L(x) \mbox{ as } m\to\infty$.
	This means that $\P\left(\lim\limits_{m\to\infty} L_m(x)=L(x)\right)=1$ for any fixed $x$.
	Because $X$ and $Y$ are independent by Assumption~\ref{assm:basic}~\ref{assm:independence}, the Independence Lemma~\citep[see Lemma 2.3.4 in][for example]{shreve2004stochastic} implies that
	\begin{align*}
	&\P\left(\lim_{m\to\infty}L_m(X)=L(X)\right)=\E\left[\1\left\lbrace \lim\limits_{m\to\infty}L_m(X)=L(X)\right\rbrace \right]\\
	=&\E\left[ \E\left[\1\left\lbrace \lim\limits_{m\to\infty}L_m(X)=L(X)\right\rbrace|X \right] \right]
	=\E\left[\left. \P\left(\lim\limits_{n\to\infty}L_m(x)=L(x)\right)\right|_{x=X}\right]=1.
	\end{align*}
	This means that $L_m(X)\stackrel{a.s.}{\to} L(X) \mbox{ as } m\to\infty$ and the proof is complete.
\end{proof}

\begin{proof}[Proof of Lemma~\ref{lem:rv_moment2p}]
	Note that
	\begin{align*}\everymath{\displaystyle}
	\E\left[(R-\E\left[R|\cG\right])^{2p}\right]\leq&\E\left[(|R|+|\E\left[R|\cG\right]|)^{2p}\right]=\E\left[\sum_{k=0}^{2p}\binom{2p}{k} |R|^{2p-k}|\E\left[R|\cG\right]|^{k}\right]\\
	=&\E [R^{2p}]+\E\left[\E\left[R|\cG\right]^{2p}\right]+\sum_{k=1}^{2p-1}\binom{2p}{k}\E \left[|R|^{2p-k}|\E\left[R|\cG\right]|^{k}\right]\\
	\stackrel{(*)}{\leq}&\E [R^{2p}]+\E\left[\E\left[R|\cG\right]^{2p}\right]+\sum_{k=1}^{2p-1}\binom{2p}{k} (\E [R^{2p}])^\frac{2p-k}{2p}\left(\E\left[\left( \E\left[R|\cG\right]\right)^{2p}\right]\right)^\frac{k}{2p}\\
	\stackrel{(**)}{\leq}&\E [R^{2p}]+\E [R^{2p}]+\sum_{k=1}^{2p-1}\binom{2p}{k} (\E\left[R^{2p}\right])^\frac{2p-k}{2p}(\E [R^{2p}])^\frac{k}{2p}\\
	=&\sum_{k=0}^{2p}\binom{2p}{k}\E [R^{2p}]=2^{2p}\E [R^{2p}],
	\end{align*}
	where inequalities $(*)$ and $(**)$ follow from H${\rm \ddot{o}}$lder's and Jensen's inequalities, respectively. The proof is complete.
\end{proof}

\begin{proof}[Proof of Lemma~\ref{lem:avg_moment2p}]
	According to the multinomial theorem and the conditional independence of $R_j$'s, we have
	\begin{eqnarray*}
		\E\left[\left(\avgmj R_j\right)^{2p}\right] &=& \frac{1}{m^{2p}}\sum_{i_1+\cdots+i_k=2p}\frac{(2p)!}{i_1!i_2!\cdots i_k!}\E\left[R_{j_1}^{i_1}\cdots R_{j_k}^{i_k}\right]\\
		&=& \frac{1}{m^{2p}}\sum_{i_1+\cdots+i_k=2p}\frac{(2p)!}{i_1!i_2!\cdots i_k!}\E\left[\E\left[R_{j_1}^{i_1}\cdots R_{j_k}^{i_k}|\cG\right]\right]\\
		&=& \frac{1}{m^{2p}}\sum_{i_1+\cdots+i_k=2p}\frac{(2p)!}{i_1!i_2!\cdots i_k!}\E\left[\E\left[R_{j_1}^{i_1}|\cG\right]\cdots \E\left[R_{j_k}^{i_k}|\cG\right]\right].
	\end{eqnarray*}
	
	We will next bound the value and the number of summands. Since $i_1+\cdots+i_k=2p$, one can show that
	\begin{equation*}
	\E\left[R_{j_1}^{i_1}\cdots R_{j_k}^{i_k}\right] \leq \E\left[\left|R_{j_1}^{i_1} R_{j_2}^{i_2}\cdots R_{j_l}^{i_l}\right|\right]\stackrel{(*)}{\leq} \left(\E\left[ R_{j_1}^{2p}\right]\right)^\frac{i_1}{2p}\cdots\left(\E\left[ R_{j_l}^{2p}\right]\right)^\frac{i_l}{2p}=\E\left[ R_1^{2p}\right] < \infty,
	\end{equation*}
	where $(*)$ follows the generalized H${\rm \ddot{o}}$lder's inequality.
	
	Since $\E\left[R_j|\cG\right]=0$ for all $1\leq j\leq m$, for a summand to be non-zero it must have all $i_1,\ldots,i_k \geq 2$.
	Combine this with $i_1+\cdots +i_k = 2 p$, we have $k\leq p$.
	Table~\ref{tab:terms} summarizes the multinomial coefficients and the number of summands of the form $\E\left[R_{j_1}^{i_1}\cdots R_{j_k}^{i_k}\right]$ for fixed numbers $k=1,\ldots,p$; the special case where $k=p$ is given in the second row.
	
	\begin{table}[h!]
		\centering
		\begin{tabular}{C{2.35cm}|C{2cm}|C{3.5cm}|C{2.8cm}|C{3.5cm}}
			\hline
			summand expression & multinomial coefficient & \# of different $\{i_1,\ldots,i_k\}$ & \# of different $\{j_1,\ldots,j_k\}$ & product\\\hline				
			$\E\left[R_{j_1}^{i_1}\cdots R_{j_k}^{i_k}\right]$& $\displaystyle\frac{(2p)!}{i_1!i_2!\cdots i_k!}$ & \# of integer solution satisfying $i_1,\ldots, i_k \geq 2$ and $i_1+\cdots+i_k=2p$. \newline Does not depend on $m$. & $\displaystyle\binom{m}{k} = \cO(m^k)$ & $\cO\left(m^k\right)\leq \cO(m^{p-1})$ for $k\leq p-1$\\\hline
			$\E\left[R_{j_1}^{2}\cdots R_{j_p}^{2}\right]$& $\displaystyle\frac{(2p)!}{2^p}$ & 1 & $\displaystyle\binom{m}{p} = \frac{m^p}{p!} + \cO(m^{p-1})$ & $\displaystyle c_p m^p + \cO(m^{p-1})$ where $c_p=\frac{(2p)!}{2^p(p!)}$ \\
			\hline
		\end{tabular}
		\caption{A breakdown of the number of summands for $k=1,\ldots,p$ unique of $R_{j}$'s. The binomial coefficients are denoted by $\binom{n}{k} =\frac{n!}{(n-1)!k!}$.}
		\label{tab:terms}
	\end{table}
	For sufficiently large $m$, we have $\binom{m}{k} \leq \binom{m}{p}$ for $k\leq p$.
	Therefore, as $m \to\infty$,
	\begin{equation*}
	\E\left[\left(\avgmj R_j\right)^{2p}\right] = \frac{1}{m^{2p}}\left(c_pm^p + \cO(m^{p-1}) \right) \E\left[ R_1^{2p}\right] = \cO\left(m^{-p}\right).
	\end{equation*}
	The proof is complete.
\end{proof}

\begin{proof}[Proof of Theorem~\ref{thm:Lm_moment2p}]
	Let $L_m(X)-L(X)=\avgmj R_j$ where $R_j=H(X, Y_j)-L(X)$ for $j=1,...,m$ and $\cG = \sigma(X)$ then it suffices to verify that the conditions of Lemma~\ref{lem:avg_moment2p} hold.
	
	Firstly, since $Y_j$ are i.i.d. so $R_j$'s are identically distributed and are conditional independent given $X$.	
	Moreover, by Equation~\eqref{eq:LtauLR} we have $\E\left[H(X, Y_j)|X\right] = L(X)$ so $\E\left[R_j|\cG\right]=0$ for $j=1,\ldots,m$.
	Lastly, the $2p$-moment of $R_1$ is bounded because
	\begin{equation*}
	\E\left[|R_1|^{2p}\right] = \E\left[\left(H(X, Y_1)-\E\left[H(X, Y_1)|X\right]\right)^{2p}\right]\stackrel{(*)}{\leq} 4^{p}\E\left[\left|H(X, Y_1)\right|^{2p}\right]<\infty,
	\end{equation*}
	where the inequality $(*)$ holds due to Lemma~\ref{lem:rv_moment2p} with $R=H(X, Y_1)$, and $\cG=\sigma(X)$.
	The proof is complete.
\end{proof}

\section{Supplementary details for asymptotic bias, variance, and MSE}\label{app:MSE}

A few special instances of Cauchy-Schwartz's inequalities are frequently used in our analysis, so we summarize them in Lemma~\ref{lem:CauchyIneq} for ease of reference.
\begin{lemma}\label{lem:CauchyIneq}
	For all vectors $\bm{x}$ and $\bm{y}$ of an inner product space, Cauchy-Schwartz's inequality asserts that $|\left\langle \bm{x},\bm{y}\right\rangle |^2 \leq \left\langle \bm{x},\bm{x}\right\rangle\cdot \left\langle \bm{y},\bm{y}\right\rangle$, where $\left\langle \cdot,\cdot\right\rangle$ is the inner product.
	In particular, if $\bm{x}=(x_1,\ldots,x_n)$ and $\bm{y}$ is a vector of ones with compatible dimension, then
	\begin{equation}\label{eq:CauthySchwarzIneq1}
	\left(\sum_{i=1}^n x_i\right)^2 \leq n\sum_{i=1}^n x_i^2.
	\end{equation}
	Also, if $X, X_1,\ldots,X_n$ are identically distributed random variables, then
	\begin{equation}\label{eq:CauthySchwarzIneq2}
	\E\left[\left(\avgni X_i\right)^2\right] = \frac{1}{n^2} \E\left[\left(\sum_{i=1}^n X_i\right)^2\right] \leq \frac{1}{n} \left(\sum_{i=1}^n\E\left[ X_i^2\right]\right) = \E[X^2].
	\end{equation}
	Moreover, define the inner product of two arbitrary random variables $X$ and $Y$ as the expectation of their product, then
	\begin{equation}\label{eq:CauthySchwarzIneq3}
	\E\left[|XY|\right] \leq \left(\E\left[|X|^2\right]\right)^{1/2}\left(\E\left[|Y|^2\right]\right)^{1/2}.
	\end{equation}
	Lastly,~\eqref{eq:CauthySchwarzIneq3} implies that the following inequality holds for arbitrary random variables $X$ and $Y$,
	\begin{align}
		&\E\left[X^2-Y^2\right] \leq \E\left[\left|X^2-Y^2\right|\right] = \E\left[\left|(X-Y)^2 + 2Y(X-Y)\right|\right]\nonumber\\
		  \leq &\E\left[(X-Y)^2\right] + 2\left(\E\left[Y^2\right]\right)^{1/2}\left(\E\left[(X-Y)^2\right]\right)^{1/2}.\label{eq:CauthySchwarzIneq4}
	\end{align}
\end{lemma}



Proposition~\ref{prop:bias}, Proposition~\ref{prop:var}, and Theorem~\ref{thm:MSE} are analyzed in Sections~\ref{subsec:Analysis_Bias},~\ref{subsec:Analysis_variance}, and~\ref{subsec:Analysis_MSE}, respectively.
This section provide additional details to unproven parts of the above results, such as proving Lemma~\ref{lem:usefuleqs} and asymptotic variance for the indicator risk function.

\begin{proof}[Proof of Lemma~\ref{lem:usefuleqs}]
	For Equation~\eqref{eq:usefuleq1}, note that
\begin{align}
	\E\left[\1\{L_m\geq 0\} - \1\{L\geq 0\}\right] &= \int\int_{-z/\sqrt{m}}^{\infty} p_m(\ell,z)\d\ell\d z - \int\int_{0}^{\infty} p_m(\ell,z)\d\ell\d z\nonumber\\
	&= \int\int_{-z/\sqrt{m}}^{0} p_m(\ell,z)\d\ell\d z \nonumber\\
	&\stackrel{(*)}{=} \int\int_{-z/\sqrt{m}}^{0} \left[p_m(0,z) + \ell\cdot\frac{\partial}{\partial \ell} p_m(u_\ell,z)\right] \d\ell\d z \nonumber\\
	&= \int \frac{z}{\sqrt{m}}p_m(0,z)dz + \int\int_{-z/\sqrt{m}}^{0} \ell \frac{\partial}{\partial \ell} p_m(u_\ell,z)\d\ell\d z.\label{eq:auxeq1}
\end{align}
	where $(*)$ holds by Assumption~\ref{assm:jointdensity}.
	The first term in~\eqref{eq:auxeq1} can be written as $\frac{\widetilde{p}(\ell)}{\sqrt{m}}\E[Z_m|L=0]$, which equals 0 because, by Proposition~\ref{prop:Lm_as},
	\[
	\frac{1}{\sqrt{m}}\E[Z_m|L=0] = \E[\E[L_m(X)-L(X)|X]|L(X)=0] = \E[L(X)-L(X)|L(X)=0]= 0.
	\]
	The second term of~\eqref{eq:auxeq1} is of order $\cO(m^{-1})$ because, by Assumption~\ref{assm:jointdensity}~\ref{assm:jointdensity3}, it is bounded by
	\[
 \int\int_{-z/\sqrt{m}}^{0} |\ell|\cdot\bar{p}_{1,m}(z)\d\ell\d z = \frac{1}{2m}\int z^2 \bar{p}_{1,m}(z)dz = \cO(m^{-1}).
	\]

	For Equation~\eqref{eq:usefuleq3}, note that
	\begin{align*}
		&\E\left[|L_m\cdot(\1\{L_m\geq 0\} - \1\{L\geq 0\})|\right]\\
		\leq& \E\left[|L_m|\cdot\1\{L_m\geq 0 > L\}\right] +\E\left[|L_m|\cdot\1\{L\geq 0 > L_m\}\right]\\
=&\int^{\infty}_{0}\int_{-z/\sqrt{m}}^{0} \left|\ell+\frac{z}{\sqrt{m}}\right| p_m(\ell,z)\d\ell\d z + \int_{-\infty}^{0}\int^{-z/\sqrt{m}}_{0} \left|\ell+\frac{z}{\sqrt{m}}\right| p_m(\ell,z)\d\ell\d z\\
\leq&\int^{\infty}_{0}\int_{-z/\sqrt{m}}^{0} \left(|\ell|+\frac{|z|}{\sqrt{m}}\right) \bar{p}_{0,m}(z)\d\ell\d z +\int_{-\infty}^{0}\int^{-z/\sqrt{m}}_{0} \left(|\ell|+\frac{|z|}{\sqrt{m}}\right) \bar{p}_{0,m}(z)\d\ell\d z\\
=&\int^{\infty}_{0}\left(\frac{z^2}{2m} + \frac{z^2}{m}\right) \bar{p}_{0,m}(z)\d z +\int_{-\infty}^{0}\left(\frac{z^2}{2m} + \frac{z^2}{m}\right) \bar{p}_{0,m}(z)\d z\\
		=&\frac{3}{2m}\int z^2 \bar{p}_{0,m}(z)\d z= \cO(m^{-1}),
	\end{align*}
	where the last equality holds by Assumption~\ref{assm:jointdensity}~\ref{assm:jointdensity3}.
	The proof is complete.
\end{proof}

\begin{proof}[Proof of Proposition~\ref{prop:var}]

Discussions in Section~\ref{subsec:Analysis_variance} assert Proposition~\ref{prop:var} for smooth and hockey-stick risk functions.
For the indicator risk function, it remains to prove that the first term in~\eqref{eq:varbound} is of order $\cO(m^{-1}) + \cO(n^{-1})$.

Note that $L_{m,i}$, $i=1,\ldots,n$ are identically distributed (so are $\Li$, $i=1,\ldots,n$), then
\begin{align}
&\E\left[\left(\avgni \left(g\left(\Lmi\right) - g\left(\Li\right)\right)\right)^2 \right]\nonumber\\
=&\frac{1}{n^2}\E\left[\sum_{i=1}^n\left(g\left(\Lmi\right) - g\left(\Li\right)\right)^2 + \sum_{i=1}^n\sum_{\substack{k=1\\ k\neq i}}^n\left(g\left(\Lmi\right) - g\left(\Li\right)\right)\left(g\left(L_{m,k}\right) - g\left(L_{k}\right)\right)\right]\nonumber\\
=&\frac{1}{n}\E\left[\left(g\left(L_{m,1}\right) - g\left(L_{1}\right)\right)^2 \right]+\frac{n-1}{n}\E\left[\left(g\left(L_{m,1}\right) - g\left(L_1\right)\right)\left(g\left(L_{m,2}\right) - g\left(L_2\right)\right)\right]\nonumber\\
\leq&\frac{1}{n}+\frac{n-1}{n}\E\left[\left(g\left(L_{m,1}\right) - g\left(L_{1}\right)\right)\left(g\left(L_{m,2}\right) - g\left(L_{2}\right)\right)\right], \label{VarIndicator01}
\end{align}
where the inequality holds because $g(x)=\1\{x\geq 0\}\leq 1$ and so $(g(x)-g(y))^2 \leq 1$.
The first term in~\eqref{VarIndicator01} is of order $\cO(n^{-1})$.
For the second term in~\eqref{VarIndicator01}, note that
\begin{align}
&\E\left[\left(g\left(L_{m,1}\right) - g\left(L_1\right)\right)\left(g\left(L_{m,2}\right) - g\left(L_2\right)\right)\right]\nonumber\\
=&\E\left[\left(\1\{L_{m,1}\geq 0 > L_1\} - \1\{L_1\geq 0>L_{m,1}\}\right)\left(\1\{L_{m,2}\geq 0 > L_2\} - \1\{L_2 \geq 0 > L_{m,2}\}\right)\right]\nonumber\\
=&\P\left(L_{m,1} \geq 0 > L_1,L_{m,2} \geq 0 > L_2 \right) - \P\left(L_1\geq 0 > L_{m,1},L_{m,2} \geq 0 >L_2\right) \nonumber\\
& - \P\left(L_{m,1} \geq 0 > L_1, L_2 \geq 0 > L_{m,2}\right) + \P\left(L_1 \geq 0 >L_{m,1},L_2 \geq 0 > L_{m,2}\right).\label{eq:auxeq2}
\end{align}
We examine the convergence rate of the first term in~\eqref{eq:auxeq2}, which is common for all four terms.
By Assumption~\ref{assm:indicatordensity2}, we can apply the Taylor's theorem to the joint density $q_m(\ell_1,\ell_2,z_1,z_2)$ so
\begin{align}
	q_m(\ell_1,\ell_2,z_1,z_2)&=q_m(0,0,z_1,z_2) + \ell_1\frac{\partial}{\partial \ell_1}q_m(\bar{\ell}_1,\bar{\ell}_2,z_1,z_2)+ \ell_2\frac{\partial}{\partial \ell_2}q_m(\bar{\ell}_1,\bar{\ell}_2,z_1,z_2)\label{indVar Taylor q}\\
	&\leq q_m(0,0,z_1,z_2) + (|\ell_1|+|\ell_2|)\cdot \bar{q}_{1,m}(z_1,z_2)\label{indVar Taylor q_ineq},
\end{align}
where $\bar{\ell}_1\in(\ell_1,0)$, $\bar{\ell}_2\in(\ell_2,0)$, and the inequality holds by Assumption~\ref{assm:indicatordensity2}~\ref{assm:jointdensityB2}.
Then we have
\begin{align*}
	&\P\left(L_{m,1}\geq0>L_{1},L_{m,2}\geq0>L_{2}\right)\\
	=&\P\left(L_{1}+Z_{m,1}/\sqrt{m}\geq0>L_{1},L_{2}+Z_{m,2}/\sqrt{m}\geq0>L_{2}\right)\\
	=&\int_0^{\infty}\int_0^{\infty}\int_{-\frac{z_1}{\sqrt{m}}}^0\int_{-\frac{z_2}{\sqrt{m}}}^0 q_m(\ell_1,\ell_2,z_1,z_2)\d\ell_1 \d\ell_2 \d z_1 \d z_2\\
\stackrel{\eqref{indVar Taylor q_ineq}}{\leq}&\frac{1}{m}\int_0^{\infty}\int_0^{\infty} z_1z_2 \bar{q}_{0,m}(z_1,z_2)\d z_1 \d z_2 + \int_0^{\infty}\int_0^{\infty}\int_{-\frac{z_1}{\sqrt{m}}}^0\int_{-\frac{z_2}{\sqrt{m}}}^0 (|\ell_1| + |\ell_2|)\bar{q}_{1,m}(z_1,z_2)\d\ell_1 \d\ell_2 \d z_1 \d z_2\\
=&\cO(m^{-1}) - \frac{1}{2m^{3/2}}\int_0^{\infty}\int_0^{\infty}(z_1^2z_2+z_1z_2^2)\bar{q}_{1,m}(z_1,z_2)\d z_1 \d z_2\\
=&\cO(m^{-1}) + \cO(m^{-3/2})=\cO(m^{-1}).
\end{align*}
This means that the first term in~\eqref{eq:auxeq2}, and indeed all four terms, converge at the rate $\cO(m^{-1})$.
So~\eqref{VarIndicator01} is of order $\cO(m^{-1})+\cO(n^{-1})$.
Combining this with the latter two terms in~\eqref{eq:varbound}, which are of order $\cO(n^{-1})$ and $\cO(m^{-2})$, we see that $\Var[\rho_{mn}]=\cO(m^{-1})+\cO(n^{-1})$, as desired.
\end{proof}

\section{Proof for Theorem~\ref{thm:varianceestimate}}\label{app:CLT}

We will use a few lemmas below to help prove Theorem~\ref{thm:varianceestimate}.
Specifically, Lemmas~\ref{lem:aux1} and~\ref{lem:sig1} show that $\widehat{\sigma}_{1,mn}^2\conprob\sigma_{1}^2$ and Lemmas~\ref{lem:aux2} and~\ref{lem:sig2} show that $\widehat{\sigma}_{2,mn}^2\conprob\sigma_{2}^2$.
Then $\widehat{\sigma}_{mn}^2/\sigma_{mn}^2$ converges to 1 in probability by the continuous mapping theorem.

\begin{lemma}\label{lem:aux1}
	Suppose the conditions for Theorem~\ref{thm:CLT} hold, then the following convergences hold for any positive integer $n$,
	\begin{align}
		&\avgni \left[g(\Lmi)-g(\Li)\right] \conlone 0 \mbox{ as } m\to \infty, \label{eq:aux1}\\
		&\avgni \left[(g(\Lmi))^2-(g(\Li))^2\right] \conlone 0 \mbox{ as } m\to \infty.\label{eq:aux2}
	\end{align}
\end{lemma}

\begin{proof}[Proof of Lemma~\ref{lem:aux1}]
	Recall that $\Lmi$, $i=1,\ldots,n$ are identically distributed, and so are $\Li$, $i=1,\ldots,n$.
	Then
	\begin{align*}
	\E\left[\left|\avgni \left[g(\Lmi)-g(\Li)\right]\right|\right]	\leq \E[|g(L_m)-g(L)|]
	\stackrel{\eqref{eq:CauthySchwarzIneq3}}{\leq} \left(\E[(g(L_m)-g(L))^2]\right)^{1/2} \stackrel{(*)}{=} \cO(m^{-1/2})
	\end{align*}
	where $(*)$ holds by Equations~\eqref{eq:varsmooth} and~\eqref{eq:varhockeystick}.
	This means that~\eqref{eq:aux1} holds.

	Moreover,
	\begin{align*}
	&\E\left[\left|\avgni \left[(g(\Lmi))^2-(g(\Li))^2\right]\right|\right]\leq \E[|(g(L_m))^2-(g(L))^2|]\\
	\stackrel{\eqref{eq:CauthySchwarzIneq4}}{\leq} &\E\left[\left(g(L_m)-g(L)\right)^2\right] + 2\left(\E\left[(g(L))^2\right]\right)^{1/2}\left(\E\left[\left(g(L_m)-g(L)\right)^2\right]\right)^{1/2}\\
	=&\cO(m^{-1}) + \cO(m^{-1/2})=\cO(m^{-1/2}),
	\end{align*}
	where the last equality holds due to Equations~\eqref{eq:varsmooth} and~\eqref{eq:varhockeystick}.
	This means that~\eqref{eq:aux2} holds.
	The proof is complete.
\end{proof}

\begin{lemma}\label{lem:sig1}
	If the conditions for Theorem~\ref{thm:CLT} hold, then $\widehat{\sigma}_{1,mn}^2\conprob\sigma_{1}^2$ as $\min\{m,n\}\to 0$.
\end{lemma}
\begin{proof}[Proof of Lemma~\ref{lem:sig1}]
	By~\eqref{eq:asymvar1} and~\eqref{eq:sig1hat}, we have
	\begin{equation}\label{eq:diff_sig1hat}
		\widehat{\sigma}_{1,mn}^2 - \sigma_{1}^2 = \left[\avgni (g(\Lmi))^2 - \E\left[(g\left(L\right))^2\right]\right] + \left[\left(\avgni g(\Lmi)\right)^2 - \left(\E\left[g\left(L\right)\right]\right)^2\right].
	\end{equation}
	We then show that both terms on the RHS converge to 0 in probability as $\min\{m,n\}\to \infty$.
	
	

	For the first term in~\eqref{eq:diff_sig1hat}, note that
	\begin{align}
	&\avgni (g(\Lmi))^2 - \E\left[(g\left(L\right))^2\right] \nonumber\\
	= &\avgni \left[(g(\Lmi))^2-(g(\Li))^2\right] +\left(\avgni (g(\Li))^2- \E\left[(g\left(L\right))^2\right]\right).\label{eq:diff1}
	\end{align}
	The first term in~\eqref{eq:diff1} converges to 0 in probability by~\eqref{eq:aux1} in Lemma~\ref{lem:aux1}.
	The second term in~\eqref{eq:diff1} converges to 0 in probability to zero as $n\to\infty$ by the weak law of large numbers because $(g(\Li))^2$, $i=1,\ldots,n$ are i.i.d. samples with the common expectation $\E[(g\left(L\right))^2]$.
 	
	For the second term in~\eqref{eq:diff_sig1hat}, by the continuous mapping theorem it suffices to show that $\avgni g(\Lmi) \conprob \E\left[g\left(L\right)\right]$.
	Note that
	\begin{align}
		\avgni g(\Lmi) - \E\left[g\left(L\right)\right] =\avgni \left[g(\Lmi)-g(\Li)\right] +\left(\avgni g(\Li)- \E\left[g\left(L\right)\right]\right).\label{eq:diff2}
	\end{align}
	The first term in the RHS of~\eqref{eq:diff2} converges to 0 in probability by~\eqref{eq:aux2} in Lemma~\ref{lem:aux1}.
	The second term in the RHS of~\eqref{eq:diff2} converges to 0 in probability as $n\to\infty$ by weak law of large numbers because $g(\Li)$, $i=1,\ldots,n$ are i.i.d. samples with the common expectation $\E[g\left(L\right)]$.
	Therefore by the Slutsky's theorem we have $\avgni g(\Lmi) \conprob \E\left[g\left(L\right)\right]$, as desired.
	

	In summary, both terms in~\eqref{eq:diff_sig1hat} converges to 0 in probability, as desired.
	The proof is complete.
\end{proof}

The next two lemmas show $\widehat{\sigma}_{2,mn}^2\conprob\sigma_{2}^2$.
We define new notations for the convenience to state and prove the lemmas:
For any $j=1,\ldots,m$,
\begin{equation}\label{eq:notations}
R_{j} := \E[g'(L)\hatH|Y=Y_j],\
\hatR_{j} := \avgni g'(\Li)\hatHij, \mbox{ and }
\hatR_{m,j} := \avgni g'(\Lmi)\hatHij.
\end{equation}
Note that $\{R_j,j=1,\ldots,m\}$, are identically distributed, and so are $\{\hatR_{j}, j=1,\ldots,m\}$ and $\{\hatR_{m,j},j=1,\ldots,m\}$.
When no confusion arises, the subscript $j$ is omitted to denote a generic index $j=1,\ldots,m$.

\begin{lemma}\label{lem:aux2}
	If the conditions for Theorem~\ref{thm:CLT} hold, then the following convergences hold:
	\begin{align}
		&\avgmj \left[\hatR_{m,j}^2-\hatR_{j}^2\right] \conlone 0, \label{eq:aux4}\\
		&\avgmj \left[\hatR_{j}^2-R_{j}^2\right] \conlone 0,\label{eq:aux5} \\
	& \avgni \left[g'(\Lmi)\Lmi - g'(\Li)\Li\right] \conlone 0.\label{eq:aux6}
	\end{align}
\end{lemma}

\begin{proof}[Proof of Lemma~\ref{lem:aux2}]
	
	Firstly, note that because $\hatR_{m,j}$, $j=1,\ldots,m$, are identically distributed (so are $\hatR_{j}$, $j=1,\ldots,m$), we have
	\begin{align}
		&\E\left[\left|\avgmj \left[\hatR_{m,j}^2-\hatR_{j}^2\right]\right|\right]\nonumber\leq \E\left[\left|\hatR_{mn}^2-\hatR_{n}^2\right|\right]\nonumber\\
		\stackrel{\eqref{eq:CauthySchwarzIneq4}}{\leq}& \E[\left(\hatR_{mn} - \hatR_n\right)^2] + 2\left(\E\left[\hatR_n^2\right]\right)^{1/2}\left(\E[(\hatR_{mn} - \hatR_n)^2]\right)^{1/2}\nonumber\\
		\stackrel{\eqref{eq:CauthySchwarzIneq2}}{\leq}& \E\left[\left(g'(L_m)\hatH - g'(L)\hatH\right)^2\right] + 2\left(\E\left[\left(g'(L)\hatH\right)^2\right]\right)^{1/2}\left(\E\left[\left(g'(L_m)\hatH - g'(L)\hatH\right)^2\right]\right)^{1/2}\label{eq:aux7}
	\end{align}
	
	\begin{itemize}
		\item For smooth risk functions that have bounded second derivative $|g''(x)|\leq C_g<\infty$, by the Taylor's theorem we have
		\begin{align*}
			&\E[((g'(L_m)-g'(L))\hatH)^2] = \E[(g''(\Lambda_m)(L_m-L)\hatH)^2] \\
			\stackrel{\eqref{eq:CauthySchwarzIneq3}}{\leq}& C_g^2 (\E\left[(L_m-L)^4\right])^{1/2}\left(\E\left[\hatH^4\right]\right)^{1/2} \stackrel{(*)}{=} \cO(m^{-1}),
		\end{align*}
		where $\Lambda_m$ is a random variable between $L$ and $L_m$ and $(*)$ holds because $\E\left[(L_m-L)^4\right] = \cO(m^{-2})$ by Theorem~\ref{thm:Lm_moment2p} with $p=2$ and $\E[\hatH^4]<\infty$ by assumption.
		Also, because $\E\left[\left(g'(L)\right)^4\right]<\infty$ and $\E\left[\hatH^4\right]<\infty$ by assumption, we have
		\begin{align*}
			\E\left[\left(g'(L)\hatH\right)^2\right] \leq \left(\E\left[\left(g'(L)\right)^4\right]\right)^{1/2}\left(\E\left[\hatH^4\right]\right)^{1/2} < \infty.
		\end{align*}
		Therefore~\eqref{eq:aux7} is of order $\cO(m^{-1})$ so it converges to zero as $m\to\infty$ for smooth risk functions.
	
		\item For the hockey-stick risk function $g(x)=\max\{x,0\}$ with $g'(x)=\1\{x\geq 0\}\leq 1 <\infty$, so $\E\left[g'(L_m)\hatH\right]\leq \E\left[\hatH\right] <0$ as $ \E\left[\hatH^2\right] <0$.
		So, by the dominated convergence theorem,
		\begin{align*}
		\lim\limits_{m\to\infty}	\E\left[\left(g'(L_m)\hatH - g'(L)\hatH\right)^2\right] &= \lim\limits_{m\to\infty}\E[((\1\{L_m\geq 0\}-\1\{L\geq 0\})\hatH)^2]\\
			&=\E\left[\lim\limits_{m\to\infty}((\1\{L_m\geq 0\}-\1\{L\geq 0\})\hatH)^2\right] = 0,
		\end{align*}
		where the last equality holds because $L_m\stackrel{a.s}{\to}L$ as $m\to\infty$ according to Proposition~\ref{prop:Lm_as}.
		Also, $\E\left[\left(g'(L)\hatH\right)^2\right]\leq \E\left[\hatH^2\right] < \infty$ where the finiteness holds by assumption.
		Therefore~\eqref{eq:aux7} converges to zero as $m\to\infty$ for the hockey-stick risk function.
	\end{itemize}
	In summary, we have shown that $\E\left[\left|\avgmj \left[\hatR_{m,j}^2-\hatR_{j}^2\right]\right|\right]\to 0$ as $m\to\infty$, which proves the $\cL^1$-convergence in~\eqref{eq:aux4}.
	
	
	Secondly, because $\hatR_{j}$, $j=1,\ldots,m$ are identically distributed (so are $R_j$, $j=1,\ldots,m$), therefore
	\begin{align}
		&\E\left[\left|\avgmj \left[\hatR_{j}^2-R_{j}^2\right]\right|\right] \leq \E\left[\left|\hatR_{n,1}^2-R_1^2\right|\right]\nonumber\\
		\stackrel{\eqref{eq:CauthySchwarzIneq4}}{\leq}& \E\left[\left(\hatR_{n,1} - R_1\right)^2\right] + 2\left(\E\left[R_1^2\right]\right)^{1/2}\left(\E\left[\left(\hatR_{n,1} - R_1\right)^2\right]\right)^{1/2}. \label{eq:aux8}
	\end{align}
	We note that
	$\E\left[R^2\right]=\E\left[\left(\E[g'(L)\hatH|Y]\right)^2\right]\leq\E[\E[(g'(L)\hatH)^2|Y]]=\E[(g'(L)\hatH)^2]\stackrel{(*)}{<}\infty,$
	where $(*)$ holds by the respective assumptions for the smooth and hockey-stick risk functions.
	Next, define $R_{n,ij}^* = g'(\Li)\hatH_{ij} -\E[g'(L)\hatH|Y=Y_j]$ so $\hatR_{n,j} - R_j =\avgni R_{n,ij}^*$.
	Given $Y_j$, $R_{n,ij}^*$, $i=1,\ldots,n$, are conditionally independent and identically distributed with mean
	\begin{align*}
		\E\left[R_{n,ij}^*|Y_j\right] = \E\left[g'(\Li)\hatH_{ij}|Y_j\right] - \E\left[\E[g'(L)\hatH|Y=Y_j]|Y_j\right] = 0.
	\end{align*}
	In addition,
	\begin{align*}
		\E[(R_{n}^*)^2] = \Var\left[g'(L)\hatH|Y\right]=\E\left[\left(g'(L)\hatH\right)^2\right] - \left(\E\left[g'(L)\hatH|Y\right]\right)^2\leq \E\left[\left(g'(L)\hatH\right)^2\right]\stackrel{(*)}{<}\infty,
	\end{align*}
	where $(*)$ holds by the respective assumptions for the smooth and hockey-stick risk functions.
	Then, using Lemma~\ref{lem:avg_moment2p} with $p=2$ (use $R_{n}^*$ in that lemma), we have
	$\E\left[\left(\hatR_{n}-R\right)^2\right]=\frac{\E\left[(R_{n}^*)^2\right]}{n} = \cO(n^{-1}).$
	Therefore,~\eqref{eq:aux8} converges to zero at the rate $\cO(n^{-1}) + \cO(n^{-1/2})=\cO(n^{-1/2})$ as $n\to\infty$.
	This proves the $\cL^1$ convergence in~\eqref{eq:aux5}.

	Lastly, because $g'(\Lmi)\Lmi$, $i=1,\ldots,n$, are identically distributed (so are $g'(\Li)\Li$, $i=1,\ldots,n$), so
	\begin{align}
		&\E\left[\left|\avgni \left[g'(\Lmi)\Lmi - g'(\Li)\Li\right]\right|\right] \leq \E\left[\left|g'(L_m)L_m - g'(L)L\right|\right]\nonumber\\
		=&\E\left[\left|(g'(L_m)-g'(L))L_m - g'(L)(L-L_m)\right|\right]\nonumber\\
		\leq& \E[|(g'(L_m)-g'(L))L_m |] + \E[|g'(L)(L-L_m)|] \stackrel{(*)}{=} \cO(m^{-1/2})\label{eq:aux9}
	\end{align}
%
	where $(*)$ holds because of the following:
	\begin{itemize}
		\item For smooth functions $g$ with bounded second derivative, i.e., $|g''(x)|\leq C_g<\infty$,~\eqref{eq:aux9} equals
		\begin{align*}
		&\E[|g''(\Lambda_m)(L_m-L)(L_m-L+L) |] + \E[|g'(L)(L-L_m)|]\\
		\leq& C_g\left(\E\left[(L_m-L)^2\right] + \E[|(L_m-L)L|]\right) + \E[|g'(L)(L-L_m)|]\\
		\leq& C_g\left(\E\left[(L_m-L)^2\right] + (\E\left[(L_m-L)^2\right])^{1/2}(\E\left[L^2\right])^{1/2}\right) + (\E\left[(g'(L))^2\right])^{1/2}(\E\left[(L_m-L)^2\right])^{1/2}\\
		\stackrel{(*)}{=}& C_g\left(\cO(m^{-1}) + \cO(m^{-1/2})\right) + \cO(m^{-1/2}) = \cO(m^{-1/2}),
		\end{align*}
		where $(*)$ holds because $\E\left[(L_m-L)^2\right]=\cO(m^{-1})$ by Theorem~\ref{thm:Lm_moment2p} with $p=1$ and $\E\left[L^2\right]<\infty$ and $\E\left[(g'(L))^2\right]$ by assumptions.
		
		\item For the hockey-stick function, $g'(x)=\1\{x\geq 0\}$,~\eqref{eq:aux9} equals
		\begin{align*}
		&\E[|L_m\cdot(\1\{L_m\geq 0\}-\1\{L\geq 0\})|] + \E[|1\{L \geq 0\}(L_m-L)|]\\
		\leq & \E[|L_m\cdot(\1\{L_m\geq 0\}-\1\{L\geq 0\})|]+ (\E\left[(L_m-L)^2\right])^{1/2} = \cO(m^{-1}) + \cO(m^{-1/2}),
		\end{align*}
	where the last equality holds by~\eqref{eq:usefuleq3} in Lemma~\ref{lem:usefuleqs} and Theorem~\ref{thm:Lm_moment2p} with $p=2$.
	\end{itemize}
	In short, we have shown that~\eqref{eq:aux9}$\to 0$ as $\min\{m,n\}\to\infty$, which proves the $\cL^1$-convergence in~\eqref{eq:aux6}.
	The proof is complete.
\end{proof}

\begin{lemma}\label{lem:sig2}
	If the conditions for Theorem~\ref{thm:CLT} hold, then $\widehat{\sigma}_{2,mn}^2\conprob\sigma_{2}^2$ as $\min\{m,n\}\to 0$.
\end{lemma}
\begin{proof}[Proof of Lemma~\ref{lem:sig2}]
	
	By Equations~\eqref{eq:asymvar2}, \eqref{eq:sig2hat} and the notations in~\eqref{eq:notations}, we have
	\begin{equation}\label{eq:diff_sig2hat}
		\widehat{\sigma}_{2,mn}^2 - \sigma_{2}^2 = \left[\avgmj \hatR_{m,j}^2 - \E\left[R^2\right]\right] + \left[\left(\avgni g'(\Lmi)\Lmi\right)^2 - \left(\E\left[g'\left(L\right)L\right]\right)^2\right].
	\end{equation}
	We then consider each of the two differences above and show that both converge to zero in probability as $\min\{m,n\}\to \infty$.
	For the first term in~\eqref{eq:diff_sig2hat}, note that
	\begin{align}
	\avgmj \hatR_{m,j}^2 - \E\left[R^2\right] =& \avgmj \left[\hatR_{m,j}^2-\hatR_{j}^2\right] +\avgmj \left[\hatR_{j}^2-R_{j}^2\right] +\avgmj R_{j}^2- \E\left[R^2\right].\label{eq:diff6}
	\end{align}
	By~\eqref{eq:aux4} and~\eqref{eq:aux5} in Lemma~\ref{lem:aux2}, the first two terms on the RHS of~\eqref{eq:diff6} converge, in $\cL^1$ and hence in probability, to zero as $\min{m,n}\to\infty$.
	Also, because $R_{j}^2$, $j=1,\ldots,m$ are i.i.d. samples of $R^2$, so the last term converges to zero as $m\to\infty$ in probability by the weak law of large numbers.
	
	For the second term in~\eqref{eq:diff_sig2hat}, note that
	\begin{align}
	&\avgni g'(\Lmi)\Lmi - \E\left[g'\left(L\right)L\right]\nonumber\\
	=&\avgni \left[g'(\Lmi)\Lmi - g'(\Li)\Li\right] + \left[\avgni g'(\Li)\Li- \E\left[g'\left(L\right)L\right]\right].\label{eq:diff8}
	\end{align}
	The first term on the RHS of~\eqref{eq:diff8} converges in probability to zero as $m\to\infty$ by the $\cL^1$ convergence~\eqref{eq:aux6} in Lemma~\ref{lem:aux2}.
	The second term on the RHS of~\eqref{eq:diff8} converges in probability to zero as $n\to\infty$ by weak law of large numbers because $g'(\Li)\Li$, $i=1,\ldots,n$ are i.i.d. samples with the common expectation $\E\left[g'\left(L\right)L\right]$.
	Therefore $\avgni g'(\Lmi)\Lmi \conprob \E\left[g'\left(L\right)L\right]$ and so $\left(\avgni g'(\Lmi)\Lmi\right)^2 \conprob \left(\E\left[g'\left(L\right)L\right]\right)^2$ by the continuous mapping theorem.
	
	In summary, both terms in~\eqref{eq:diff_sig2hat} converge to 0 in probability, as desired.
	The proof is complete.
\end{proof}

\section{Proofs for results in Section~\ref{subsubsec:CLTindicator}}\label{app:CLTindicator}

Consider the decomposition~\eqref{eq:decomposerho_indicator}, in this appendix we will show that $\widetilde{\sigma}_{mn}^{-1}(\cU_{\epsilon_m,mn}-\rho)\condist \cN(0,1)$ and all $r_{\epsilon_m,mn}^a$, $r_{\epsilon_m,mn}^b$ and $r_{\epsilon_m,mn}^c$ converges to zero quickly in Lemmas~\ref{lem indLCT U},~\ref{lem indLCT ra},~\ref{lem indLCT rb} and~\ref{lem indLCT rc}, respectively.
We omit the lengthy discussions on the technical assumptions needed to ensure that the remainder term in $r_{\epsilon_m,mn}^d$ is negligible and focus on analyzing the other terms.
Then, applying the Slutsky's theorem to the decomposition~\eqref{eq:decomposerho_indicator}, we have that $\widetilde{\sigma}_{mn}^{-1}(\rho_{mn}-\rho)\condist \cN(0,1)$ as so the proof for Theorem~\ref{thm:CLTIndicator} is complete.

In addition, in Lemmas~\ref{lem indLCT sig1} and~\ref{lem indLCT sig2}, we show that $\widehat{\widetilde{\sigma}}_{1,mn}^2$ and $\widehat{\widetilde{\sigma}}_{2,mn}^2$ converge to $\sigma_1^2$ and $\widetilde{\sigma}_2^2$, respectively.
Then, applying the continuous mapping theorem to $\widehat{\widetilde{\sigma}}_{mn}^2 = \frac{\widehat{\widetilde{\sigma}}_{1,mn}^2}{n} + \frac{\widehat{\widetilde{\sigma}}_{2,mn}^2}{m}$, the proof for Theorem~\ref{thm:varianceestimateIndicator} is complete.

Before proceeding, we recall the function $\gepsm(x) = \int_{-\infty}^{x/\epsilon_m} \phi(u) du$ where $\phi(u) = \frac{1}{4\pi} (1-\cos(u))\cdot \1\{|u| \leq 2\pi\}$, as defined in~\eqref{eq:smoothapprox}.
Then, by construction, $\gepsm'(x) = \frac{1}{4\pi \epsilon_m} \left(1-\cos\left(x/\epsilon_m\right)\right)\cdot \1\left\lbrace\left|x\right| \leq 2\pi\epsilon_m\right\rbrace$, $\gepsm''(x)= \frac{1}{4\pi \epsilon_m^2} \sin \left( x/\epsilon_m\right)\cdot \1\left\lbrace\left|x\right| \leq 2\pi\epsilon_m\right\rbrace$,
\begin{align}
	\int_{-\infty}^{\infty} \phi(u) du = \int_{-2\pi}^{2\pi} \frac{1}{4\pi} (1-\cos(u)) =1,\,\int_{-\infty}^{\infty} u \cdot \phi(u) du =0,\mbox{ and }\label{eq:aux12}\\
	\int_{-\infty}^{\infty} |u|^{r_1} \cdot [\phi(u)]^{r_2} du < \infty, \,\, r_1=0,1,2,3, r_2=1,2.\label{eq:aux13}
\end{align}

\begin{lemma}\label{lem indLCT U}
	Suppose the conditions for Theorem~\ref{thm:CLTIndicator} hold.
	Then,
	$$\frac{\cU_{\epsilon_m,mn}-\rho}{\widetilde{\sigma}_{mn}}\condist \cN(0,1), \mbox{ as } \min\{m,n\}\to\infty,$$
	where $\widetilde{\sigma}_{mn}^2 = \frac{\widetilde{\sigma}_1^2}{n} + \frac{\widetilde{\sigma}_2^2}{m}$ and $\widetilde{\sigma}_{1}^2$ and $\widetilde{\sigma}_2^2$ are defined as~\eqref{eq:IndVar1} and~\eqref{eq:IndVar2}, respectively.
\end{lemma}
\begin{proof}[Proof of Lemma~\ref{lem indLCT U}]
Let $\cV_{\epsilon_m,ij} = g'_{\epsilon_m}(\Li)(\hatHij-\Li)$, then $\cU_{\epsilon_m,mn}=\frac{1}{mn}\sum_{i=1}^n\sum_{j=1}^m [g(\Li)+\cV_{\epsilon_m,ij}]$.
Note that $\cV_{\epsilon_m,ij}$ are identically distributed for all $i=1,\ldots,n$ and $j=1,\ldots,m$ so we can write a generic $\cV_{\epsilon_m, ij}$ simply as $\cV_{\epsilon_m}$ for notational convenience.
For any $\epsilon_m>0$, we have that
\begin{align}\label{Lambda}
\E[\cV_{\epsilon_m}|X] = g'_{\epsilon_m}(L(X))\left(\E[\hatH(X,Y)|X]-L(X)\right)= g'_{\epsilon_m}(L(X))\left(L(X)-L(X)\right) = 0,
\end{align}
which also means that $\E[\cV_{\epsilon_m}] = \E[\E[\cV_{\epsilon_m}|X]]=0$.
Moreover, we see that $\E[\cU_{\epsilon_m,mn}]=\E[g(L(X))] +\E[\E[\cV_{\epsilon_m}|X]] =\E[g(L(X))] + 0= \rho$, i.e., $\cU_{\epsilon_m,mn}$ is an unbiased estimator of $\rho$.

Consider the following random variables (Hoeffding decomposition):
\begin{equation*}\label{eq:Hoeffding}
	\widetilde{\cU}_{\epsilon_m,mn} =\widetilde{\cU}_{\epsilon_m,n} + \widetilde{\cU}_{\epsilon_m,m}:=\sum_{i=1}^n \E[\cU_{\epsilon_m,mn}-\rho|X_i] + \sum_{j=1}^m \E[\cU_{\epsilon_m,mn}-\rho|Y_j] .
\end{equation*}
We then consider the following decomposition:
\begin{equation}\label{eq:decomposition1}
	\frac{\cU_{\epsilon_m,mn}-\rho}{\widetilde{\sigma}_{mn}} = \frac{\widetilde{\cU}_{\epsilon_m,mn}}{\widetilde{\sigma}_{mn}} + \frac{\cU_{\epsilon_m,mn}-\rho-\widetilde{\cU}_{\epsilon_m,mn}}{\widetilde{\sigma}_{mn}}.
\end{equation}
To establish $\widetilde{\sigma}_{mn}^{-1}(\cU_{\epsilon_m,mn}-\rho)\condist \cN(0,1)$ it suffices to show that $\widetilde{\sigma}_{mn}^{-1}\widetilde{\cU}_{\epsilon_m,mn}\condist \cN(0,1)$ and that $\widetilde{\sigma}_{mn}^{-1}(\cU_{\epsilon_m,mn}-\rho-\widetilde{\cU}_{\epsilon_m,mn})\stackrel{d}{\to} 0$.

To show $\widetilde{\sigma}_{mn}^{-1}\widetilde{\cU}_{\epsilon_m,mn}\condist \cN(0,1)$, we consider the convergences of $\widetilde{\cU}_{\epsilon_m,n}$ and $\widetilde{\cU}_{\epsilon_m,m}$ separately.
Firstly, for any $i=1,\ldots,n$, because $X_k$ is independent of $X_i$ for any $k\neq i$, we have
\begin{align}
	&\E[\cU_{\epsilon_m,mn}-\rho|X_i] = \E\left[\left.\frac{1}{mn}\sum_{i=1}^n\sum_{j=1}^m g(\Li)\right|X_i\right] + \E\left[\left.\frac{1}{mn}\sum_{i=1}^n\sum_{j=1}^m \cV_{\epsilon_m,ij}\right|X_i\right] - \rho\nonumber\\
	\stackrel{\eqref{Lambda}}{=}& \frac{1}{n}\left(g(\Li)+\sum_{\substack{k=1 \\ k\neq i}}^n\E\left[\left. g(L_k)\right|X_i\right]\right) + 0 - \rho= \frac{1}{n}g(\Li) + \frac{n-1}{n}\rho - \rho \nonumber\\
	=& \frac{1}{n} g(\Li) - \frac{1}{n}\rho.\label{eq:aux14}
\end{align}
Therefore $\widetilde{\cU}_{\epsilon_m,n}=\sum_{i=1}^n \E[\cU_{\epsilon_m,mn}-\rho|X_i]=\avgni g(\Li)-\rho$.
Because $g(\Li)$, $i=1,\ldots,n$ are i.i.d. random variables with common expectation $\E[g(L)]=\rho$, so by the classic CLT we have
\begin{align}\label{CLT Un}
	\sqrt{n}\widetilde{\cU}_{\epsilon_m,n}\stackrel{d}{\to} \cN(0,\widetilde{\sigma}_{1}^2) \mbox{ as } n\to\infty,
\end{align}
where $\widetilde{\sigma}_{1}^2 = \Var[g(L)] = \E[\1\{L\geq 0 \}] - (\E[\1\{L\geq 0\}])^2$.
Secondly, for any $j=1,\ldots,m$, because all $X_i$, $i=1,\ldots,n$ are independent of $Y_j$ and $Y_k$ is independent of $Y_j$ for any $k\neq j$, so
\begin{align}
	&\E[\cU_{\epsilon_m,mn}-\rho|Y_j] = \E\left[\left.\frac{1}{mn}\sum_{i=1}^n\sum_{j=1}^m g(\Li)\right|Y_j\right] + \E\left[\left.\frac{1}{mn}\sum_{i=1}^n\sum_{j=1}^m \cV_{\epsilon_m,ij}\right|Y_j\right] - \rho\nonumber\\
	=& \E\left[g(L)\right] + \frac{1}{m}\left(\E\left[\cV_{\epsilon_m, 1j}|Y_j\right] + \sum_{\substack{k=1 \\ k\neq j}}^m \E\left[\left.\cV_{\epsilon_m, 1k}\right|Y_j\right]\right) - \rho\stackrel{\eqref{Lambda}}{=} \rho + \frac{1}{m}\left(\E\left[\cV_{\epsilon_m, 1j}|Y_j\right] + 0\right) - \rho \nonumber\\
	=& \frac{1}{m}\E\left[\cV_{\epsilon_m, 1j}|Y_j\right] =: \frac{1}{m}\widetilde{Y}_{\epsilon_m,j}.\label{eq:aux15}
\end{align}
Therefore $\widetilde{\cU}_{\epsilon_m,m}=\sum_{j=1}^m \E[\cU_{\epsilon_m,mn}-\rho|Y_j]=\avgmj \E\left[\cV_{\epsilon_m, 1j}|Y_j\right] =\avgmj \widetilde{Y}_{\epsilon_m,j}$.
Assumption~\ref{assm:indicatordensity} implies that $\psi(x,\epsilon_m u)=\psi_0(x) +\epsilon_m u\cdot \frac{\partial}{\partial \ell}\psi(x,\bar{u})$
where $\psi_0(x)=\psi(x,0)$ and $\bar{u}$ is between 0 and $\epsilon_m u$.
Denote a generic $\widetilde{Y}_{\epsilon_m} = \widetilde{Y}_{\epsilon_m,j}$ for notational convenience.
Then it follows that
 \begin{align}
	\widetilde{Y}_{\epsilon_m}=&\int \int\phi(u)(\hatH(x,Y)-\epsilon_m u)\psi(x,\epsilon_m u) \d x\d u\nonumber\\
 =&\int \int\phi(u)(\hatH(x,Y)-\epsilon_m u)\left[\psi_0(x) +\epsilon_m u \frac{\partial}{\partial \ell}\psi(x,\bar{u}) \right] \d x\d u\nonumber\\
 =&\int\phi(u)\d u \int\hatH(x,Y)\psi_0(x) \d x -\epsilon_m\int u\cdot\phi(u) \d u \int\psi_0(x) \d x \label{eq:doubleintegral1}\\
 &+ \epsilon_m\int \int \phi(u)u\left(\hatH(x,Y)-\epsilon_m u\right) \frac{\partial}{\partial \ell}\psi(x,\bar{u}) \d x\d u.\label{eq:doubleintegral}
	\end{align}
By~\eqref{eq:aux12}, the first term in~\eqref{eq:doubleintegral1} equals $\int\hatH(x,Y)\psi_0(x) \d x$ and
the second term equals $0$.
Moreover, recall $|\frac{\partial}{\partial \ell} \psi(x,\ell)|\leq \psi_1(x)$ in Assumption~\ref{assm:indicatordensity}~\ref{assum:existence}, so
\begin{align*}
	\eqref{eq:doubleintegral}\leq&\epsilon_m\int \int \phi(u)|u|\left(|\hatH(x,Y)|+\epsilon_m |u|\right) \psi_1(x) \d x\d u\\
	\leq &\epsilon_m\int |u| \cdot \phi(u) \d u \int |\hatH(x,Y)| \psi_1(x) \d x + \epsilon_m^2\int u^2\cdot\phi(u) \d u \int \psi_1(x) \d x \\
	\stackrel{(*)}{=} &\cO(\epsilon_m) + \cO(\epsilon_m^2) = \cO(\epsilon_m),
\end{align*}
where $(*)$ holds because of~\eqref{eq:aux13} and Assumption~\ref{assm:indicatordensity}~\ref{assum:boundedmoments}.
Therefore,
%
\begin{align}
\E\left[\widetilde{Y}_{\epsilon_m}^2\right] =& \E\left[\left(\int\hatH(x,Y)\psi_0(x)\d x+\cO(\epsilon_m)\right)^2\right]\nonumber\\
=& \E\left[\left(\int\hatH(x,Y)\psi_0(x)\d x\right)^2\right]+\cO(\epsilon_m)=\widetilde{\sigma}_2^2+\cO(\epsilon_m).\label{ind CLT Y}
\end{align}

Since $\widetilde{Y}_{\epsilon_m,j}$, $j=1,\ldots,m$ are i.i.d. samples, the characteristic function for $\sqrt{m}\widetilde{\cU}_{\epsilon_m,m}$ is given by
\begin{align*}
\varphi_{\epsilon_m,m}(t) = \E\left[\exp\left(it \sum_{j=1}^m\frac{\widetilde{Y}_{\epsilon_m,j}}{\sqrt{m}}\right)\right] = \left(\E\left[\exp\left(it \frac{\widetilde{Y}_{\epsilon_m}}{\sqrt{m}}\right)\right]\right)^m.
\end{align*}
Using the Taylor's theorem, $\E\left[\widetilde{Y}_{\epsilon_m}\right]=0$, and $\E\left[\widetilde{Y}_{\epsilon_m}^2\right]=\widetilde{\sigma}_2^2+\cO(\epsilon_m)$ we have
\begin{align*}
\E\left[\exp\left(it \frac{\widetilde{Y}_{\epsilon_m}}{\sqrt{m}}\right)\right]=&1-\frac{t^2}{2m}\E\left[\widetilde{Y}_{\epsilon_m}^2\right]+o\left(\frac{t^2}{m}\right)\\
=&1-\frac{t^2}{2m}\widetilde{\sigma}_2^2+\frac{t^2}{2m}\cO(\epsilon_m)+o\left(\frac{t^2}{m}\right),\mbox{ as } \frac{t}{\sqrt{m}}\to 0.
\end{align*}
So the characteristic function $\varphi_{\epsilon_m,m}(t) \to \exp\left(-\frac{t^2}{2}\widetilde{\sigma}_2^2\right)$ as $m\to\infty$ and $\epsilon_m\to 0$.
By the L\'{e}vy's continuity theorem, this means that
\begin{align}\label{CLT Um}
\sqrt{m}\widetilde{\cU}_{\epsilon_m,m} \stackrel{d}{\to} \cN(0,\widetilde{\sigma}_2^2).
\end{align}

Because $\widetilde{\cU}_{\epsilon_m,mn} = \widetilde{\cU}_{\epsilon_m,n} + \widetilde{\cU}_{\epsilon_m,m}$, where $\widetilde{\cU}_{\epsilon_m,n}$ and $\widetilde{\cU}_{\epsilon_m,m}$ are independent.
Then it follows from~\eqref{CLT Un} and~\eqref{CLT Um} that,
\begin{equation}\label{eq:result1}
\frac{\widetilde{\cU}_{\epsilon_m,mn}}{\widetilde{\sigma}_{mn}} \stackrel{d}{\to} \cN(0,1),\quad \mbox{ as }\min\{m,n\}\to\infty,
\end{equation}
where $\widetilde{\sigma}_{mn}^2 = \frac{\widetilde{\sigma}_{1}^2}{n} + \frac{\widetilde{\sigma}_2^2}{m}$, $\widetilde{\sigma}_{1}^2 = \Var[g(L)]$, and $\widetilde{\sigma}_2^2 = \E\left[\left(\int\hatH(x,Y)\psi_0(x) \d x\right)^2\right]$.

Next, we show $
\E\left[\widetilde{\sigma}_{mn}^{-2}(\cU_{\epsilon_m,mn}-\rho-\widetilde{\cU}_{\epsilon_m,mn})^2\right]\to 0,\ \text{as}\ \min\{m,n\}\to\infty$,
which implies $\widetilde{\sigma}_{mn}^{-1}(\cU_{\epsilon_m,mn}-\rho-\widetilde{\cU}_{\epsilon_m,mn})\stackrel{d}{\to} 0,\ \text{as}\ \min\{m,n\}\to\infty$.

Note that
\begin{align}\label{CLT indi dif}
\E\left[(\cU_{\epsilon_m,mn}-\rho-\widetilde{\cU}_{\epsilon_m,mn})^2\right]=\E\left[(\cU_{\epsilon_m,mn}-\rho)^2\right]+\E\left[\widetilde{\cU}_{\epsilon_m,mn}^2\right] -2\E\left[\left(\cU_{\epsilon_m,mn}-\rho\right)\widetilde{\cU}_{\epsilon_m,mn}\right].
\end{align}
We investigate the three terms on the RHS of~\eqref{CLT indi dif} as follows:
\begin{enumerate}[label=(\Roman*)]
	\item For $ \E\left[(\cU_{\epsilon_m,mn}-\rho)^2\right]$, it follows from the definition of $\cU_{\epsilon_m,mn}$ that
	\begin{align*}
		&\E\left[(\cU_{\epsilon_m,mn}-\rho)^2\right]=\E\left[\left(\frac{1}{mn}\sum_{i=1}^n \sum_{j=1}^m \left[g(\Li) + \cV_{\epsilon_m,ij}\right]-\rho\right)^2\right]\nonumber\\
		=&\E\left[\left(\frac{1}{mn}\sum_{i=1}^n \sum_{j=1}^m \cV_{\epsilon_m,ij}\right)^2\right]+\E\left[\left(\frac{1}{n}\sum_{i=1}^n g(\Li)-\rho\right)^2\right]+2\E\left[\frac{1}{mn}\sum_{i=1}^n \sum_{j=1}^m \cV_{\epsilon_m,ij}\left(\frac{1}{n}\sum_{i=1}^n g(\Li)-\rho\right)\right].
	\end{align*}
	We define $\cG=\sigma\left(X_1,...,X_n\right)$ and analyze the three terms on the RHS one by one.
	
	Firstly, it follows that
	\begin{align*}
		&\E\left[\left(\frac{1}{mn}\sum_{i=1}^n \sum_{j=1}^m \cV_{\epsilon_m,ij}\right)^2\right]=\E\left[\E\left[\left.\left(\frac{1}{mn}\sum_{i=1}^n \sum_{j=1}^m \cV_{\epsilon_m,ij}\right)^2\right|\cG \right]\right]\\
		\stackrel{(*)}{=}&\E\left[\frac{1}{m}\E\left[\left.\left(\frac{1}{n}\sum_{i=1}^n \cV_{\epsilon_m,i1}\right)^2\right|\cG \right]\right]=\frac{1}{m}\E\left[\left(\frac{1}{n}\sum_{i=1}^n \cV_{\epsilon_m,i1}\right)^2\right]\\
		=&\frac{1}{mn^2}\left( \sum_{i=1}^n\E\left[\cV_{\epsilon_m,i1}^2\right]+\sum_{i=1}^n\sum_{\substack{k=1 \\ k\neq i}} \E\left[\cV_{\epsilon_m,i1}\cdot\cV_{\epsilon_m,k1}\right]\right)=\frac{1}{mn} \E\left[\cV_{\epsilon_m,11}^2\right]+\frac{n-1}{mn} \E\left[\cV_{\epsilon_m,11}\cdot\cV_{\epsilon_m,21}\right]\\
		\stackrel{(**)}{=}&\frac{1}{mn} \E\left[\cV_{\epsilon_m,11}^2\right]+\frac{n-1}{mn} \E\left[\E\left[\left.\cV_{\epsilon_m,11}\right|Y_1\right]\E\left[\left.\cV_{\epsilon_m,21}\right|Y_1\right]\right]=\frac{1}{mn} \E\left[\cV_{\epsilon_m}^2\right]+\frac{n-1}{mn} \E\left[\left(\E\left[\left.\cV_{\epsilon_m}\right|Y\right]\right)^2\right],
	\end{align*}
	where $(*)$ holds by Lemma~\ref{lem:rv_moment2p} because given $\cG$, $\frac{1}{n}\sum_{i=1}^n \cV_{\epsilon_m,ij}$ ($j=1,...,m$) are i.i.d samples with mean 0 and $(**)$ holds because $\cV_{\epsilon_m,11}$ and $\cV_{\epsilon_m,21}$ are conditionally independent given $Y_1$.
	
	Secondly, by Lemma~\ref{lem:rv_moment2p}, $\E\left[\left(\frac{1}{n}\sum_{i=1}^n g(\Li)-\rho\right)^2\right]=\frac{1}{n}\E\left[\left(g(L)-\rho\right)^2\right]$.

	Thirdly, note that
	\begin{align*}
		&\E\left[\frac{1}{mn}\sum_{i=1}^n \sum_{j=1}^m \cV_{\epsilon_m,ij}\left(\frac{1}{n}\sum_{i=1}^n g(\Li)-\rho\right)\right]=\E\left[\E\left[\left.\frac{1}{mn}\sum_{i=1}^n \sum_{j=1}^m \cV_{\epsilon_m,ij}\left(\frac{1}{n}\sum_{i=1}^n g(\Li)-\rho\right)\right|\cG \right]\right]\\
		=&\E\left[\E\left[\left.\frac{1}{mn}\sum_{i=1}^n \sum_{j=1}^m \cV_{\epsilon_m,ij}\right|\cG \right]\left(\frac{1}{n}\sum_{i=1}^n g(\Li)-\rho\right)\right]=\E\left[\E\left[\left. \cV_{\epsilon_m}\right|X \right]\left(\frac{1}{n}\sum_{i=1}^n g(\Li)-\rho\right)\right]\stackrel{\eqref{Lambda}}{=}0.
	\end{align*}
	
	Summarizing the above discussions, we get
	\begin{align}\label{CLT indi dif1}
		\E\left[\left(\widetilde{\cU}_{\epsilon_m,mn}-\rho\right)^2\right]=\frac{1}{mn} \E\left[\cV_{\epsilon_m}^2\right]+\frac{n-1}{mn} \E\left[\left(\E\left[\left.\cV_{\epsilon_m}\right|Y\right]\right)^2\right]+\frac{1}{n}\E\left[\left(g(L)-\rho\right)^2\right].
	\end{align}

	\item Recall that $\widetilde{\cU}_{\epsilon_m,mn} = \widetilde{\cU}_{\epsilon_m,n} + \widetilde{\cU}_{\epsilon_m,m}$ where $\widetilde{\cU}_{\epsilon_m,n}$ and $\widetilde{\cU}_{\epsilon_m,m}$ are independent. Moreover $\E\left[\widetilde{\cU}_{\epsilon_m,n}\right]=\E\left[\widetilde{\cU}_{\epsilon_m,m}\right]=0$, therefore
	\begin{align}\label{CLT indi dif2}
		\E \left[\widetilde{\cU}_{\epsilon_m,mn}^2\right] =& \E\left[\left(\widetilde{\cU}_{\epsilon_m,n} + \widetilde{\cU}_{\epsilon_m,m}\right)^2\right]=\E\left[\widetilde{\cU}_{\epsilon_m,n}^2\right] + \E\left[\widetilde{\cU}_{\epsilon_m,m}^2\right] + 2\E\left[\widetilde{\cU}_{\epsilon_m,n}\cdot\widetilde{\cU}_{\epsilon_m,m}\right]\nonumber\\
		=& \E\left[\widetilde{\cU}_{\epsilon_m,n}^2\right] + \E\left[\widetilde{\cU}_{\epsilon_m,m}^2\right] \stackrel{(*)}{=} \frac{1}{n}\E\left[\left(g(L)-\rho\right)^2\right]+\frac{1}{m}\E\left[\left(\E\left[\left.\cV_{\epsilon_m}\right|Y\right]\right)^2\right],
	\end{align}
	where $(*)$ holds by~\eqref{eq:aux14} and~\eqref{eq:aux15}.
	
	\item Consider the two terms on the RHS in the following
	\begin{align*}
		\E\left[\left(\cU_{\epsilon_m,mn}-\rho\right)\widetilde{\cU}_{\epsilon_m,mn}\right]=\E\left[\left(\cU_{\epsilon_m,mn}-\rho\right)\widetilde{\cU}_{\epsilon_m,n}\right]+\E\left[\left(\cU_{\epsilon_m,mn}-\rho\right)\widetilde{\cU}_{\epsilon_m,m}\right].
	\end{align*}

	Firstly,
	\begin{align*}
		&\E\left[\left(\cU_{\epsilon_m,mn}-\rho\right)\widetilde{\cU}_{\epsilon_m,n}\right]\stackrel{\eqref{eq:aux14}}{=}\frac{1}{n}\sum_{i=1}^n\E\left[\left(\cU_{\epsilon_m,mn}-\rho\right)\left(g(\Li)-\rho\right)\right]\\
		=&\frac{1}{n}\sum_{i=1}^n\E\left[\E\left[\left. \left(\cU_{\epsilon_m,mn}-\rho\right)\left(g(\Li)-\rho\right)\right|X_i\right]\right]=\frac{1}{n}\sum_{i=1}^n\E\left[\left(g(\Li)-\rho\right)\E\left[ \cU_{\epsilon_m,mn}-\rho \left|X_i\right.\right]\right]\\
		\stackrel{\eqref{eq:aux14}}{=}&\frac{1}{n}\sum_{i=1}^n\E\left\{\left(g(\Li)-\rho\right)\cdot \frac{1}{n}\left(g(\Li)-\rho\right)\right\}=\frac{1}{n}\E\left[\left(g(L)-\rho\right)^2\right].
	\end{align*}
	Secondly,
	\begin{align*}
		&\E\left[\left(\cU_{\epsilon_m,mn}-\rho\right)\widetilde{\cU}_{\epsilon_m,m}\right]\stackrel{\eqref{eq:aux15}}{=}\frac{1}{m}\sum_{j=1}^m \E\left[\left(\cU_{\epsilon_m,mn}-\rho\right)\E\left[\cV_{\epsilon_m,1j}\left|Y_j\right. \right]\right]\nonumber\\
		=&\frac{1}{m}\sum_{j=1}^m \E\left[\E\left[\left. \left(\cU_{\epsilon_m,mn}-\rho\right)\E\left[\cV_{\epsilon_m,1j}\left|Y_j\right. \right]\right|Y_j\right]\right]
		=\frac{1}{m}\sum_{j=1}^m \E\left[\E\left[\cV_{\epsilon_m,1j}\left|Y_j\right. \right] \E\left[ \cU_{\epsilon_m,mn}-\rho \left|Y_j\right. \right]\right]\nonumber\\
		\stackrel{\eqref{eq:aux15}}{=}&\frac{1}{m}\sum_{j=1}^m \E\left\{\E\left[\cV_{\epsilon_m,1j}\left|Y_j\right. \right]\cdot \frac{1}{m}\E\left[\cV_{\epsilon_m,1j}\left|Y_j\right. \right]\right\}
		=\frac{1}{m}\E\left[\left(\E\left[\left.\cV_{\epsilon_m}\right|Y\right]\right)^2\right].
	\end{align*}
	Therefore,
	\begin{align}\label{CLT indi dif3}
		\E\left[\left(\cU_{\epsilon_m,mn}-\rho\right)\widetilde{\cU}_{\epsilon_m,mn}\right]=\frac{1}{n}\E\left(g(L)-\rho\right)^2+\frac{1}{m}\E\left[\left(\E\left[\left.\cV_{\epsilon_m}\right|Y\right]\right)^2\right].
	\end{align}
\end{enumerate}

Plugging~\eqref{CLT indi dif1},~\eqref{CLT indi dif2}, and~\eqref{CLT indi dif3} into~\eqref{CLT indi dif}, we get
\begin{align*}
\E\left[(\cU_{\epsilon_m,mn}-\rho-\widetilde{\cU}_{\epsilon_m,mn})^2\right]=\frac{1}{mn} \E\left[\cV_{\epsilon_m}^2\right] - \frac{1}{mn} \E\left[\left(\E\left[\left.\cV_{\epsilon_m}\right|Y\right]\right)^2\right]=\frac{1}{mn} \E\left[\cV_{\epsilon_m}^2\right] - \frac{1}{mn} \E\left[\widetilde{Y}_{\epsilon_m}^2\right].
\end{align*}
Note that
\begin{align}
\E[\cV_{\epsilon_m}^2] &= \E\left[\E\left[\left. \left(g'_{\epsilon_m}(L)(\hatH(X,Y)-L)\right)^2\right|Y\right]\right]= \E\left[\int \int\left(\frac{1}{\epsilon_m}\phi\left(\frac{\ell}{\epsilon_m}\right)(\hatH(x,Y)-\ell) \right)^2\psi(x,\ell)\d x\d\ell\right]\nonumber\\
&= \epsilon_m^{-1}\E\left[\int \int\left(\phi(u)(\hatH(x,Y)-\epsilon_m u) \right)^2\psi(x,\epsilon_m u)\d x\d u\right]\nonumber\\
&\leq \epsilon_m^{-1}\E\left[\int \int\left(\phi(u)(\hatH(x,Y)-\epsilon_m u) \right)^2 \left[\psi_0(x) +\epsilon_m |u| \psi_1(x) \right] \d x\d u\right]\nonumber\\
&\stackrel{(*)}{=}\cO(\epsilon_m^{-1}),\label{ind CLT Lambda}
\end{align}
where (*) holds because the double integrals inside the expectation is bounded (e.g., multiply out all the terms then use~\eqref{eq:aux12},~\eqref{eq:aux13}, and Assumption~\ref{assm:indicatordensity}~\ref{assum:boundedmoments}).

By~\eqref{ind CLT Y} and~\eqref{ind CLT Lambda} and provided that $m\epsilon_m \to \infty$, we have
\begin{align*}
	\E\left[\widetilde{\sigma}_{mn}^{-2}(\cU_{\epsilon_m,mn}-\rho-\widetilde{\cU}_{\epsilon_m,mn})^2\right]=\sigma_{mn}^{-2}\left[\cO\left(\frac{1}{mn\epsilon_m}\right) - \frac{\sigma_2^2 +\cO(\epsilon_m)}{mn}\right]\to 0, \ \text{as} \ m\to\infty.
\end{align*}
This shows that $\widetilde{\sigma}_{mn}^{-1}(\cU_{\epsilon_m,mn}-\rho-\widetilde{\cU}_{\epsilon_m,mn})\stackrel{\cL^2}{\to} 0$ and hence $\widetilde{\sigma}_{mn}^{-1}(\cU_{\epsilon_m,mn}-\rho-\widetilde{\cU}_{\epsilon_m,mn})\condist 0$.
Combining this with~\eqref{eq:result1} then apply the Slutsky's theorem to~\eqref{eq:decomposition1}, we have
	\begin{equation*}
	\frac{\cU_{\epsilon_m,mn}-\rho}{\sigma_{mn}} \stackrel{d}{\to} \cN(0,1).
	\end{equation*}
The proof is complete.
\end{proof}

\begin{lemma}\label{lem indLCT ra}
	Suppose the conditions for Theorem~\ref{thm:CLTIndicator} hold.
	Then,
	\begin{align*}
		\widetilde{\sigma}_{mn}^{-1} r_{\epsilon_m,mn}^a \conlone 0,\ \mbox{ as }\ \min\{m,n\}\to\infty.
	\end{align*}
\end{lemma}
\begin{proof}[Proof of Lemma~\ref{lem indLCT ra}]
Recall that $r_{\epsilon_m,mn}^a=\avgni\left[ g_{\epsilon_m}''(\Li)(\Lmi-\Li)^2\right]$. By Assumption~\ref{assm:jointdensity},
\begin{align*}
&\E[|r_{\epsilon_m,mn}^a|] \leq \E\left[\left|g''_{\epsilon_m}(L)(L-L_m)^2\right|\right]= \int \int_{-2\pi\epsilon_m}^{2\pi\epsilon_m} \left|\frac{1}{4\pi\epsilon_m^2}\sin\left(\frac{\ell}{\epsilon_m}\right)\left(\frac{z}{\sqrt{m}}\right)^2\right|p_m(\ell,z)\d\ell\d z\\
=& \frac{1}{4\pi m\epsilon_m}\int \int_{-2\pi}^{2\pi} \left|\sin u\right|z^2p_m(\epsilon_m u,z)\d u\d z\leq \frac{1}{4\pi m \epsilon_m}\int_{-2\pi}^{2\pi} \left|\sin u\right|\d u\int z^2\bar{p}_{0,m}(z)\d z=\cO\left((m \epsilon_m)^{-1}\right).
\end{align*}
Then as $m\to\infty$, provided that $ m\epsilon_m^2 \to \infty$ we get
\begin{equation*}
\E\left[\left|\widetilde{\sigma}_{mn}^{-1} r_{\epsilon_m,mn}^a\right|\right] =\cO\left(\left[\frac{m^2\epsilon_m^2}{n}\widetilde{\sigma}_{1}^2+m\epsilon_m^2\widetilde{\sigma}_{2}^2\right]^{-1/2}\right) \to 0.
\end{equation*}
The proof is complete.
\end{proof}

\begin{lemma}\label{lem indLCT rb}
	Suppose the conditions for Theorem~\ref{thm:CLTIndicator} hold.
	Then,
	\begin{align*}
		\widetilde{\sigma}_{mn}^{-1} r_{\epsilon_m,mn}^b \conlone0,\ \mbox{ as }\ \min\{m,n\}\to\infty.
	\end{align*}
\end{lemma}
\begin{proof}[Proof of Lemma~\ref{lem indLCT rb}]
Note that
\begin{align*}
g_{\epsilon_m}(x)-g(x)=-\frac{1}{2}\left[\1{\{0\leq x\leq2\pi\epsilon_m\}}-\1{\{-2\pi\epsilon_m\leq x<0\}}\right] + \dfrac{1}{4\pi}\left[\dfrac{x}{\epsilon_m}-\sin\left(\dfrac{x}{\epsilon_m}\right)\right]\1{\{|x|\leq2\pi\epsilon_m\}}.
\end{align*}
Then it follows that
\begin{align*}
r_{\epsilon_m,mn}^b=\avgni \left[g_{\epsilon_m}(\Li)-g(\Li)\right]=&-\frac{1}{2}\frac{1}{n}\sum_{i=1}^n \left[\1{\{0\leq \Li\leq2\pi\epsilon_m\}}-\1{\{-2\pi\epsilon_m\leq \Li<0\}}\right]\\
&+\dfrac{1}{4\pi}\frac{1}{n}\sum_{i=1}^n \left[\dfrac{\Li}{\epsilon_m}-\sin\left(\dfrac{\Li}{\epsilon_m}\right)\right]\1{\{|\Li|\leq2\pi\epsilon_m\}}\nonumber\\
=&-\frac{1}{2}\frac{1}{n}\sum_{i=1}^n \bar{C}_{i}+\dfrac{1}{4\pi}\frac{1}{n}\sum_{i=1}^n \widetilde{C}_{i},
\end{align*}
where $
\bar{C}_{i}=\1{\{0\leq \Li\leq2\pi\epsilon_m\}}-\1{\{-2\pi\epsilon_m\leq \Li<0\}},\ {\rm and}\ \widetilde{C}_{i}=\left[\dfrac{\Li}{\epsilon_m}-\sin\left(\dfrac{\Li}{\epsilon_m}\right)\right]\1{\{|\Li|\leq2\pi\epsilon_m\}}.
$
By Chebyshev's inequality, it suffices to prove that, as $\min\{m,n\}\to\infty$, $\sigma^{-1}_{mn}\cdot\E [\bar{C}_1]\to 0$, $\sigma^{-1}_{mn}\cdot\E [\widetilde{C}_1]\to 0$, $\sigma^{-2}_{mn}n^{-1}\cdot\Var(\bar{C}_1)\to 0$, and $\sigma^{-2}_{mn}n^{-1}\cdot\Var(\widetilde{C}_1)\to 0$.
Also note that
\begin{align}
	\sigma^{-1}_{mn}\epsilon_m^2 =& \left(\frac{\sigma_1^2}{n}+\frac{\sigma_2^2}{m}\right)^{-1/2} \cdot \epsilon_m^2 = \left(\frac{\sigma_1^2}{n\epsilon_m^4}+\frac{\sigma_2^2}{m\epsilon_m^4}\right)^{-1/2} \leq \left(\frac{m\epsilon_m^4}{\sigma_2^2}\right)^{1/2}\to 0, \mbox{ as }\ m\to\infty, \label{indCLT ConV1}\\
	\sigma^{-2}_{mn}n^{-1}=&\left(\frac{\sigma_1^2}{n}+\frac{\sigma_2^2}{m}\right)^{-1}\cdot n^{-1} = \left(\sigma_1^2+\frac{n\sigma_2^2}{m}\right)^{-1}<\sigma_1^{-2},\label{indCLT ConV2}
\end{align}
where~\eqref{indCLT ConV1} holds provided that $m\epsilon_m^4 \to 0$.
Then, we just have to prove that (I) $\E [\bar{C}_1]=\cO(\epsilon_m^2)$, (II) $\E [\widetilde{C}_1]=\cO(\epsilon_m^2)$, (III) $\Var(\bar{C}_1)\to 0$, and (IV) $\Var(\widetilde{C}_1)\to 0$ as $\min\{m,n\}\to\infty$.

\begin{enumerate}[label=(\Roman*)]
	\item Proving $\E [\bar{C}_1]=\cO(\epsilon_m^2)$. Note that, by Assumption~\ref{assm:jointdensity} and the mean value theorem,
\begin{align*}
\E [\bar{C}_1]=&\E \left[\1{\{0\leq L\leq2\pi\epsilon_m\}}-\1{\{-2\pi\epsilon_m\leq L<0\}}\right]=\P \left(0\leq L\leq2\pi\epsilon_m\right)-\P \left(-2\pi\epsilon_m\leq L<0\right)\\
=&\int_0^{2\pi\epsilon_m}\int p_m(\ell,z)\d z\d\ell -\int_{-2\pi\epsilon_m}^0\int p_m(\ell,z)\d z\d\ell\\
=& 2\pi\epsilon_m \int p_m(\xi_1,z)\d z-2\pi\epsilon_m \int p_m(\xi_2,z)\d z= 2\pi\epsilon_m \int \left[p_m(\xi_1,z)-p_m(\xi_2,z)\right]\d z\\
=& 2\pi\epsilon_m \int (\xi_1-\xi_2) \frac{\partial}{\partial \ell}p_m(\xi,z)\d z,
\end{align*}
where $\xi_1\in(0,2\pi\epsilon_m)$, $\xi_2\in(-2\pi\epsilon_m,0)$, $\xi\in(\xi_2,\xi_1)$ and so $|\xi_1-\xi_2|<4\pi\epsilon_m$. Therefore,
\begin{align*}
\left|\E [\bar{C}_1]\right|\leq 2\pi\epsilon_m \int \left|(\xi_1-\xi_2) \frac{\partial}{\partial \ell} p_m(\xi,z)\right|\d z\leq8\pi\epsilon_m^2 \int \bar{p}_{1,m}(z)\d z=\cO(\epsilon_m^2).
\end{align*}

	\item Proving $\E [\widetilde{C}_1]=\cO(\epsilon_m^2)$. Note that, by integration by parts and the mean value theorem,
	\begin{align*}
		&\E [\widetilde{C}_1]=\E \left\{\left[\dfrac{L}{\epsilon_m}-\sin\left(\dfrac{L}{\epsilon_m}\right)\right]\1{\{|L|\leq2\pi\epsilon_m\}}\right\}=\int \int_{-2\pi\epsilon_m}^{2\pi\epsilon_m} \left[\dfrac{\ell}{\epsilon_m}-\sin\left(\dfrac{\ell}{\epsilon_m}\right)\right] p_m(\ell,z)\d\ell\d z\\
		=&\epsilon_m\int \int_{-2\pi}^{2\pi} (t-\sin t) p_m(\epsilon_m t,z)\d t\d z=\epsilon_m\int \int_{-2\pi}^{2\pi} (t-\sin t) \left[p_m(0,z)+\epsilon_m t\frac{\partial}{\partial \ell} p_m(\xi,z)\right]\d t\d z\\
		=&\epsilon_m\int p_m(0,z)\d z\int_{-2\pi}^{2\pi} (t-\sin t)\d t+\epsilon_m^2\int \int_{-2\pi}^{2\pi} (t-\sin t) t\frac{\partial}{\partial \ell} p_m(\xi,z)\d t\d z\\
		=&0+\epsilon_m^2\int \int_{-2\pi}^{2\pi} (t-\sin t) t\frac{\partial}{\partial \ell} p_m(\xi,z)\d t\d z,
	\end{align*}
	where $\xi$ is between 0 and $\epsilon_m t$. Hence,
	\begin{align*}
		\left|\E [\widetilde{C}_1]\right|\leq\epsilon_m^2\int \int_{-2\pi}^{2\pi} \left|(t-\sin t) t \frac{\partial}{\partial \ell} p_m(\xi,z)\right|\d t\d z\leq\epsilon_m^2\int_{-2\pi}^{2\pi} \left|(t-\sin t) t\right|\d t\int \bar{p}_{1,m}(z)\d z=\cO(\epsilon_m^2).
	\end{align*}

	\item Proving $\Var(\bar{C}_1)\to 0$. Note that, by integration by parts, the mean value theorem, and Assumption~\ref{assm:jointdensity},
	\begin{align*}
		&\Var(\bar{C}_1)\leq\E\left[\bar{C}_1^2\right]=\E\left[\left(\1{\{0\leq L\leq2\pi\epsilon_m\}}-\1{\{-2\pi\epsilon_m\leq L<0\}}\right)^2\right]\\
		=& \E\left[\left(\1{\left\{0\leq L\leq2\pi\epsilon_m\right\}}\right)^2 + \left(\1{\left\{-2\pi\epsilon_m\leq L<0\right\}}\right)^2\right]
		=\P\left(0\leq L\leq2\pi\epsilon_m\right)+ \P\left(-2\pi\epsilon_m\leq L<0 \right)\\
		=&\int \int_0^{2\pi\epsilon_m} p_m(\ell,z)\d\ell\d z + \int \int_{-2\pi\epsilon_m}^0 p_m(\ell,z)\d\ell\d z= 2\pi\epsilon_m \int p_m(\xi_1,z)\d z + 2\pi\epsilon_m \int p_m(\xi_2,z)\d z\\
		=& 2\pi\epsilon_m \int \left[p_m(\xi_1,z)+p_m(\xi_2,z)\right]\d z
		\leq 4\pi\epsilon_m \int \bar{p}_{0,m}(z)\d z=\cO(\epsilon_m)\to 0, \mbox{ as } m\to \infty,
	\end{align*}
	where $\xi_1\in(0,2\pi\epsilon_m)$ and $\xi_2\in(-2\pi\epsilon_m,0)$.
	
	\item Proving $\Var(\widetilde{C}_1)\to 0$. Note that, due to Assumption~\ref{assm:jointdensity},
	\begin{align*}
		\Var(\widetilde{C}_1)\leq&\E\left[\widetilde{C}_1^2\right]=\E\left|\left[\dfrac{L}{\epsilon_m}-\sin\left(\dfrac{L}{\epsilon_m}\right)\right]\1{\{|L|\leq2\pi\epsilon_m\}}\right|^2\\
		=& \int \int_{-2\pi\epsilon_m}^{2\pi\epsilon_m} \left[\dfrac{\ell}{\epsilon_m}-\sin\left(\dfrac{\ell}{\epsilon_m}\right)\right]^2 p_m(\ell,z)\d\ell\d z \\
		=& \epsilon_m\int \int_{-2\pi}^{2\pi} \left(t-\sin t\right)^2 p_m(\epsilon_m t,z)\d t\d z \\
		\leq& \epsilon_m \int_{-2\pi}^{2\pi} \left(t-\sin t\right)^2\d t\int \bar{p}_{0,m}(z)\d z=\cO(\epsilon_m)\to 0.
	\end{align*}
\end{enumerate}

The proof is complete.
\end{proof}

\begin{lemma}\label{lem indLCT rc}
	Suppose the conditions for Theorem~\ref{thm:CLTIndicator} hold.
	Then,
	\begin{align*}
		\widetilde{\sigma}_{mn}^{-1} r_{\epsilon_m,mn}^c \conlone0,\ \mbox{ as }\ \min\{m,n\}\to\infty.
	\end{align*}
\end{lemma}
\begin{proof}[Proof of Lemma~\ref{lem indLCT rc}]
Recall that $r_{\epsilon_m,mn}^c:=\avgni \left[g(\Lmi)-g_{\epsilon_m}(\Lmi)\right]$. Note that
\begin{align*}
r_{\epsilon_m,mn}^c =& \frac{1}{2n}\sum_{i=1}^n \left[\1{\left\{0\leq \Lmi\leq2\pi\epsilon_m\right\}}-\1{\left\{-2\pi\epsilon_m\leq \Lmi<0\right\}}\right]\\
&-\frac{1}{4\pi n}\sum_{i=1}^n \left[\dfrac{\Lmi}{\epsilon_m}-\sin\left(\dfrac{\Lmi}{\epsilon_m}\right)\right]\1{\left\{|\Lmi|\leq2\pi\epsilon_m\right\}}\nonumber\\
=&\frac{1}{2n}\sum_{i=1}^n \bar{D}_i - \frac{1}{4\pi n}\sum_{i=1}^n \widetilde{D}_i,
\end{align*}
where
$
\bar{D}_i= \1{\left\{0\leq \Lmi\leq2\pi\epsilon_m\right\}}-\1{\left\{-2\pi\epsilon_m\leq \Lmi<0\right\}} \mbox{ and }
\widetilde{D}_i= \left[\dfrac{\Lmi}{\epsilon_m}-\sin\left(\dfrac{\Lmi}{\epsilon_m}\right)\right]\1{\left\{|\Lmi|\leq2\pi\epsilon_m\right\}}.
$
By Chebyshev's inequality and~\eqref{indCLT ConV1}, it suffices to prove that, as $\min\{m,n\}\to\infty$, (I) $\E [\bar{D}_1]=\cO(\epsilon_m^2)$, (II) $\E [\widetilde{D}_1]=\cO(\epsilon_m^2)$, (III) $\Var(\sigma^{-1}_{mn}\cdot\frac{1}{n}\sum_{i=1}^n \bar{D}_i)\to 0$, and (IV) $\Var(\sigma^{-1}_{mn}\cdot\frac{1}{n}\sum_{i=1}^n \widetilde{D}_i)\to 0$.
Note that (I) and (II) are similar to (I) and (II) in the proof for Lemma~\ref{lem indLCT rb}, but (III) and (IV) are different because of the correlation among $\Lmi$.

\begin{enumerate}[label=(\Roman*)]
	\item Proving $\E [\bar{D}_1]=\cO(\epsilon_m^2)$. Note that, by the mean value theorem and Assumption~\ref{assm:jointdensity},
	\begin{align*}
		\E [\bar{D}_1]=&\E \left[\1{\left\{0\leq L_m\leq2\pi\epsilon_m\right\}}-\1{\left\{-2\pi\epsilon_m\leq L_m<0\right\}}\right]\\
		=& \P \left(0\leq L+Z_m/\sqrt{m}\leq2\pi\epsilon_m\right)-\P \left(-2\pi\epsilon_m\leq L+Z_m/\sqrt{m}<0\right)\\
		=& \int \int_{-z/\sqrt{m}}^{2\pi\epsilon_m-z/\sqrt{m}} p_m(\ell,z)\d\ell\d z - \int \int_{-2\pi\epsilon_m-z/\sqrt{m}}^{-z/\sqrt{m}} p_m(\ell,z)\d\ell\d z\\
		=& \int 2\pi\epsilon_m p_m(\xi_1,z)\d z - \int 2\pi\epsilon_m p_m(\xi_2,z)\d z\\
		=& 2\pi\epsilon_m \int [p_m(\xi_1,z)- p_m(\xi_2,z)]\d z=2\pi\epsilon_m \int (\xi_1-\xi_2)\frac{\partial}{\partial \ell} p_m(\xi,z)\d z,
	\end{align*}
	where $\xi_1\in(-z/\sqrt{m},2\pi\epsilon_m-z/\sqrt{m})$, $\xi_2\in(-2\pi\epsilon_m-z/\sqrt{m},-z/\sqrt{m})$ and $\xi\in(\xi_2,\xi_1)$. Then $|\xi_1-\xi_2|<4\pi\epsilon_m$. Therefore, by Assumption~\ref{assm:jointdensity},
	\begin{align*}
		\left|\E [\bar{D}_1]\right|\leq 2\pi\epsilon_m \int \left|(\xi_1-\xi_2)\frac{\partial}{\partial \ell} p_m(\xi,z)\right|\d z\leq 8\pi\epsilon_m^2 \int \bar{p}_{1,m}(z)\d z=\cO\left(\epsilon_m^2\right).
	\end{align*}

	\item Proving $\E [\widetilde{D}_1]=\cO(\epsilon_m^2)$. Note that
	\begin{align*}
		\E [\widetilde{D}_1]=&\E \left\{\left[\dfrac{L_m}{\epsilon_m}-\sin\left(\dfrac{L_m}{\epsilon_m}\right)\right]\1{\left\{|L_m|\leq2\pi\epsilon_m\right\}}\right\}\\
		=&\E \left\{\left[\dfrac{L+Z_m/\sqrt{m}}{\epsilon_m}-\sin\left(\dfrac{L+Z_m/\sqrt{m}}{\epsilon_m}\right)\right]\1{\left\{|L+Z_m/\sqrt{m}|\leq2\pi\epsilon_m\right\}}\right\}\\
		=&\int \int_{-2\pi\epsilon_m-z/\sqrt{m}}^{2\pi\epsilon_m-z/\sqrt{m}} \left[\dfrac{\ell+z/\sqrt{m}}{\epsilon_m}-\sin\left(\dfrac{\ell+z/\sqrt{m}}{\epsilon_m}\right)\right] p_m(\ell,z)\d\ell\d z\\
		=&\epsilon_m \int \int_{-2\pi}^{2\pi} (t-\sin t) p_m(\epsilon_m t-z/\sqrt{m},z)\d t\d z\\
		=&\epsilon_m \int \int_{-2\pi}^{2\pi} (t-\sin t) \left[p_m(-z/\sqrt{m},z)+\epsilon_m t\frac{\partial}{\partial \ell} p_m(\xi,z)\right]\d t\d z\\
		=&\epsilon_m \int p_m(-z/\sqrt{m},z)\d z \int_{-2\pi}^{2\pi} (t-\sin t)\d t + \epsilon_m^2 \int \int_{-2\pi}^{2\pi} (t-\sin t) t\frac{\partial}{\partial \ell} p_m(\xi,z)\d t\d z\\
		=&0 + \epsilon_m^2 \int \int_{-2\pi}^{2\pi} (t-\sin t) t\frac{\partial}{\partial \ell} p_m(\xi,z)\d t\d z,
	\end{align*}
where $\xi\in(-z/\sqrt{m},2\pi\epsilon_m t-z/\sqrt{m})$. Hence
	\begin{align*}
		\left|\E [\widetilde{D}_1]\right|\leq& \epsilon_m^2 \int \int_{-2\pi}^{2\pi} \left|(t-\sin t) t\frac{\partial}{\partial \ell} p_m(\xi,z)\right|\d t\d z 		\leq \epsilon_m^2 \int_{-2\pi}^{2\pi} \left|(t-\sin t) t\right|\d t\int \bar{p}_{1,m}(z)\d z
		= \cO\left(\epsilon_m^2\right).
	\end{align*}

	\item Proving $\Var(\sigma^{-1}_{mn}\cdot\frac{1}{n}\sum_{i=1}^n \bar{D}_i)\to 0$. By~\eqref{indCLT ConV2} and
	\begin{align*}
		\Var\left(\sigma^{-1}_{mn}\cdot\frac{1}{n}\sum_{i=1}^n \bar{D}_i\right)=\sigma^{-2}_{mn}\cdot\frac{1}{n}\Var(\bar{D}_1)+ \sigma^{-2}_{mn}\cdot\frac{n-1}{n}{\rm Cov}(\bar{D}_1,\bar{D}_2).
	\end{align*}
	we just have to prove that $\Var(\bar{D}_1)\to 0$ and $\sigma^{-2}_{mn}\cdot{\rm Cov}(\bar{D}_1,\bar{D}_2)\to 0$.
	
	Firstly, note that by integration by parts, the mean value theorem, and Assumption~\ref{assm:jointdensity},
	\begin{align*}
		\Var(\bar{D}_1)\leq&\E\left[\bar{D}_1^2\right]=\E \left[\left(\1{\left\{0\leq L_m\leq2\pi\epsilon_m\right\}}-\1{\left\{-2\pi\epsilon_m\leq L_m<0\right\}}\right)^2\right]\\
		=& \E \left[\left(\1{\left\{0\leq L_m\leq2\pi\epsilon_m\right\}}\right)^2 + \left(\1{\left\{-2\pi\epsilon_m\leq L_m<0\right\}}\right)^2\right]\\
		=& \P \left(0\leq L+Z_m/\sqrt{m}\leq2\pi\epsilon_m\right) + \P \left(-2\pi\epsilon_m\leq L+ Z_m/\sqrt{m}<0\right)\\
		=& \int \int_{-z/\sqrt{m}}^{2\pi\epsilon_m-z/\sqrt{m}} p_m(\ell,z)\d\ell\d z + \int \int_{-2\pi\epsilon_m-z/\sqrt{m}}^{-z/\sqrt{m}} p_m(\ell,z)\d\ell\d z\\
		=& \int 2\pi\epsilon_m p_m(\xi_1,z)\d z + \int 2\pi\epsilon_m p_m(\xi_2,z)\d z\\
		\leq&\int 2\pi\epsilon_m \bar{p}_{0,m}(z)\d z + \int 2\pi\epsilon_m \bar{p}_{0,m}(z)\d z\\
		=&4\pi\epsilon_m \int \bar{p}_{0,m}(z)\d z =\cO(\epsilon_m)\to 0, \mbox{ as } m \to \infty
	\end{align*}
	where $\xi_1\in(-z/\sqrt{m},2\pi\epsilon_m-z/\sqrt{m})$ and $\xi_2\in(-2\pi\epsilon_m-z/\sqrt{m},-z/\sqrt{m})$.
	
	Secondly, because $\sigma^{-2}_{mn}{\rm Cov}(\bar{D}_1,\bar{D}_2)=\sigma^{-2}_{mn}\E [\bar{D}_1\bar{D}_2] - (\sigma^{-1}_{mn}\E [\bar{D}_1])^2$ and
	we have proved that $\E [\bar{D}_1]=\cO(\epsilon_m^2)$ so $\sigma^{-1}_{mn}\E [\bar{D}_1]\to 0$ as $m\to\infty$ in (I), then to show $\sigma^{-2}_{mn}\cdot{\rm Cov}(\bar{D}_1,\bar{D}_2)\to 0$ it suffices to prove that $\sigma^{-2}_{mn}\cdot\E [\bar{D}_1\bar{D}_2]\to 0$.
	
	Note that
	\begin{align*}
		\E [\bar{D}_1\bar{D}_2]=& \E \left[\bar{D}_1\left(\1{\left\{0\leq L_{m,2}\leq2\pi\epsilon_m\right\}}-\1{\left\{-2\pi\epsilon_m\leq L_{m,2}<0\right\}}\right)\right]\\
		=& \E \left[\bar{D}_1 \1{\left\{0\leq L_{m,2}\leq2\pi\epsilon_m\right\}}\right] - \E \left[\bar{D}_1\1{\left\{-2\pi\epsilon_m\leq L_{m,2}<0\right\}}\right].
	\end{align*}
	Evidently, the two terms in the above difference have the same convergence rate, so we only analyze the first term. Note that
	\begin{align}
		&\E \left[\bar{D}_1 \1{\left\{0\leq L_{m,2}\leq2\pi\epsilon_m\right\}}\right]\nonumber\\
		=& \E \left[\left(\1{\left\{0\leq L_{m,1}\leq2\pi\epsilon_m\right\}}-\1{\left\{-2\pi\epsilon_m<L_{m,1}<0\right\}}\right) \1{\left\{0\leq L_{m,2}\leq2\pi\epsilon_m\right\}}\right]\nonumber\\
		=& \E \left[\1{\left\{0\leq L_{m,1}\leq2\pi\epsilon_m\right\}} \1{\left\{0\leq L_{m,2}\leq2\pi\epsilon_m\right\}}\right] - \E \left[\1{\left\{-2\pi\epsilon_m<L_{m,1}<0\right\}} \1{\left\{0\leq L_{m,2}\leq2\pi\epsilon_m\right\}}\right].\label{indLCT.31}
	\end{align}
	We examine the two terms on the RHS of~\eqref{indLCT.31} separately. By Taylor expansion~\eqref{indVar Taylor q}, the first term equals
	\begin{align*}
		& \int \int \int_{-\frac{z_1}{\sqrt{m}}}^{2\pi\epsilon_m-\frac{z_1}{\sqrt{m}}}\int_{-\frac{z_2}{\sqrt{m}}}^{2\pi\epsilon_m-\frac{z_2}{\sqrt{m}}} q_m(\ell_1,\ell_2,z_1,z_2)\d\ell_1 \d\ell_2 \d z_1 \d z_2\\
		\stackrel{\eqref{indVar Taylor q_ineq}}{\leq} & (2\pi\epsilon_m)^2\int \int q_m(0,0,z_1,z_2)\d z_1 \d z_2 +  \\
		&\int \int \int_{-\frac{z_1}{\sqrt{m}}}^{2\pi\epsilon_m-\frac{z_1}{\sqrt{m}}}\int_{-\frac{z_2}{\sqrt{m}}}^{2\pi\epsilon_m-\frac{z_2}{\sqrt{m}}}(|\ell_1| + |\ell_2|)\bar{q}_{1,m}(z_1,z_2)\d\ell_1 \d\ell_2 \d z_1 \d z_2\\
		=& (2\pi\epsilon_m)^2\int \int q_m(0,0,z_1,z_2)\d z_1 \d z_2+\int \int \left[(2\pi\epsilon_m)^3-\frac{z_1+ z_2}{\sqrt{m}}(2\pi\epsilon_m)^2\right]\bar{q}_{1,m}(z_1,z_2) \d z_1 \d z_2\\
		=& (2\pi\epsilon_m)^2\int \int q_m(0,0,z_1,z_2)\d z_1 \d z_2+\cO\left(\epsilon_m^3\right) + \cO\left(\frac{\epsilon_m^2}{\sqrt{m}}\right).
	\end{align*}
	It can be shown similarly that the second term on the RHS of~\eqref{indLCT.31} also equals $$(2\pi\epsilon_m)^2\int \int q_m(0,0,z_1,z_2)\d z_1 \d z_2+\cO\left(\epsilon_m^3\right) + \cO\left(\frac{\epsilon_m^2}{\sqrt{m}}\right).$$
	Therefore, taking the difference, we have~\eqref{indLCT.31}$=\cO\left(\epsilon_m^3\right) + \cO\left(\frac{\epsilon_m^2}{\sqrt{m}}\right)$.
	Provided that $m\epsilon_m^3\to 0$, which implies that $\sqrt{m}\epsilon_m^2\to 0$, we have
	\begin{align*}
		\sigma^{-2}_{mn}\cdot\E [\bar{D}_1\bar{D}_2] &= \left(\frac{\sigma_1^2}{n}+\frac{\sigma_2^2}{m}\right)^{-1}\cdot\left(\cO\left(\epsilon_m^3\right) + \cO\left(\frac{\epsilon_m^2}{\sqrt{m}}\right)\right)\\
		&= \cO\left(\left(\frac{\sigma_1^2}{n\epsilon_m^3}+\frac{\sigma_2^2}{m\epsilon_m^3}\right)^{-1}\right)+\cO\left(\left(\frac{\sigma_1^2 \sqrt{m}}{n\epsilon_m^2}+\frac{\sigma_2^2}{\sqrt{m}\epsilon_m^2}\right)^{-1}\right) \to 0, \mbox{ as } m \to \infty.
	\end{align*}
	
	\item Proving $\Var(\sigma^{-1}_{mn}\cdot\frac{1}{n}\sum_{i=1}^n \widetilde{D}_i)\to 0$. By~\eqref{indCLT ConV2} and
	\begin{align*}
		\Var\left(\sigma^{-1}_{mn}\cdot\frac{1}{n}\sum_{i=1}^n \widetilde{D}_i\right)=\sigma^{-2}_{mn}\cdot\frac{1}{n}\Var( \widetilde{D}_1)+ \sigma^{-2}_{mn}\cdot\frac{n-1}{n}{\rm Cov}(\widetilde{D}_1,\widetilde{D}_2),
	\end{align*}
	we just have to prove that $\Var( \widetilde{D}_1)\to 0$ and $\sigma^{-2}_{mn}\cdot{\rm Cov}(\widetilde{D}_1,\widetilde{D}_2)\to 0$.
	Firstly, note that
	\begin{align*}
		&\Var( \widetilde{D}_1)\leq\E\left[\widetilde{D}_1^2\right]=\E \left\{\left[\dfrac{L_m}{\epsilon_m}-\sin\left(\dfrac{L_m}{\epsilon_m}\right)\right]^2\1{\left\{|L_m|\leq2\pi\epsilon_m\right\}}\right\}\\
		=&\E \left\{\left[\dfrac{L+Z_m/\sqrt{m}}{\epsilon_m}-\sin\left(\dfrac{L+Z_m/\sqrt{m}}{\epsilon_m}\right)\right]^2 \1{\left\{|L+Z_m/\sqrt{m}|\leq2\pi\epsilon_m\right\}}\right\}\\
		=&\int \int_{-2\pi\epsilon_m-z/\sqrt{m}}^{2\pi\epsilon_m-z/\sqrt{m}} \left[\dfrac{\ell+z/\sqrt{m}}{\epsilon_m}-\sin\left(\dfrac{\ell+z/\sqrt{m}}{\epsilon_m}\right)\right]^2 p_m(\ell,z)\d\ell\d z\\
		=&\epsilon_m \int \int_{-2\pi}^{2\pi} (t-\sin t)^2 p_m(\epsilon_m t-z/\sqrt{m},z)\d t\d z \stackrel{(*)}{\leq} \epsilon_m \int \bar{p}_{0,m}(z)\d z\int_{-2\pi}^{2\pi} (t-\sin t)^2\d t\\
		\stackrel{(**)}{=}&\cO(\epsilon_m),
	\end{align*}
where $(*)$ and $(**)$ hold by Assumption~\ref{assm:jointdensity}~\ref{assm:jointdensity2} and~\ref{assm:jointdensity}~\ref{assm:jointdensity3}, respectively.

	Second, because $\sigma^{-2}_{mn}{\rm Cov}(\widetilde{D}_1,\widetilde{D}_2)=\sigma^{-2}_{mn}\E [\widetilde{D}_1\widetilde{D}_2] - (\sigma^{-1}_{mn}\E [\widetilde{D}_1])^2$ and
	we have proved that $\E [\widetilde{D}_1]=\cO(\epsilon_m^2)$ so $\sigma^{-1}_{mn}\E [\widetilde{D}_1]\to 0$ as $m\to\infty$ in (II), then to show $\sigma^{-2}_{mn}\cdot{\rm Cov}(\widetilde{D}_1,\widetilde{D}_2)\to 0$ it suffices to prove that $\sigma^{-2}_{mn}\cdot\E [\widetilde{D}_1\widetilde{D}_2]\to 0$.
	Note that $\int_{-2\pi}^{2\pi} (x-\sin(x))dx = 0$, $\int_{-2\pi}^{2\pi}|x| \d x = 4\pi^2$, and $\int_{-2\pi}^{2\pi} \d x = 4\pi$, using the Taylor's theorem for $q_m(\ell_1,\ell_2,z_1,z_2)$ we have
	\begin{align*}
		&\E \left[\widetilde{D}_1\widetilde{D}_2\right]\\
		=& \E \left[\left[\dfrac{L_{m,1}}{\epsilon_m}-\sin\left(\dfrac{L_{m,1}}{\epsilon_m}\right)\right]\1{\left\{|L_{m,1}|\leq2\pi\epsilon_m\right\}}\right.
		\left. \cdot\left[\dfrac{L_{m,2}}{\epsilon_m}-\sin\left(\dfrac{L_{m,2}}{\epsilon_m}\right)\right]\1{\left\{|L_{m,2}|\leq2\pi\epsilon_m\right\}}\right]\\
		=& \epsilon_m^2 \int \int \int_{-2\pi}^{2\pi}\int_{-2\pi}^{2\pi} (u -\sin u)(v-\sin v) q_m\left(\epsilon_m u-\frac{z_1}{\sqrt{m}},\epsilon_m v-\frac{z_2}{\sqrt{m}},z_1,z_2\right)\d u\d v\d z_1 \d z_2\\
		=& \epsilon_m^2 \int \int \int_{-2\pi}^{2\pi}\int_{-2\pi}^{2\pi} (u -\sin u)(v-\sin v) \left[q_m(0,0,z_1,z_2)\right.\\
		&\left. + \left(\epsilon_m u-\frac{z_1}{\sqrt{m}}\right)\frac{\partial}{\partial y_1}q_m(\bar{u},\bar{v},z_1,z_2)+ \left(\epsilon_m u-\frac{z_2}{\sqrt{m}}\right)\frac{\partial}{\partial y_2}q_m(\bar{u},\bar{v},z_1,z_2)\right]\d u\d v\d z_1 \d z_2\\
		\leq& \epsilon_m^2 \int \int q_m(0,0,z_1,z_2)\d z_1 \d z_2 \cdot \int_{-2\pi}^{2\pi} (u -\sin u)\d u \cdot \int_{-2\pi}^{2\pi}(v-\sin v)\d v\\
		& +  \epsilon_m^2 \int \int \int_{-2\pi}^{2\pi}\int_{-2\pi}^{2\pi} \left(\epsilon_m |u|+\frac{|z_1|}{\sqrt{m}} + \epsilon_m |v|+\frac{|z_2|}{\sqrt{m}}\right)\bar{q}_{1,m}(z_1,z_2)\d u\d v\d z_1 \d z_2\\
		=& 0 + \epsilon_m^3  \left(\int_{-2\pi}^{2\pi}|u| \d u + \int_{-2\pi}^{2\pi}|v|\d v\right) \int \int\bar{q}_{1,m}(z_1,z_2)\d z_1 \d z_2\\
		& + \epsilon_m^2 \left(\int_{-2\pi}^{2\pi}\d u \int_{-2\pi}^{2\pi} \d v\right)\int \int  \left(\frac{|z_1|}{\sqrt{m}} +\frac{|z_2|}{\sqrt{m}}\right)\bar{q}_{1,m}(z_1,z_2)\d z_1 \d z_2 \\
		=&\cO\left(\epsilon_m^3\right) + \cO\left(\frac{\epsilon_m^2}{\sqrt{m}}\right).
	\end{align*}
	Similar to the arguments in (III), this implies that $\sigma^{-2}_{mn}\cdot\E [\widetilde{D}_1\widetilde{D}_2]=\cO\left(\left(\frac{\sigma_1^2}{n\epsilon_m^3}+\frac{\sigma_2^2}{m\epsilon_m^3}\right)^{-1}\right)+\cO\left(\left(\frac{\sigma_1^2 \sqrt{m}}{n\epsilon_m^2}+\frac{\sigma_2^2}{\sqrt{m}\epsilon_m^2}\right)^{-1}\right) \to 0$ as $m\to\infty$, provided that $m\epsilon_m^3 \to 0$.
\end{enumerate}

The proof is complete.
\end{proof}

\begin{lemma}\label{lem indLCT sig1}
Suppose the conditions for Theorem~\ref{thm:varianceestimateIndicator} hold, then $\widehat{\widetilde{\sigma}}_{1,mn}^2\conprob\tilde{\sigma}_{1}^2$ as $\min\{m,n\}\to 0$.
\end{lemma}

\begin{proof}
	By~\eqref{eq:IndVar1} and~\eqref{eq:sig1hatIndicator}, we have
	\begin{align}
			\widehat{\widetilde{\sigma}}_{1,mn}^2 - \widetilde{\sigma}_{1}^2 =& \left[\avgni \1\{\Lmi\geq 0\} - \E\left[\1 \{L\geq 0\}\right]\right]\label{eq:diff_sig1hatIndicator1}\\
			& + \left[\left(\avgni \1\{\Lmi\geq 0\}\right)^2 - \left(\E\left[\1 \{L\geq 0\}\right]\right)^2\right].\label{eq:diff_sig1hatIndicator2}
	\end{align}
%
	For~\eqref{eq:diff_sig1hatIndicator1}, note that
	\begin{align*}
		\lim_{m\to\infty} \E\left[\left|\avgni \1\{\Lmi\geq 0\} - \E\left[\1 \{L\geq 0\}\right]\right|\right] \stackrel{(*)}{=} \lim_{m\to\infty}\E\left[\left|\1\{L_m\geq 0\} -\1\{L \geq 0\}\right|\right]\stackrel{(**)}{=}0,
	\end{align*}
	where $(*)$ holds because $\Lmi$, $i=1,\ldots,n$ are identically distributed and $(**)$ holds by the dominated convergence theorem (as $\1\{x\geq 0\} \leq 1$ and $L_m\stackrel{a.s}{\to}L$ as $m\to\infty$ according to Proposition~\ref{prop:Lm_as}). Therefore~\eqref{eq:diff_sig1hatIndicator1}$\conlone 0$ and so~\eqref{eq:diff_sig1hatIndicator1}$\conprob 0$.
	For~\eqref{eq:diff_sig1hatIndicator2}, due to the continuous mapping theorem it suffices to show that $\avgni \1\{\Lmi\geq 0\} \conprob \E\left[\1 \{L\geq 0\}\right]$, which was shown above.
	
	The proof is complete.	
\end{proof}

\begin{lemma}\label{lem indLCT sig2}
	Suppose the conditions for Theorem~\ref{thm:varianceestimateIndicator} hold, then $\widehat{\widetilde{\sigma}}_{2,mn}^2 \conprob \widetilde{\sigma}_2^2$ as $\min\{m,n\}\to 0$.

\end{lemma}
\begin{proof}[Proof of Lemma~\ref{lem indLCT sig2}]
For convenience, define shorthanded notations
\begin{equation*}
R := \int \hatH(x,Y)\psi_0(x)\d x,\,\, R_{j} := \int \hatH(x,Y_j)\psi_0(x)\d x, \,\, R_{\epsilon,j} := \E[g'_{\epsilon}(L)\hatH_{1j}|Y_j], \mbox{ and }
\end{equation*}
\begin{equation*}
\widehat{R}_{\epsilon,j} = \avgni g'_{\epsilon}(\Li)\hatHij,\,\, \widehat{R}_{\epsilon, m,j} = \avgni g'_{\epsilon}(\Lmi)\hatHij.
\end{equation*}
Since $R_{j},R_{\epsilon,j},\widehat{R}_{\epsilon,j}$, and $\widehat{R}_{\epsilon, m,j}$ are identically distributed for $j=1,\ldots,m$ so for a generic index $j$ we omit the subscript for notational convenience.
By~\eqref{eq:IndVar2} and \eqref{eq:sig2hatIndicator}, we have
\begin{align*}
&\widehat{\widetilde{\sigma}}_{2,mn}^2 - \widetilde\sigma_2^2 = \avgmj \widehat{R}_{\epsilon, m,j}^2 - \E[R^2]\\
=&\avgmj \left[\widehat{R}_{\epsilon, m,j}^2 -\widehat{R}_{\epsilon,j}^2\right] + \avgmj\left[\widehat{R}_{\epsilon,j}^2 - R_{\epsilon,j}^2\right] +\avgmj\left[R_{\epsilon,j}^2-R_{j}^2\right] + \left[\avgmj R_{j}^2- \E[R^2]\right]\\
=:&A_1+A_2+A_3+A_4,
\end{align*}
where $A_1,A_2,A_3,$ and $A_4$ are defined as the respective terms in the second last line.
In the following, we show that $A_i\conprob 0$, $i=1,2,3,4$, one by one.
\begin{enumerate}[label=(\Roman*)]
	\item Proving $A_1\conprob 0$. Note that
	\begin{align*}
		\E\left[|A_1|\right] \stackrel{\eqref{eq:CauthySchwarzIneq2}}{\leq}&\E\left[\left|\widehat{R}_{\epsilon, m,1}^2 -\widehat{R}_{\epsilon,1}^2\right|\right]\\
		\stackrel{\eqref{eq:CauthySchwarzIneq4}}{\leq}& \E\left[\left(\widehat{R}_{\epsilon, m,1} -\widehat{R}_{\epsilon,1}\right)^2\right] + 2 \left(\E\left[\widehat{R}_{\epsilon,1}^2\right]\right)^{1/2}\left(\E\left[\left(\widehat{R}_{\epsilon, m,1} -\widehat{R}_{\epsilon,1}\right)^2\right]\right)^{1/2}
	\end{align*}
	Note that by construction $\geps'(x)$ is continuous and $|\geps''(x)|\leq\frac{1}{4\pi\epsilon^2}$.
	Therefore,
	\begin{align*}
		&\E\left[\left(\widehat{R}_{\epsilon, m} - \widehat{R}_{\epsilon}\right)^2\right] \stackrel{\eqref{eq:CauthySchwarzIneq2}}{\leq} \E\left[\left| \left[\geps'(L_m)- \geps'(L) \right] \hatH \right|^2\right] =\E\left[\left|\geps''(\Lambda_m)(L_m-L) \hatH \right|^2\right]\\
		 \leq& \frac{1}{16\pi\epsilon^4}\E\left[\left|\left(L_m - L\right) \hatH \right|^2\right] \stackrel{\eqref{eq:CauthySchwarzIneq3}}{\leq} \frac{1}{16\pi\epsilon^4}\left(\E\left[\left(L_m - L\right)^4\right]\right)^{1/2} \left(\E\left[\hatH^4\right]\right)^{1/2}\\
		  \stackrel{(*)}{=}& \frac{1}{16\pi\epsilon^4}\left(\cO\left(\frac{1}{m^2}\right)\right)^{1/2} \left(\E\left[\hatH^4\right]\right)^{1/2} = \cO\left(\frac{1}{m\epsilon^4}\right),
	\end{align*}
	where $(*)$ holds by Theorem~\ref{thm:Lm_moment2p} with $p=2$ and $\E[\hatH^4]<\infty$ by assumption.
	Also, we can show that $\E\left[\widehat{R}_{\epsilon,1}^2\right]=\cO(\epsilon^{-1})$ by a similar analysis as~\eqref{ind CLT Lambda}.
	As a result, $\E[|A_1|]= \cO\left((m\epsilon^5)^{-1/2}\right)\to 0$ because $m\epsilon^5\to\infty$ by assumption.
	This means that $A_1\conlone 0$ and so $A_1 \conprob 0$.
	
	\item Proving $A_2\conprob 0$. Note that
	\begin{align*}
		\E[|A_2|] \stackrel{\eqref{eq:CauthySchwarzIneq2}}{\leq}& \E\left[\left|\widehat{R}_{\epsilon,1}^2 - R_{\epsilon,1}^2\right|\right]\nonumber\\
		\stackrel{\eqref{eq:CauthySchwarzIneq4}}{\leq}& \E\left[\left(\widehat{R}_{\epsilon,1} -R_{\epsilon,1}\right)^2\right] + 2 \left(\E\left[R_{\epsilon,1}^2\right]\right)^{1/2}\left(\E\left[\left(\widehat{R}_{\epsilon,1} -R_{\epsilon,1}\right)^2\right]\right)^{1/2} .
	\end{align*}
	As $\geps(x)$ is a smooth function, we can use the same arguments for~\eqref{eq:aux8} to deduce that $\E\left[\left(\widehat{R}_{\epsilon} - R_{\epsilon}\right)^2\right]\leq \frac{1}{n}\cdot\E\left[\left(\geps'(L) \hatH \right)^2\right] \stackrel{(*)}{=} \cO\left((n\epsilon)^{-1}\right)$,
	where $(*)$ holds because we can show $\E\left[\left(\geps'(L) \hatH \right)^2\right]=\cO\left(\epsilon^{-1}\right)$ by a similar analysis as~\eqref{ind CLT Lambda}.
	So $\E\left[|A_2|\right] = \cO\left((n\epsilon^2)^{-1/2}\right) \to 0$ because $n\epsilon^2\to \infty$ by assumption.
	This means that $A_2\conlone 0$ and so $A_2 \conprob 0$.
	
	\item Proving $A_3\conprob 0$. Note that
		\begin{align*}
		&\E[|A_3|] \stackrel{\eqref{eq:CauthySchwarzIneq2}}{\leq} \E\left[\left|R_{\epsilon,1}^2 - R_{1}^2\right|\right]\stackrel{\eqref{eq:CauthySchwarzIneq4}}{\leq} \E\left[\left(R_{\epsilon,1} -R_{1}\right)^2\right] + 2 \left(\E\left[R_{1}^2\right]\right)^{1/2}\left(\E\left[\left(R_{\epsilon,1} -R_{1}\right)^2\right]\right)^{1/2} .\label{IndCI sig23}
	\end{align*}
By Assumption~\ref{assm:indicatordensity}~\ref{assum:boundedmoments}, $\E\left[R^2\right] = \E\left[\left(\int \hatH(x,Y)\psi_0(x)\d x\right)^2\right]<\infty$.
Moreover,
\begin{align*}
	R_{\epsilon,1} &= \E[g'_{\epsilon}(L)\hatH_{11}|Y_1] = \int \frac{1}{\epsilon}\phi\left(\frac{\ell}{\epsilon}\right) \hatH(x,Y_1)\psi(x,\ell) \d x \d \ell=\int \phi\left(u\right) \hatH(x,Y_1)\psi(x,\epsilon u) \d x \d u\\
	&= \int \phi\left(u\right) \hatH(x,Y_1)\left[\psi_0(x) + \epsilon u \frac{\partial}{\partial \ell} \psi(x,\bar{u})\right]\d x \d u\\
	&= R_1 + \int \phi\left(u\right) \hatH(x,Y_1)\left[\epsilon u \frac{\partial}{\partial \ell} \psi(x,\bar{u})\right]\d x \d u
\end{align*}
Also, because $|\frac{\partial}{\partial \ell} \psi(x,\ell)| \leq \psi_1(x)$ by Assumption~\ref{assm:indicatordensity}~\ref{assum:existence}, we have
\begin{align*}
	\E\left[(R_{\epsilon,1}-R_1)^2\right] &\leq \E\left[\left(\int \phi\left(u\right) \hatH(x,Y_1)\left[\epsilon |u| \frac{\partial}{\partial \ell} \psi(x,\bar{u})\right]\d x \d u\right)^2\right]\\
	&= \epsilon^2 \left(\int \phi(u)|u|\d u\right)^2 \E\left[ \left(\int \hatH(x,Y_1) \frac{\partial}{\partial \ell} \psi(x,\bar{u})\d x\right)^2\right] = \cO(\epsilon^2)
\end{align*}
Therefore $\E[|A_3|]\to 0$ as $\epsilon\to 0$. This means that $A_3\conlone 0$ and so $A_3 \conprob 0$.
	
	\item Proving $A_4\conprob 0$. Since $R_j^2$, $j=1,\ldots,m$, are i.i.d. samples with the common expectation $\E\left[R^2\right]$, so $A_4\conprob 0$ by the weak law of large numbers.
\end{enumerate}




The proof is complete.
\end{proof}

\end{document}